%% file: main.tex
\documentclass[a4paper,UKenglish,11pt]{article}
\usepackage{amssymb,amsthm,amsmath}
\usepackage{xcolor}
\usepackage{enumitem}
\usepackage{thmtools,thm-restate}
\usepackage{todonotes}
\usepackage[noend,boxed,linesnumbered]{algorithm2e}
\usepackage{fullpage}
\usepackage{authblk}
\usepackage{cases}
\usepackage[draft]{fixme}
\usepackage[colorlinks=true,linkcolor=blue,citecolor=red]{hyperref}
\usepackage[capitalize]{cleveref}
\usepackage{tikz}
\usepackage{wrapfig}

\setitemize{itemsep=0pt,topsep=2pt,parsep=0pt,partopsep=0pt}
\setdescription{itemsep=0pt,topsep=2pt,parsep=0pt,partopsep=0pt}
\setenumerate{itemsep=0pt,topsep=2pt,parsep=0pt,partopsep=0pt}

\fxsetup{mode=multiuser,theme=color, layout=inline}
\FXRegisterAuthor{tk}{atk}{\color{violet} \bf Tomasz}

\newtheorem{theorem}{Theorem}[section]
\newtheorem{lemma}[theorem]{Lemma}
\newtheorem{corollary}[theorem]{Corollary}
\newtheorem{definition}[theorem]{Definition}
\newtheorem{claim}[theorem]{Claim}
\newtheorem{proposition}[theorem]{Proposition}

\newtheorem{observation}[theorem]{Observation}
\newtheorem{fact}[theorem]{Fact}
\crefname{fact}{Fact}{Facts}

\newtheorem{theorem-restatable}[theorem]{Theorem}
\newtheorem{corollary-restatable}[theorem]{Corollary}

\theoremstyle{remark}
\newtheorem{example}[theorem]{Example}
\newtheorem{remark}[theorem]{Remark}
\newtheorem{problem}[theorem]{Problem}{\bfseries}{\itshape}

\newcommand{\sub}{\subseteq}
\newcommand{\sm}{\setminus}

\newcommand{\Oh}{\mathcal{O}}
\newcommand{\tOh}{\tilde{\mathcal{O}}}

\newcommand{\ed}{\mathsf{ed}}
\newcommand{\hd}{\mathsf{hd}}
\newcommand{\per}{\mathsf{per}}
\newcommand{\Occ}{\mathsf{OCC}}
\newcommand{\aOcc}{\mathsf{aOCC}}
\newcommand{\eps}{\varepsilon}
\newcommand{\brkp}{\mathcal{B}}
\newcommand{\poly}{\mathrm{poly}}

\newcommand{\G}{\mathcal{G}}
\newcommand{\chain}{\mathcal{C}}
\newcommand{\rot}{\mathsf{rot}}

\newcommand{\gr}{\mathsf{GR}}
\newcommand{\qgr}{\mathsf{qGR}}
\newcommand{\ga}{\mathsf{GA}}
\newcommand{\qga}{\mathsf{qGA}}
\newcommand{\mtch}{\mathcal{M}}
\newcommand{\qmtch}{\mathsf{q}\mathcal{M}}
\newcommand{\cost}{\mathsf{cost}}
\newcommand{\width}{\mathsf{width}}
\newcommand{\LZ}{\mathsf{LZ}}
\newcommand{\A}{\mathcal{A}}
\newcommand{\rhs}{\mathsf{rhs}}

\newcommand{\LCE}{\mathsf{LCE}}

\newcommand{\CGK}{\mathsf{CGK}}
\newcommand{\bCGK}{\mathsf{CE}}
\newcommand{\greedyalg}{\mathsf{Greedy}}

\newcommand{\dd}{\mathinner{.\,.\allowbreak}}
\newcommand{\Exp}{\mathbb{E}}

\newcommand{\rank}{\mathsf{rank}}
\newcommand{\select}{\mathsf{select}}
\newcommand{\num}{\mathsf{num}}
\newcommand{\rev}[1]{\overline{#1}}
\newcommand{\Zz}{\mathbb{Z}_{\ge 0}}
\newcommand{\Zp}{\mathbb{Z}_{+}}
\newcommand{\CR}{\mathsf{D}}
\newcommand{\skE}{\mathsf{sk}^E}
\newcommand{\skH}{\mathsf{sk}^H}
\newcommand{\MI}{\mathsf{MI}}
\newcommand{\PMI}{\mathsf{PMI}}
\newcommand{\MP}{\mathsf{MP}}
\newcommand{\skL}{\mathsf{sk}^{P}}
\newcommand{\skq}{\mathsf{sk}^{\mathsf{q}}}
\newcommand{\Alg}{\mathsf{W}}
\newcommand{\maxLZ}{\rev{\mathsf{maxLZ}}}
\newcommand{\Enc}{\mathsf{E}}

\title{Small space and streaming pattern matching with \texorpdfstring{$k$}{k} edits}

\author[1]{Tomasz Kociumaka\thanks{Partly supported by NSF 1652303, 1909046, and HDR TRIPODS 1934846 grants, and an Alfred P. Sloan Fellowship.}}
\author[2]{Ely Porat}
\author[3]{Tatiana Starikovskaya\thanks{Partly supported by the grant ANR-20-CE48-0001 from the French National Research Agency (ANR).}}

\affil[1]{University of California, Berkeley, USA}
\affil[ ]{\texttt{kociumaka@berkeley.edu}}
\affil[2]{Bar-Ilan University, Israel}
\affil[ ]{\texttt{porately@cs.biu.ac.il}}
\affil[3]{DI/ENS, PSL Research University, France}
\affil[ ]{\texttt{tat.starikovskaya@gmail.com}}

\date{\vspace{-.5cm}}

\begin{document}
\maketitle

\begin{abstract}
In this work, we revisit the fundamental and well-studied problem of approximate pattern matching under edit distance. Given an integer $k$, a pattern $P$ of length $m$, and a text $T$ of length $n \ge m$, the task is to find substrings of $T$ that are within edit distance $k$ from $P$. 
Our main result is a streaming algorithm that solves the problem in $\tOh(k^5)$ space and $\tOh(k^8)$ amortised time per character of the text, providing answers correct with high probability. (Hereafter, $\tOh(\cdot)$ hides a $\poly(\log n)$ factor.) This answers a decade-old question:
since the discovery of a $\poly(k\log n)$-space streaming algorithm for pattern matching under Hamming distance by Porat and Porat [FOCS 2009], the existence of an analogous result for edit distance remained open.
Up to this work, no $\poly(k\log n)$-space algorithm was known even in the simpler semi-streaming model, where $T$ comes as a stream but $P$ is available for read-only access. In this model, we give a deterministic algorithm that achieves slightly better complexity.

Our central technical contribution is a new space-efficient deterministic encoding of two strings, called the greedy encoding, which encodes a set of all alignments of cost $\le k$ with a certain property (we call such alignments \emph{greedy}). On strings of length at most $n$, the encoding occupies $\tOh(k^2)$ space. We use the encoding to compress substrings of the text that are close to the pattern. In order to do so, we compute the encoding for substrings of the text and of the pattern, which requires read-only access to the latter.

In order to develop the fully streaming algorithm, we further introduce a new edit distance sketch parametrised by integers $n\ge k$. For any string of length at most $n$, the sketch is of size $\tOh(k^2)$ and it can be computed with an $\tOh(k^2)$-space streaming algorithm.  Given the sketches of two strings, in $\tOh(k^3)$ time we can compute their edit distance or certify that it is larger than~$k$.
This result improves upon $\tOh(k^8)$-size sketches of Belazzougui and Zhu [FOCS 2016]
and very recent $\tOh(k^3)$-size sketches of Jin, Nelson, and Wu [STACS 2021].
\end{abstract}

\section{Introduction}
\input{intro}

\section{Preliminaries}\label{sec:prelim}
\input{prelim}

\section{Technical Overview}\label{sec:overview}
\input{overview}

\section{Compressed String Representation}\label{sec:compr}
\input{compression}

\section{Greedy Alignments and Encodings}\label{sec:greedy}
\input{greedy}

\section{Edit Distance Sketches}\label{sec:sketches}
\input{sketches}

\section{Pattern Matching with \texorpdfstring{\boldmath$k$}{k} Edits}\label{sec:algorithms}
\input{rpst}

\input{spst}

\bibliography{main}
\end{document}

%% file: intro.tex
In the pattern matching problem, given two strings, a \emph{pattern} $P$ of length $m$ and a \emph{text} $T$ of length $n$, one must find all substrings of the text equal to the pattern. 
This is a fundamental problem of string processing with a myriad of applications in such fields as computational biology, information retrieval, and signal processing, to mention just a few. However, in many applications, retrieving substrings that are exactly equal to the pattern is not enough, and one must search for substrings merely \emph{similar} to the pattern. 
This task, which is often referred to as \emph{approximate pattern matching problem}, can be formalised in the following way: for each position $i$ in the text, compute the smallest distance $d_i$ between $P$ and any substring of $T$ that ends at position $i$.
In string processing, the two most popular distances are the Hamming distance and the edit distance. Recall that the Hamming distance between two equal-length strings is the number of mismatching pairs of characters of the strings. The edit distance of two strings, not necessarily of equal lengths, is the smallest number of edits (character insertions, deletions, and substitutions) needed to transform one string into the other.
Due to its practical importance, the  approximate pattern matching problem has been extensively studied in the literature, originally in the classic setting with input strings explicitly stored in memory.

In general, computing Hamming distance is easier and can be considered as a preliminary step towards tackling edit distance. The first solution for approximate pattern matching under the Hamming distance was given by Abrahamson~\cite{Abrahamson} and, independently, Kosaraju~\cite{Kosaraju}; based on the fast Fourier transform, it spends $\Oh(n\sqrt{m \log m})$ time to compute the Hamming distance between the pattern and all the length-$m$ substrings of the text. Up to date, no algorithms improve upon this time complexity for the general version of approximate pattern matching under Hamming distance, but there are better solutions when one is interested only in the distances not exceeding a given $k$, a variant known as the \emph{$k$-mismatch problem}. The first algorithm for the $k$-mismatch problem was given by Landau and Vishkin~\cite{LandauV86}, who improved the running time to $\Oh(kn)$ via so-called ``kangaroo jumps'', a technique utilising the suffix tree to compute the longest common prefix of two suffixes of a string in constant time. This bound was further improved by Amir et al.~\cite{AmirLP04} who showed two solutions, one with running time $\Oh(n\sqrt{k \log k})$ and one with running time $\tOh(n+k^3 n/m)$-time algorithm\footnote{Hereafter, $\tOh(\cdot)$ hides a factor of $\poly(\log n)$.}.
Continuing this line of research, Clifford et al.~\cite{DBLP:conf/soda/CliffordFPSS16} presented an $\tOh(n+k^2 n/m)$-time algorithm, while Gawrychowski and Uzna\'{n}ski~\cite{GawrychowskiU18} demonstrated a smooth trade-off between the latter and the solution of Amir et al. by designing an $\tOh(n+ kn/\!\sqrt{m})$-time algorithm. Very recently, Chan et al.~\cite{DBLP:conf/stoc/ChanGKKP20} shaved off most of the polylogarithmic factors and achieved the running time of  $\Oh(n+\min(k^2 n/m, kn\sqrt{\log m}/\sqrt{m}))$ at the cost of Monte-Carlo randomization. 

For the edit distance, a detailed survey of previous solutions can be found in~\cite{Navarro:2001}, and here we only discuss the landmark results of the theoretical landscape. For the general variant of the problem, the first algorithm was given by Sellers~\cite{SELLERS1980359}. The algorithm was based on dynamic programming and used $\Oh(nm)$ time. Masek and Paterson~\cite{MASEK198018} improved the running time of the algorithm to $\Oh(nm/ \log n)$ via the Four Russians technique. On the lower bound side, it is known that there is no solution with strongly subquadratic time complexity unless the Strong Exponential Time hypothesis~\cite{IP01} is false, even for the binary alphabet~\cite{10.1145/2746539.2746612,DBLP:conf/focs/BringmannK15}. 
Abboud et al.~\cite{10.1145/2897518.2897653} gave a more precise bound under a weaker assumption: Namely, they showed that even  shaving an arbitrarily large polylog factor would imply that NEXP does not have non-uniform $\mathrm{NC}^1$ circuits. Finally, Clifford et al.~\cite{DBLP:conf/soda/CliffordJS15} showed that, in the cell-probe model with the word size $w = 1$, any randomised algorithm that computes the edit distances between the pattern and the text online must spend $\Omega(\frac{\sqrt{\log n}}{(\log \log n)^{3/2}})$ expected amortised time per character of the text. 

Similarly to the Hamming distance, one can define the threshold variant of approximate pattern matching, which we refer to as approximate pattern matching with $k$ edits. The first algorithm for this variant of the problem was developed by Landau and Vishkin~\cite{DBLP:journals/jal/LandauV89}; this by-now classical algorithm solves the problem in $\Oh(nk)$ time. The current best result was achieved by a series of work~\cite{SahinalpV96,DBLP:journals/siamcomp/ColeH02} with the running time (for some range of parameters) $\Oh(n+k^4 n/m)$. Very recently, Charalampopoulos et al.~\cite{DBLP:journals/corr/abs-2004-08350} studied the problem for both distances in the grammar-compressed setting. Their result, in particular, implies an $\Oh(k^4)$-space and $\Oh(nk^4)$-time algorithm for the read-only model, where random access to characters in $P$ and $T$ is allowed and one accounts only for the ``extra'' space, the space required beyond the space needed to store the pattern and the text.

In this work, we focus on developing algorithms for approximate pattern matching with $k$ edits that use as little space as possible. In particular, we consider the streaming model of computation. In this model, we assume that the input arrives as a stream, one character at a time. We define the space complexity of the algorithm to be all the space used, in other words, we cannot store any information about the input without accounting for~it. 

The field of streaming algorithms for string processing is relatively recent but, because of its practical interest, it has received a lot of attention in the literature. It started with a seminal paper of  Porat and Porat in FOCS 2009~\cite{Porat:09}, who showed streaming algorithms for exact pattern matching and for the $k$-mismatches problem. The result of Porat and Porat was followed by a series of works on streaming pattern matching~\cite{DBLP:journals/talg/BreslauerG14,DBLP:conf/esa/CliffordFPSS15,DBLP:conf/esa/GolanP17, DBLP:conf/soda/CliffordFPSS16, starikovskaya:LIPIcs:2017:7320, DBLP:conf/icalp/GolanKP18, gawrychowski_et_al:LIPIcs:2019:10492,clifford2018streaming,DBLP:conf/cpm/GolanKKP20, DBLP:journals/iandc/RadoszewskiS20}, search of repetitions in streams~\cite{Ergun:10,stream-periodicity-mismatches, stream-periodicity-wildcards,DBLP:journals/algorithmica/GawrychowskiMSU19, DBLP:conf/cpm/MerkurevS19, DBLP:conf/spire/MerkurevS19, DBLP:conf/cpm/GawrychowskiRS19}, and recognising formal languages in streams~\cite{DBLP:journals/siamcomp/MagniezMN14,ganardi_et_al:LIPIcs:2018:9131, DBLP:conf/lata/GanardiHL18, ganardi_et_al:LIPIcs:2018:8485, ganardi_et_al:LIPIcs:2016:6853, DBLP:conf/mfcs/GanardiJL18, DBLP:journals/tcs/BabuLRV13, franois_et_al:LIPIcs:2016:6355,ganardi_et_al:LIPIcs:2019:11502}. 

All known streaming algorithms for approximate pattern matching under the Hamming distance~\cite{Porat:09, DBLP:conf/soda/CliffordFPSS16,clifford2018streaming} are based on some rolling hash~--- a hash on strings of fixed length that can be efficiently updated when we delete the first character of a string and append a new one to the end of the string, such as, for example, the famous Karp--Rabin fingerprint, which allows computing the Hamming distance between two strings or certify that it exceeds $k$. The current best algorithm was demonstrated by Clifford et al. in SODA 2019~\cite{clifford2018streaming}. The algorithm uses only $\Oh(k \log \frac{m}{k})$ space, which is optimal up to a logarithmic factor, and spends $\Oh(\log \frac{m}{k} (\sqrt{k \log k} + \log^3 m))$ time per character. The algorithm is necessarily randomised and can err with high probability\footnote{With high probability means with probability at least $1-1/n^c$ for any predefined constant $c>1$.}. In other words, approximate pattern matching under the Hamming distance in the streaming model is essentially fully understood.

On the other hand, for the edit distance there are no small-space solutions, in particular, because there are no rolling hashes that allow to compute the edit distance between strings. When $k$ is small, the state-of-the-art solution~\cite{starikovskaya:LIPIcs:2017:7320} uses $\Oh ( k^8 \sqrt{m} \log^6 m)$ space and $\Oh((k^2 \sqrt{m} + k^{13}) \cdot \log^4 m)$ worst-case time per symbol. Again, the algorithm is randomised and outputs all substrings at edit distance at most $k$ from the pattern with high probability. Another interesting result is that of Chakraborty et al.~\cite{DBLP:conf/fsttcs/Chakraborty0K19}, who developed an algorithm for the general version of approximate pattern matching under edit distance in the model where the text is streaming, but the pattern is read-only. They showed a randomised algorithm that, for every position $i$ of the text $T$, computes the smallest edit distance $d_i$ between $P$ and a suffix of $T[1\dd i]$ with constant multiplicative and $m^{8/9}$-additive approximation (in other words, the algorithm returns a number between $d_i$ and $c \cdot d_i + m^{8/9}$, where $c \ge 1$ is a predetermined constant). The algorithm receives the text online and uses $\Oh(m^{1-1/54})$ extra space, in addition to the space required to store the pattern.  

Naturally, a question arises: is the true complexity of streaming approximate pattern matching under edit distance is on par with that of the Hamming distance? In this work, we answer this question affirmatively.

\subsection{Our results}
The main result of our work is a fully streaming algorithm for approximate pattern matching under the edit distance that uses $\tOh(k^5)$ space and $\tOh(k^8)$ amortized time per character of the text (\cref{th:spst}). The algorithm is randomised and its answers are correct with high probability.

As a stepping stone, we also consider a simpler semi-streaming model introduced in~\cite{DBLP:conf/fsttcs/Chakraborty0K19}. In this model, we assume that the text arrives in a stream, but the pattern is read-only, which means that, at any moment, the algorithm can access any character of the pattern in constant time (but re-writing characters is prohibited). The space complexity of the algorithm is defined as the total space used on top of the read-only memory holding the pattern. In this setting, we show a \emph{deterministic} algorithm for approximate pattern matching under the edit distance that uses $\tOh(k^5)$ space and $\tOh(k^6)$ amortized time per character of the text (\cref{th:rpst}). 

Additionally, we design a new sketch for retrieving the exact edit distance (capped with a threshold $k$) between strings of length at most $n$  (\cref{thm:ske}).
The sketch is of size $\tOh(k^2)$, and it can be built using a streaming algorithm that costs $\tOh(nk)$ total time and uses $\tOh(k^2)$ space. Given the sketches of two strings $X,Y$, in $\tOh(k^3)$ time and $\tOh(k^2)$ space, we can compute the edit distance between $X,Y$ or certify that it is larger than~$k$. The answer is correct with a large constant probability (with standard amplification, we then achieve high probability of success). This improves upon the $\tOh(k^8)$-size sketches of Belazzougui and Zhang~\cite{sketches} and the $\tOh(k^3)$-size sketches of Jin, Nelson, and Wu~\cite{jin2020improved} developed independently.

The conceptual contribution of our work is described in the technical overview (\cref{sec:overview}).

%% file: prelim.tex
We assume an integer alphabet $\Sigma = \{1, 2, \ldots, \sigma\}$ with $\sigma$ \emph{characters}. A \emph{string} $Y$ is a sequence of characters numbered from $1$ to $n = |Y|$. By $Y[i]$ we denote the $i$-th symbol of $Y$. For a string $Y$ of length $n$, we denote its \emph{reverse} $Y[n] Y[n-1] \ldots Y[1]$ by $\rev{Y}$. We define $Y[i \dd j]$ to be equal to $Y[i] \dots Y[j]$ which we call a \emph{fragment} of $Y$ if $i \le j$ and to the empty string $\eps$ otherwise. We also use notations $Y[i \dd j)$ and $Y(i\dd j]$ which naturally stand for $Y[i] \dots Y[j-1]$ and $Y[i+1] \dots Y[j]$, respectively. 
We call a fragment $Y[1] \dots Y[j]$ \emph{a prefix} of $Y$ and use a simplified notation $Y[\dd i]$, and a fragment $Y[i] \dots Y[n]$ \emph{a suffix} of $Y$ denoted by $Y[i \dd]$. We say that $X$ is a \emph{substring} of $Y$ if $X = Y[i \dd j]$ for some $1 \le i \le j \le |Y|$. The fragment $Y[i \dd j]$ is called an \emph{occurrence} of $X$. 

For a string $Y$, we define $Y^m$ to be the concatenation of $m$ copies of $Y$. We also define $Y^\infty$ to be an infinite string obtained by concatenating infinitely many of copies of $Y$.
We say that a string $X$ of length $x$ is a \emph{period} of a string~$T$ if $X = T[1\dd x]$ and $T[i]=T[i+x]$ for all $i=1,\ldots,|T|-x$. By $\per(T)$ we denote the length of the shortest period of $T$. The string $T$ is called \emph{periodic} if $2 \, \per(T) \le |T|$. 
For a string $Y\in \Sigma^n$, we define a \emph{forward rotation} $\rot(Y)=Y[2] \cdots Y[n]Y[1]$.
In general, a \emph{cyclic rotation} $\rot^s(Y)$ with \emph{shift} $s\in \mathbb{Z}$ is obtained by iterating $\rot$ or the inverse operation $\rot^{-1}$.
A non-empty string $X\in \Sigma^n$ is \emph{primitive} if it is distinct from its non-trivial rotations, i.e., if $X=\rot^s(X)$ holds only when $s$ is a multiple of $n$.

We say that a fragment $X[i\dd i+\ell)$ is a \emph{previous factor} if 
$X[i\dd i+\ell)=X[i'\dd i'+\ell)$ holds for some $i'\in [1\dd i)$.
The \emph{LZ77 factorization} of $X$ is a factorization $X=F_1\cdots F_z$ into non-empty \emph{phrases}
such tht the $j$th phrase $F_j$ is the longest previous factor starting at position $1+|F_1\cdots F_{j-1}|$; if no previous factor starts there, then $F_j$ consists of a single character.
In the underlying \emph{LZ77 representation}, every phrase $F_j=T[j\dd j+\ell)$ that is a previous fragment is encoded as $(i',\ell)$, where $i'\in [1\dd i)$ satisfies $X[i\dd i+\ell)=X[i'\dd i'+\ell)$.
The remaining length-1 phrases are represented by the underlying character.
We use $\LZ(X)$ to denote the underlying \emph{LZ77 representation} 
and $|\LZ(X)|$ to denote its size (the number of phrases).

\subsection{Edit Distance Alignments}
The \emph{edit distance} $\ed(X, Y)$ between two strings $X$ and $Y$ is defined as the smallest number of character insertions, deletions, and substitutions required to transform $X$ to $Y$. The \emph{Hamming distance} $\hd(X,Y)$ allows substitutions only (and we assume $\hd(X,Y)=\infty$ if $|X|\ne |Y|$).

\begin{definition}
  A sequence $(x_t,y_t)_{t=1}^m$ is an \emph{alignment}
  of $X,Y\in \Sigma^*$ if $(x_1,y_1)=(1,1)$,
  $(x_m,y_m)=(|X|+1,|Y|+1)$, and $(x_{t+1},y_{t+1})\in \{(x_t+1,y_t+1),(x_t+1,y_t),(x_t,y_t+1)\}$ for $t\in [1\dd m)$.\footnote{This definition is rather complex, but it is equivalent to the standard definition given in the textbooks. We chose this particular formulation as it allowed us to introduce notions essential for this work in a rigorous way.}
\end{definition}
Given an alignment $\A = (x_t,y_t)_{t=1}^m$ of strings $X,Y\in \Sigma^*$,
for every $t\in [1\dd m)$:
\begin{itemize}
  \item If $(x_{t+1},y_{t+1})=(x_t+1,y_t)$, we say that $\A$ \emph{deletes} $X[x_t]$,
  \item If $(x_{t+1},y_{t+1})=(x_t,y_t+1)$, we say that $\A$ \emph{deletes} $Y[y_t]$,
  \item If $(x_{t+1},y_{t+1})=(x_t+1,y_t+1)$, we say that $\A$ \emph{aligns} $X[x_t]$ and $Y[y_t]$, denoted $X[x_t] \sim_\A Y[y_t]$. If~additionally $X[x_t]= Y[y_t]$, we say that $\A$ \emph{matches} $X[x_t]$ and $Y[y_t]$, denoted $X[x_t] \simeq_\A Y[y_t]$.
  Otherwise, we say that $\A$ \emph{substitutes} $X[x_t]$ for $Y[y_t]$.
\end{itemize}

The \emph{cost} of an edit distance alignment $\A$ is the total number characters that $\A$ deletes or substitutes.
We denote the cost by $\cost_{X,Y}(\A)$, omitting the subscript if $X,Y$ are clear from context.
The cost of an alignment $\A=(x_t,y_t)_{t=1}^m$ is at least its width $\width(\A)=\max_{t=1}^m |x_t-y_t|$.
Observe that $\ed(X,Y)$ can be defined as the minimum cost of an alignment of $X$ and~$Y$.
An alignment of $X$ and $Y$ is \emph{optimal} if its cost is equal to $\ed(X, Y)$.

Given an alignment $\A=(x_t,y_t)_{t=1}^m$ of $X,Y\in \Sigma^+$, we partition the elements $(x_t,y_t)$ of $\A$
into \emph{matches} (for which $X[x_t]\simeq_{\A} Y[y_t]$) and \emph{breakpoints} (the remaining elements).
We denote the set of matches and breakpoints by $\mtch_{X,Y}(\A)$ and $\brkp_{X,Y}(\A)$, respectively,
omitting the subscripts if the strings $X,Y$ are clear from context. Observe that $|\brkp_{X,Y}(\A)|=1+\cost(\A)$.

We call $M\sub [1\dd |X|]\times [1\dd |Y|]$ a \emph{non-crossing matching} of $X,Y\in \Sigma^*$
if $X[x]=Y[y]$ holds for all $(x,y)\in M$ and there are no distinct pairs $(x,y),(x',y')\in M$ with $x\le x'$ and $y\ge y'$.
Note that, for every alignment $\A$ of $X,Y$, the set $\mtch(\A)$ is a non-crossing matching of~$X,Y$.

\begin{example}\label{ex:edit}
Consider strings $X=\mathtt{abbaabcb}$ and $Y=\mathtt{acabaabab}$ and a cost-$4$ alignment
\[\A: (1,1), (2,2), (3,3), (3,4), (4,5), (5,6), (6,7), (7,8), (8,8), (8,9), (9,10).\] 
\end{example}

\begin{wrapfigure}{r}{0.325\textwidth}
  \vspace{-.5cm}
  \begin{tikzpicture}[xscale=0.5]
    \draw (1.8,0.7) node[above] {$X$:};
    \foreach \x/\c in {1/a,2/{\color{red} b},3/b,4/a,5/a,6/b,7/c,8/b}{
      \draw (\x+2,0.7) node[above] {\tt \c};
      \draw (\x+2,1) node[above] {\tiny{\textcolor{blue}{\x}}};
    } 
    
    \draw (1.8,-0.5) node[above] {$Y$:};
    \foreach \x/\c in {1/a,2/{\color{red} c},3/a,4/b,5/a,6/a,7/b,8/a,9/b}{
      \draw (\x+2,-0.5) node[above] {\tt \c};
      \draw (\x+2,-0.2) node[above] {\tiny{\color{blue} \x}};
    }
    
    \draw (3,0.6)--(3,0.2);
    \draw[red,densely dashed] (4,0.6)--(4,0.2);
    \draw (5,0.6)--(6,0.2);
    \draw (6,0.6)--(7,0.2);
    \draw (7,0.6)--(8,0.2);
    \draw (8,0.6)--(9,0.2);
    \draw (10,0.6)--(11,0.2);
  \end{tikzpicture}
  \end{wrapfigure}

\noindent
The breakpoints are $\brkp(\A)=\{(2,2), (3,3), (7,8), (8,8), (9,10)\}$;
the first $4$ breakpoints correspond to a substitution of $X[2]$ for $Y[2]$, a deletion of $Y[3]$, a deletion of $X[7]$, and a deletion of $Y[8]$, respectively. Graphically, the alignment is depicted on the right; the aligned pairs of characters are connected with an edge, and the substituted pair is highlighted.

\smallskip
Given an alignment $\A= (x_t,y_t)_{t=1}^m$ of $X$ and $Y$, for every $\ell,r\in [1\dd m]$ with $\ell\le r$,
we say that $\A$ \emph{aligns} $X[x_\ell\dd x_{r})$ and $Y[y_\ell\dd y_{r})$, denoted $X[x_\ell\dd x_{r})\sim_\A Y[y_{\ell}\dd y_{r})$. If there is no breakpoint $(x_t,y_t)$ with $t\in [\ell\dd r)$,
we further say that $\A$ \emph{matches} $X[x_\ell\dd x_{r})$ and $Y[y_\ell\dd y_{r})$, denoted $X[x_\ell\dd x_{r})\simeq_\A Y[y_{\ell}\dd y_{r})$. 

An alignment $\A = (x_t,y_t)_{t=1}^m$ of $X,Y\in \Sigma^*$ naturally induces a unique alignment of any two fragments $X[x\dd x')$ and $Y[y\dd y')$.
Formally, the \emph{induced alignment} $\A_{[x\dd x'),[y\dd y')}$ is obtained  by removing repeated entries from $(\max(x,\min(x',x_t))-x+1,\max(y,\min(y',y_t))-y+1)_{t=1}^m$.

\begin{fact}\label{fact:edk2}
If an alignment $\A$ satisfies $X[x \dd x') \sim_{\A} Y[y \dd y')$, then $|x-x'|,|y-y'|\le  \width(\A)$
and $|(x'-x)-(y'-y)|\le \ed(X[x \dd x'), Y[y \dd y'))\leq \cost(\A_{[x\dd x'),[y\dd y')}) \leq \cost(\A)$.
\end{fact}




%% file: overview.tex
In this section, we provide an overview of our conceptual and technical contribution. Let us start with a formal statement of the pattern matching with $k$ edits problem.
We say that $T[\ell \dd r]$ is a \emph{$k$-edit occurrence} of $P$ if $\ed(P,T[\ell \dd r]) \le k$, and we denote the set of the \emph{right} endpoints of the $k$-edit occurrences of $P$ in $T$ by $\Occ^E_k(P,T)$.

\begin{problem}[Pattern matching with $k$ edits]
Given a pattern $P$ of length $m$ over an alphabet $\Sigma$, a text $T$ of length $n$ over $\Sigma$, and an integer $k$, compute $\Occ_k^E(P,T)$.
\end{problem}

We solve an \emph{online} version of the problem, where the text arrives in a stream (character by character)
and the algorithm must decide whether $r\in \Occ^E_k(P, T)$ while processing $T[r]$.
The pattern is preprocessed in advance (consistently with~\cite{Porat:09}, in the current version of this paper we do not account for this prepro\-cessing in the complexity analysis). We consider two settings: 
\begin{enumerate}
\item In the \emph{streaming} setting, the algorithm can no longer access $P$ or $T[1\dd r)$ while processing $T[r]$.
In other words, all the information regarding these strings needs to be stored explicitly and accounted for in the space complexity of the algorithm.
\item In the \emph{semi-streaming} setting, the algorithm can no longer access $T[1\dd r)$ while processing $T[r]$,
but it is given an oracle providing read-only constant-time access to individual characters of $P$. This oracle is not counted towards the space complexity of the algorithm.
\end{enumerate}
For the semi-streaming setting, we provide a deterministic solution, whereas our solution for the streaming setting is Monte-Carlo randomized. Both algorithms are designed for the $w$-bit word RAM model, where $w = \Omega(\log n)$, and integer alphabets $\Sigma = [0\dd n^{\Oh(1)})$.

\subsection{(Semi-)Streaming Algorithm for Pattern Matching with \texorpdfstring{$k$}{k} Edits}
Our algorithms solve a slightly stronger problem: every element $r\in \Occ_k^E(P,T)$ is augmented with the smallest integer $k'\in [0\dd k]$ such that $r\in \Occ_{k'}^E(P,T)$.
At a very high-level, we reuse the structure of existing streaming algorithms for exact pattern matching and the $k$-mismatch problem~\cite{Porat:09,DBLP:journals/talg/BreslauerG14,DBLP:conf/soda/CliffordFPSS16,clifford2018streaming}.
Namely, we consider $\Oh(\log m)$ prefixes $P_i=P[1\dd \ell_i]$
of exponentially increasing lengths $\ell_i$. The algorithms are logically decomposed into $\Oh(\log m)$ levels,
with the $i$th level receiving $\Occ_k^E(P_{i-1},T)$ and producing $\Occ_k^E(P_i,T)$.
In other words, the task of the $i$th level is determine which $k$-edit occurrences of $P_{i-1}$
can be extended to $k$-edit occurrences of~$P_i$.
When the algorithm processes $T[r]$, the relevant positions $p\in \Occ_k^E(P_{i-1},T)$ are 
those satisfying $|r-p-(\ell_i-\ell_{i-1})|\le k$.
Since each $p\in \Occ_k^E(P_{i-1},T)$ is reported when the algorithm processes $T[p]$,
we need a buffer storing the \emph{active} $k$-edit occurrences of $P_{i-1}$.
We implement it using a recent combinatorial characterization of $k$-edit occurrences~\cite{DBLP:journals/corr/abs-2004-08350}, which classifies strings based on the following notion of approximate periodicity:
\begin{restatable}[$k$-periodic string]{definition}{defkperiodic}
	A string $X$ is \emph{$k$-periodic} if there exists a primitive string $Q$ with $|Q| \le |X|/128k$ such that the edit distance between $X$ and a prefix of $Q^\infty$ is at most $2k$. We call $Q$ a \emph{$k$-period} of~$X$. 
\end{restatable}

The main message of~\cite{DBLP:journals/corr/abs-2004-08350} is that only $k$-periodic strings may have many $k$-edit occurrences.

\begin{corollary}[of {\cite[Theorem 5.1]{DBLP:journals/corr/abs-2004-08350}}]\label{cor:structure-nonperiod}
Let $X \in \Sigma^m$,  $k\in [1\dd m]$, and $Y\in \Sigma^{n}$ with $n\le 2m$.
If $X$ is not $k$-periodic, then $|\Occ_k^E(X, Y)| = \Oh(k^2)$.
\end{corollary}

In particular, if $P_{i-1}$ is not $k$-periodic, then it has $\Oh(k^2)$ active $k$-edit occurrences. For each active occurrence $p\in \Occ_k^E(P_{i-1}, T)$, we maintain an edit-distance sketch $
\skE_k(T(p\dd r])$ and combine it with a sketch $\skE_k(P(\ell_{i-1}\dd \ell_i])$ (constructed at preprocessing)
in order to derive $\ed(T(p\dd r],P(\ell_{i-1}\dd \ell_i])$ or certify that this distance exceeds~$k$.
Since we have stored the smallest $k'\in [0\dd k]$ such that $p\in \Occ_{k'}^E(P,T)$,
this lets us check whether any $k$-edit occurrence of $P_{i-1}$ ending at position $p$
extends to a $k$-edit occurrence of $P_i$ ending at position $r$.
With existing $k$-edit sketches~\cite{sketches,jin2020improved}, this already yields an $\poly(k\log n)$-space implementation in this case.

The difficulty lies in $k$-periodic strings whose occurrences form \emph{chains}.
\begin{definition}[Chain of occurrences]
	Consider strings $X,Y\in \Sigma^*$ and an integer $k\in \Zz$. An increasing sequence of positions $p_1,\ldots,p_c$
	forms a \emph{chain} of $k$-edit occurrences of $X$ in $Y$ if:
		\begin{enumerate}
		\item there is a \emph{difference} string $D\in \Sigma^*$ such that $D=Y(p_j\dd p_{j+1}]$ for $j\in [1\dd c)$, and 
		\item there is an integer $k'\in [0\dd k]$ such that $p_j \in \Occ_{k'}^E(X,Y)\setminus \Occ_{k'-1}^E(X,Y)$ for $j\in [1\dd c]$.
		\end{enumerate}
\end{definition}

\begin{corollary}[of {\cite[Theorem 5.2, Claim 5.16, Claim 5.17]{DBLP:journals/corr/abs-2004-08350}}]\label{cor:structure-period}
Let $X\in \Sigma^m$,  $k\in [1\dd m]$, and $Y\in \Sigma^{n}$ with $n\le 2m$.
If $X$ is $k$-periodic with period $Q$, then $\Occ_k^E(X,Y)$ can be decomposed into $\Oh(k^3)$ chains whose differences are of the form $\rot^{s}(Q)$ with $|m-s|\le 10k$.
\end{corollary}

In the following discussion, assume that $P_{i-1}$ is $k$-periodic with period $Q_{i-1}$.
Compared to the previous algorithm, we cannot afford maintaining a sketch $\skE_k(T(p\dd r])$ for all active $p\in \Occ_k^E(P,T)$.
If $\skE_k$ were a rolling sketch (like the $k$-mismatch sketches of~\cite{clifford2018streaming}), we would compute $\skE_k(D)$ at preprocessing time for all $\Oh(k)$ feasible chain differences $D$ and then, for any two subsequent positions $p_{j},p_{j+1}$ in a chain with difference $D$,
we could use $\skE_k(D)=\skE_k(T(p_{j}\dd p_{j+1}])$ to transform $\skE_k(T(p_j\dd r])$ into $\skE_k(T(p_{j+1}\dd r])$.
However, despite extensive research, no rolling edit distance sketch is known, which remains the main obstacle in designing streaming algorithms for approximate pattern matching with $k$ edits.

Our workaround relies on a novel \emph{encoding} $\qgr(X,Y)$ that, for a pair of strings $X,Y\in \Sigma^*$,
represents a large class of low-distance edit distance alignments between $X,Y$.
In the preprocessing phase of our algorithm, we build $\qgr(P(\ell_{i-1}\dd \ell_i],D^\infty[1\dd \ell_{i}-\ell_{i-1}])$  for every feasible chain difference $D$.
In the main phase, for subsequent positions $p_j\in \Occ_k^E(P_{i-1},T)$ in a chain with difference $D$, we aim to build $\qgr(T(p_j\dd r],D^\infty[1\dd \ell_{i}-\ell_{i-1}])$ when necessary, i.e., $|r-p_j-(\ell_i-\ell_{i-1})|\le k$.
We then combine the two encodings to derive $\qgr(P(\ell_{i-1}\dd \ell_i],T(p_j\dd r])$ and $\ed(P(\ell_{i-1}\dd \ell_i],T(p_j\dd r])$.
Except for such \emph{products} (transitive compositions), our encoding supports \emph{concatenations},
i.e., $\qgr(X_1,Y_1)$ and $\qgr(X_2,Y_2)$ can be combined into $\qgr(X_1X_2,Y_1Y_2)$.
Consequently, it suffices to maintain $\qgr(T(p_c\dd r],D^\infty[1\dd r-p_c-k])$ (where $p_c$ is the rightmost element of the chain). When necessary, we prepend $(j-c)$ copies of $\qgr(D,D)$ (merged by doubling)
and append $\qgr(\eps, D^\infty(r-p_j-k\dd \ell_i-\ell_{i-1}])$ to derive $\qgr(P(\ell_{i-1}\dd \ell_i],D^\infty[1\dd \ell_{i}-\ell_{i-1}])$.

In the semi-streaming setting, we extend $\qgr(T(p_c\dd r],D^\infty[1\dd r-p_c-k])$ one character at a time using read-only random access to $D^\infty$.
In the streaming setting, we cannot afford storing $D$, so we append the entire difference $D$ in a single step
and utilize a new edit-distance sketch $\skq$ that allows retrieving $\qgr(T(r-|D|\dd r],D)$ from $\skq(T(r-|D|\dd r])$ and $\skq(D)$.
The sketch $\skq(D)$ is constructed in the preprocessing phase,
whereas $\skq(T(r-|D|\dd r])$ is built as the algorithm scans $T$. Similarly, we can (temporarily) append any of the $\Oh(k)$ prefixes that of $D$ may arise when $\qgr(T(p_c\dd r],D^\infty[1\dd r-p_c-k])$ is necessary.
A complete presentation of our algorithms is provided in \cref{sec:algorithms}.
Below, we outline the ideas behind our two main conceptual and technical contributions: the encoding $\qgr(\cdot,\cdot)$ and the sketch $\skq(\cdot)$.

\subsection{Greedy Alignments and Encodings} 
Recall that the encoding $\qgr(\cdot,\cdot)$ needs to support the following three operations:
\begin{description}[labelwidth=\widthof{Capped edit distance:} + 5.5mm, labelsep=2mm, align=right]
	\item[Capped edit distance:] given $\qgr(X,Y)$, compute $\ed(X,Y)$ or certify that $\ed(X,Y)$ is large;
	\item[Product:] given $\qgr(X,Y)$ and $\qgr(Y,Z)$, retrieve $\qgr(X,Z)$;
	\item[Concatenation:] given $\qgr(X_1,Y_1)$ and $\qgr(X_2,Y_2)$, retrieve $\qgr(X_1X_2,Y_1Y_2)$.
\end{description}
Our encoding is parameterized with a threshold $k\in \Zp$ such that $\ed(\cdot,\cdot)>k$ is considered large,
and the goal is to achieve $\tOh(k^{\Oh(1)})$ encoding size.
In fact, whenever $\ed(X,Y)>k$, we shall simply assume that $\qgr_k(X,Y)$ is undefined (formally, $\qgr_k(X,Y)=\bot$).
Consequently, products and concatenations will require sufficiently large thresholds in  
the input encodings so that if either of them is undefined, the output encoding is also undefined.

In order to support concatenations alone, we could use so-called \emph{semi-local} edit distances.
For now, suppose that we only need to encode pairs of equal-length strings.\footnote{Pairs of strings of any lengths can be supported in the same way provided that concatenations require larger input thresholds (compared to the output threshold) to accommodate length differences.} 
Through a sequence of concatenations, we may only extend $\qgr_k(X,Y)$ to $\qgr_k(X',Y')$ so that $X=X'[\ell \dd r)$ and $Y=Y'[\ell\dd r)$.
For any alignment $\A'$ of $X',Y'$ with $\cost(\A')\le k$, consider the induced alignment $\A:=\A'_{[\ell\dd r),[\ell\dd r)}$. Note that $\A$ mimics the behavior of $\A'$ except that it deletes some characters
at the extremes of $X$ and $Y$ (which $\A'$ aligns outside $Y$ and $X$, respectively).
By \cref{fact:edk2}, we have $\cost(\A)\le k$ and, in particular, $\A$ deletes at most $k$ characters at the extremes of $X$ and~$Y$. If, after performing these deletions, we replace $\A$ with an optimal alignment between
the remaining fragments of $X$ and $Y$, this modification may only decrease $\cost(\A)$ and $\cost(\A')$.
Consequently, it suffices to store the $\Oh(k)$ characters at the extremes of $X,Y$ and the $\Oh(k^4)$ edit distances\footnote{In fact, $\Oh(k^2)$ edit distances suffice and they can be encoded in $\Oh(k)$ space using techniques of Tiskin~\cite{Tiskin13}.}
 between long fragments of $X$ and $Y$.
In a sense, this encoding represents $\Oh(k^4)$ alignments between $X,Y$
that are sufficient to derive an optimal alignment of any extension.

The main challenge is to handle products, for which we develop a \emph{greedy encoding} $\gr_k(X,Y)$
that compactly represents the following family $\ga_k(X,Y)$ of \emph{greedy alignments} of~$X,Y$.

\begin{restatable}[Greedy alignment]{definition}{defgreedy}\label{def:greedy_alignment}
We say that an alignment $\A$ of two strings $X, Y\in \Sigma^*$ is \emph{greedy} if  $X[x] \ne Y[y]$ holds for every $(x,y)\in \brkp(\A) \cap ([1\dd |X|]\times [1\dd |Y|])$.
Given $k\ge 0$, we denote by $\ga_k(X,Y)$ the set of all greedy alignments $\A$ of $X,Y$
satisfying $\cost(\A)\le k$.
\end{restatable}

Intuitively, whenever a greedy alignment encounters a pair of matching characters $X[x]$ and $Y[y]$,
it must (greedily) match these characters (it cannot delete $X[x]$ or $Y[y]$).
As stated below, this restriction does not affect the optimal cost.
(All claims are proved in \cref{sec:greedy}.)

\begin{restatable}{fact}{greedyoptimal}\label{obs:greedy_and_optimal}
For any two strings $X, Y\in \Sigma^*$, there is an optimal greedy alignment of $X, Y$.
\end{restatable}

For strings $X,Y\in \Sigma^*$ and an integer $k\ge \ed(X,Y)$, we define a set $\mtch_k(X,Y)$ of \emph{common matches} of all alignments $\A\in \ga_k(X,Y)$; formally $\mtch_k(X,Y) = \bigcap_{\A \in \ga_k(X,Y)} \mtch_{X,Y}(\A)$.
In our greedy encoding, we shall mask out all the characters involved in the common matches.
Below, this transformation is defined for an arbitrary non-crossing matching of~${X,Y}$.

\begin{restatable}{definition}{defnoncrossing}
Let $M$ be a non-crossing matching of strings $X,Y\in \Sigma^*$.
We define $X^M,Y^M$ to be the strings obtained from $X,Y$ by replacing $X[x]$ and $Y[y]$
with $\# \notin \Sigma$ for every $(x,y)\in M$. 
We refer to $\#$ as a \emph{dummy symbol} and to maximal blocks of $\#$'s as \emph{dummy segments}.
\end{restatable}

The following lemma proves that masking out common matches does not affect $\ed(X,Y)$ or $\mtch_k(X,Y)$
provided that we \emph{enumerate} the dummy symbols,
that is, any string $Z$ is transformed to $\num(Z)$ by replacing the $i$th leftmost occurrence of $\#$ with a unique symbol $\#_i \notin \Sigma$.
\begin{restatable}{lemma}{lmmasking}\label{lm:masking_does_not_change_ga}
Consider strings $X,Y\in \Sigma^*$, an integer $k\ge \ed(X,Y)$, and a set  $M\sub \mtch_k(X,Y)$.
Then, $\ed(X,Y)=\ed(\num(X^M),\num(Y^M))$ and $\mtch_k(X,Y)=\mtch_k(\num(X^M),\num(Y^M))$.
\end{restatable}

At the same time, after masking out the common matches, the strings become compressible.
Intuitively, this is because once two greedy alignments converge, they stay together until they encounter a mismatch.
Moreover, when two alignments proceed in parallel without any mismatch, this incurs a small period (at most $2k$)
that is captured by the LZ factorization.
\begin{restatable}{lemma}{lmgreedysize}\label{lm:greedy_size}
	Let $M=\mtch_k(X,Y)$ for strings $X,Y\in \Sigma^*$ and a positive integer $k\ge \ed(X,Y)$.
	Then, $|\LZ(X^M)|, |\LZ(Y^M)| = \Oh(k^2)$,
	and $X^M,Y^M$ contain $\Oh(k)$ dummy segments.
\end{restatable}

Consequently, for $k\ge \ed(X,Y)$, we could define the \emph{greedy encoding} $\gr_k(X,Y)$ so that it consists
of $\LZ(X^M)$ and $\LZ(Y^M)$. Instead, we use a more powerful compressed representation
(developed in \cref{sec:compr}) that supports more efficient queries concerning $X^M$ and $Y^M$.

Even though $\gr_k(X,Y)$ is small, $\ga_k(X,Y)$ may consist of $2^{\Theta(k)}$ alignments,
which is why constructing $\gr_k(X,Y)$ in $\poly(k)$ time is far from trivial. 
The following combinatorial lemma lets us obtain an $\Oh(k^5)$-time algorithm in \cref{sec:greedyalg}. 
Intuitively, the alignments in $\ga_k(X,Y)$ can be interpreted as paths in a directed 
acyclic graph with $\Oh(k^5)$ branching vertices.

\begin{restatable}{lemma}{lmgreedyalg}\label{lm:greedyalg}
	For all $X,Y\in \Sigma^*$ and $k\in \Zp$, the set $\brkp_k(X,Y) = \hspace{-.4cm}  \bigcup\limits_{\A \in \ga_k(X,Y)} \hspace{-.4cm} \brkp(\A)$ is of size~$\Oh(k^5)$.
\end{restatable}

The reason why $\gr_k(X,Y)$ supports products
is that every greedy alignment of $X,Z$ can be interpreted as a product of a greedy alignment of $X,Y$
and a greedy alignment of $Y,Z$.
\begin{restatable}{definition}{defalignmentsproduct}\label{def:alignments_product}
	Consider strings $X,Y,Z\in \Sigma^*$, an alignment $\A^{X,Y}$ of $X,Y$, an alignment $\A^{Y,Z}$ of $Y,Z$, and an alignment $\A^{X,Z}$ of $X,Z$. 
	We say that $\A^{X,Z}$ is a \emph{product} of $\A^{X,Y}$ and $\A^{Y,Z}$ if, for every $(x,z)\in \A^{X,Z}$, there is $y\in [1\dd |Y|+1]$ such that $(x,y)\in \A^{X,Y}$ and $(y,z)\in \A^{Y,Z}$.
\end{restatable}
	
\begin{restatable}{lemma}{lmgreedyproduct}\label{lm:greedy_product}
	Consider strings $X,Y,Z\in \Sigma^*$ and $k\in \Zz$.
	Every alignment $\A^{X,Z}\in \ga_k(X,Z)$ is a product of alignments $\A^{X,Y}\in \ga_{d}(X,Y)$ and $\A^{Y,Z}\in \ga_d(Y,Z)$, where $d = 2k+\ed(X,Y)$.
\end{restatable}
As a result, $\gr_d(X,Y)$ and $\gr_d(Y,Z)$ contain enough information to derive $\gr_k(X,Z)$.
The underlying algorithm propagates the characters of $Y$
stored in $\gr_d(X,Y)$ and $\gr_d(Y,Z)$ along the matchings $\mtch_{d}(Y,Z)$ and $\mtch_d(X,Y)$, respectively.
Then, $\gr_k(X,Z)$ is obtained by masking out all the characters corresponding to $\mtch_k(X,Y)$ (see \cref{sec:product} for details).

To support concatenations,
we extend the family $\ga_k(X,Y)$ of greedy alignments to a family $\qga_k(X,Y)$ of \emph{quasi-greedy alignments},
which are allowed to delete a prefix of $X$ or $Y$ in violation of \cref{def:greedy_alignment}.
The \emph{quasi-greedy encoding} $\qgr_k(X,Y)$ is defined analogously to $\gr_k(X,Y)$.
Equivalently, $\qga_k(X,Y)$ can be derived from $\ga_{k+1}(\$_1X,\$_2Y)$, where $\$_1\ne \$_2$ are sentinel symbols outside $\Sigma$, by taking the alignments induced by $X,Y$.
The latter characterization makes all our claims regarding $\ga$ and $\gr$ easily portable to $\qga$ and $\qgr$
(see \cref{sec:quasi}).
In particular, this is true for the sketches, described in the following subsection for $\gr$ only.

\subsection{Edit Distance Sketches}
Recall that we need an edit-distance sketch $\skE$ allowing to retrieve $\gr_k(X,Y)$ from $\skE_k(X)$ and $\skE_k(Y)$
for any strings $X,Y\in \Sigma^{\le n}$ and any threshold $k\in [0\dd n]$.
Furthermore, we need to make sure that $\skE_k(S)$ can be computed given streaming access to $S\in \Sigma^{\le n}$, and that the encoding and decoding procedures use $\poly(k, \log n)$ space.
While the existing sketches~\cite{sketches,jin2020improved} are designed to compute the exact edit distance $\ed(X,Y)$ capped with $k$, we believe that they could be adapted to output $\qgr_k(X,Y)$.
Nevertheless, these sketches are relatively large (taking $\tOh(k^8)$ and $\tOh(k^3)$ bits, respectively),
and we would need to strengthen the bulk of their analyses to prove that, in principle, they provide enough information to retrieve $\qgr_k(X,Y)$.
Hence, to further demonstrate the power of our techniques, we devise a novel $\tOh(k^2)$-size sketch specifically designed to output $\qgr_k(X,Y)$. We note that the $\tOh(k^2)$ size is optimal for $\qgr_k(X,Y)$, but we are not aware of a matching lower bound for retrieving $\ed(X,Y)$ capped with~$k$.

Just like the sketches of~\cite{sketches,jin2020improved}, ours relies on the embedding of Chakraborty, Goldenberg, and Kouck\'{y}~\cite{CGK}. The CGK algorithm performs a random walk over the input string (with forward and stationary steps only). In abstract terms, such walk can be specified as follows:

\begin{definition}[Complete walk]\label{def:walk}
For a string $S\in \Sigma^*$, we say that $(s_t)_{t=1}^{m+1}$ is an \emph{$m$-step complete walk over $S$} if $s_1 = 1$, $s_{m+1} = |S|+1$, and $s_{t+1}\in \{s_t,s_t+1\}$ for $t\in [1\dd m]$.
\end{definition}
	
For any two strings $X,Y\in \Sigma^*$, the two walks underlying $\CGK(X)$ and $\CGK(Y)$ can be interpreted as an edit distance alignment using the following abstract definition:
\begin{definition}[Zip alignment]\label{def:zip}
	The \emph{zip alignment} of $m$-step complete walks $(x_t)_{t=1}^{m+1}$ and $(y_t)_{t=1}^{m+1}$ over ${X,Y\in \Sigma^*}$
	is obtained by removing repeated entries in $(x_t,y_t)_{t=1}^{m+1}$.
\end{definition}
	
The key result of~\cite{CGK} is that the cost of the zip alignment of CGK walks over $X,Y\in \Sigma^*$ is $\Oh(\ed(X,Y)^2)$ with good probability, which is then exploited to derive a metric embedding (mapping edit distance to Hamming distance) with quadratic distortion. In our sketch, we also need to observe that the CGK alignment is greedy and that its width is $\Oh(\ed(X,Y))$ with good probability. The following proposition, proved in \cref{sec:CGK}, provides a 
complete black-box interface of the properties of the CGK algorithm utilized in our sketches. 
It also encapsulates Nisan's pseudorandom generator~\cite{Nisan} that reduces the number of (shared) random bits.

\begin{restatable}{proposition}{prpalg}\label{prp:alg}
For every constant $\delta \in (0,1)$, there exists a constant $c$ and an algorithm $\Alg$ that, given an integer $n$, a seed $r$ of $\Oh(\log^2 n)$ random bits, and a string $S\in \Sigma^{\le n}$, outputs a $3n$-step complete walk $\Alg(n,r,S)$ over $S$ satisfying the following property for all  $X,Y\in \Sigma^{\le n}$ and the zip alignment $\A_\Alg$ of  $\Alg(n,r,X)$ and $\Alg(n,r,Y)$:
\[\Pr_r\left[ \A_\Alg\in \ga_{c\cdot \ed(X,Y)^2}(X,Y)\text{ and }\width(\A_\Alg)\le c\cdot\ed(X,Y)\right] \ge 1-\delta.\]
Moreover, $\Alg$ is an $\Oh(\log^2 n)$-bit streaming algorithm that costs $\Oh(n\log n)$ time and reports any element $s_t\in [1\dd |S|]$ of $\Alg(n,r,S)$ while processing the corresponding character $S[s_t]$.
\end{restatable}

Next, we analyze the structural similarity between $\A_\Alg$ and any alignment $\A\in \ga_k(X,Y)$.
Based on \cref{prp:alg}, we may assume that $\A_{\Alg}\in \ga_{\Oh(k^2)}(X,Y)$ and $\width(\A_\Alg)=\Oh(k)$.
Consider the set $M=\mtch(\A)\cap\mtch(\A_\Alg)$ of the common matches of $\A$ and $\A_\Alg$ and the string $X^M$ obtained by masking out the underlying characters of $X$.
Whereas \cref{lm:greedy_size} immediately implies that the $\LZ$ factorization of $X^M$ consists of $\Oh(k^4)$ phrases, a more careful application of the same technique provides a refined bound of $\Oh(k^2)$ phrases. Furthermore, there are $\Oh(k)$ dummy segments in $X^M$
and, if $X[x]$ is not masked out in $X^M$ (for some $x\in [1\dd |X|]$), then there is a breakpoint $(x',y')\in \brkp_{X,Y}(\A_{\Alg})$ with $x'\in [1\dd x]$ and $|\LZ(X[x'\dd x])| = \Oh(k)$.
Intuitively, this means that $\A$ and $\A_{\Alg}$ diverge only within highly compressible regions following the breakpoints $\brkp_{X,Y}(\A_{\Alg})$. We call these regions \emph{forward contexts} (formally, a forward context is the longest fragment starting at a given position and satisfying certain compressiblity condition).
Since our choice of $\A\in \ga_k(X,Y)$ was arbitrary, any two alignments $\A,\A'\in \ga_k(X,Y)$ diverge only within these forward contexts. Hence, in order to reconstruct $\gr_k(X,Y)$ and, in particular, $X^{\mtch_k(X,Y)}$, the sketch should be powerful enough to retrieve all characters in forward contexts of breakpoints $\brkp_{X,Y}(\A_{\Alg})$.
Even though $\brkp_{X,Y}(\A_{\Alg})$ could be of size $\Theta(k^2)$, due to the aforementioned bounds on $|\LZ(X^M)|$ and the number of dummy segments in $X^M$, it suffices to take $\Oh(k)$ among these forward contexts to cover the unmasked regions of $X^M$ and $X^{\mtch_k(X,Y)}$.
Each context can be encoded in $\tOh(k)$ bits, so this paves a way towards sketches of size $\tOh(k^2)$.

Nevertheless, while processing a string $X\in \Sigma^{\le n}$,
we only have access to the string $X$ and the $m$-complete walk $(x_t)_{t=1}^m = \Alg(n,r,X)$ over $X$.
In particular, depending on $Y$, any position in $X$ could be involved in a breakpoint.
A naive strategy would be to build a \emph{context encoding} $\bCGK(X)[1\dd m]$ that stores at $t\in [1\dd m]$ (a compressed representation of) the forward context starting at $X[x_t]$,
and then post-process it using a Hamming-distance sketch.
This is sufficiently powerful because $X[x_t]\ne Y[y_t]$ holds for any $(x_t,y_t)\in \brkp(\A_\Alg)$ (recall that $\A_\Alg$ is greedy). Unfortunately, this construction does not guarantee any upper bound on $\hd(\bCGK(X),\bCGK(Y))$ in terms of $k$. (For example, if $X$ is compressible, modifying its final character of $X$ affects the entire $\bCGK(X)$.)
Hence, we sparsify $\bCGK(X)$ by placing a blank symbol $\bot$ at some positions $\bCGK(X)[t]$
so that just a few forward contexts stored in $\bCGK(X)[t]$ cover any single position in $X$. 

This brings two further challenges.
First, if $X[x]$ is involved in a breakpoint, then we are only guaranteed that it is covered by some forward context $X[p\dd q)$  of $\bCGK(X)[t]$ (i.e., $x\in [p\dd q)$). In particular, the forward context starting at position $x$ could extend beyond $X[x\dd q)$. Hence, the string $\bCGK(X)[t]$ actually stores \emph{double forward contexts} $X[p\dd r)=X[p\dd q)X[q\dd r)$ defined as the concatenation of the forward contexts of $X[p]$ and $X[q]$.
We expect this double forward context $X[p\dd r)$ to cover the entire forward context of $X[x]$.
Unfortunately, this is not necessarily true if we use the Lempel--Ziv factorization to quantify compressibility:
we could have $|\LZ(X[x\dd r))| < |\LZ(X[q\dd r))|$ because $\LZ(\cdot)$ is not monotone.
Instead, we use an ad-hoc compressibility measure defined as $\maxLZ(S) = \max_{[\ell\dd r)\sub [1\dd |S|]} \LZ(\rev{S[\ell\dd r)})$. Here, maximization over substrings guarantees monotonicity whereas reversal helps designing an efficient streaming algorithm constructing contexts (beyond the scope of this overview).

Another challenge is that the sparsification needs to be consistent between $\bCGK(X)$ and $\bCGK(Y)$: assuming $\ed(X,Y)\le k$, we should have $\hd(\bCGK(X),\bCGK(Y))=\tOh(k)$, which also accounts for mismatches between $\bot$ and a stored double forward context.
This rules out a naive strategy of covering $X$ from left to right using disjoint forward contexts: any substitution at the beginning of $X$ could then have a cascade of consequences throughout $\bCGK(X)$. Hence, we opt for a memory-less strategy that decides on $\bCGK(X)[t]$ purely based on the forward contexts $X[x_{t-1}\dd x'_{t-1})$ and $X[x_{t}\dd x'_{t})$. For example, we could set $\bCGK(X)[t]=\bot$ unless the smallest dyadic interval containing $[x_{t-1}\dd x'_{t-1})$ differs from the smallest dyadic interval containing $[x_{t}\dd x'_{t})$ (a dyadic interval is of the form $[i2^j\dd (i+1)2^j)$ for some integers $i,j\ge 0$).
With this approach, each position of $X$ is covered by at least one and at most $\Oh(\log |X|)$ forward contexts. Furthermore, substituting any character in $X$ does not have far-reaching knock-on effects.
Unfortunately, insertions and deletions are still problematic as they shift the positions $x_t$.
Thus, the decision concerning $\bCGK(X)[t]$ should be independent of the numerical value of $x_t$.
Consequently, instead of looking at the smallest dyadic interval containing $[x_{t}\dd x'_{t})$, we choose the largest $t'$ such that $[x_t\dd x'_t) = [x_t \dd x_{t'})$, and we look at the smallest dyadic interval containing $[t\dd t')$.

With each forward context $X[x_t\dd x'_t)$ retrieved, we also need to determine the value $x_t$ (so that we know which fragment of $X$ we can learn from $\bCGK(X)[t]$).
To avoid knock-on effects, we actually store differences $x_{t}-x_{u}$ with respect to the previous index satisfying $\bCGK(X)[u]\ne \bot$. 
This completes the intuitive description of the context encoding $\bCGK$ studied in \cref{sec:cover}.

Our edit-distance sketch contains the Hamming-distance sketch of $\bCGK(X)$. For this, we use an existing construction~\cite{clifford2018streaming}, augmented in a black-box manner to support large alphabets (recall that the each forward contexts takes $\tOh(k)$ bits). Furthermore, to retrieve the starting positions $x_t$ (rather than just the differences $x_t - x_u$), we use a hierarchical Hamming-distance sketch similar to those used in~\cite{sketches,jin2020improved}. This way, given sketches of $X$ and $Y$, we can recover all characters that remain unmasked in $X^{\mtch_k(X,Y)}$ and $Y^{\mtch_k(X,Y)}$. The tools of \cref{sec:greedy} are then used to compute $\ed(X,Y)$ (or certify $\ed(X,Y)>k$) and retrieve the greedy encoding $\gr_k(X,Y)$ (see \cref{sec:actualsketches}).
We summarize the properties of the edit distance sketches below:

\begin{restatable}{theorem-restatable}{thskE}\label{thm:ske}
For every constant $\delta \in (0,1)$, there is a sketch $\skE_k$ (parametrized by integers $n\ge k \ge 1$, an alphabet $\Sigma = [0\dd n^{\Oh(1)})$, and a seed of $\Oh(\log^2 n)$ random bits) such that:
\begin{enumerate}[label=\textrm{(\alph*)}]
	\item The sketch $\skE_k(S)$ of a string $\Sigma^{\le n}$ takes $\Oh(k^2 \log^3 n)$ bits. Given streaming access to $S$, it can be constructed in $\tOh(nk)$ time using $\tOh(k^2)$ space.
	\item There exists an $\tOh(k^2)$-space decoding algorithm that, given $\skE_{k}(X), \skE_{k}(Y)$ for $X,Y \in \Sigma^{\le n}$,
	with probability at least $1-\delta$ outputs $\gr_k(X,Y)$ and $\min(\ed(X,Y),k+1)$.
	Retrieving $\gr_k(X,Y)$ costs $\tOh(k^5)$ time, whereas computing  $\min(\ed(X,Y),k+1)$ costs $\tOh(k^3)$ time. 
\end{enumerate}
\end{restatable}


%% file: compression.tex
In this section, we develop a data structure that stores a string $X$ in $\Oh(|\LZ(X)|\log^2 |X|)$ space,
supporting various operations in a relatively efficient way (the $\log |X|$ factors are not optimized).

Our data structure can be constructed not only from $\LZ(X)$,
but from any \emph{LZ-like representation} describing a factorization $X=F_1\cdots F_f$ into non-empty phrases such that each phrase $F_j$ with $|F_j|>1$ is a previous factor (unlike LZ77, the phrases do not need to be a maximal).

\begin{observation}\label{fct:lz_properties}
    Every LZ-like factorization of a string $X\in \Sigma^*$ has at least $|\LZ(X)|$ phrases. Moreover, 
    $|\LZ(X)| \le |\LZ(XY)| \le |\LZ(X)|+|\LZ(Y)|$ holds for all strings $X,Y\in \Sigma^*$.
\end{observation}

\begin{proposition}\label{prp:RLSLP}
    In the $w$-bit word RAM model, every string $X\in \Sigma^n$ satisfying $w = \Omega(\log n + \log |\Sigma|)$, $z = |\LZ(X)|$, and $\rev{z}=|\LZ(\rev{X})|$ has a representation $\CR(X)$ that uses 
 $\Oh(z \log^2 n)$ space and supports the following queries:
    \begin{enumerate}[label=\textrm{(\alph*)}]
        \item retrieve $n=|X|$, in $\Oh(1)$ time;\label{it:length}
        \item given $i\in [1\dd n]$, retrieve $X[i]$, in $\Oh(\log n)$ time;\label{it:access}
        \item given $i,j\in [1\dd n]$, compute $\LCE(X[i\dd n],X[j\dd n])$, in $\Oh(\log n)$ time;\label{it:lce}
        \item given $i,j\in [1\dd n]$, compute $\LCE(\rev{X}[i\dd n],\rev{X}[j\dd n])$, in $\Oh(\log n)$ time;\label{it:rlce}
        \item compute $\LZ(X)$, in $\Oh(z\log^4 n)$ time;\label{it:lz}
        \item compute $\LZ(\rev{X})$, in $\Oh(\rev{z}\log^4 n)$ time;\label{it:rlz}
        \item given $i,j\in [1\dd n]$ with $i\le j$, compute $\CR(X[i\dd j])$, in $\Oh(z \log^4 n)$ time.\label{it:extract}
    \end{enumerate}
    Moreover, $\CR(X)$ can be constructed:
    \begin{enumerate}[label=\textrm{(\alph*)},resume]
        \item in $\Oh(f \log^2 n)$ time given an $f$-phrase LZ-like representation of $X$ or $\rev{X}$;\label{it:fromLZ}
        \item in $\Oh(z\log^4 n)$ time given $\CR(X[1\dd i))$ and $\CR(X[i\dd n])$ for some $i\in [1\dd n]$.\label{it:concat}
    \end{enumerate}
\end{proposition}

\renewcommand{\S}{\mathcal{S}}
\newcommand{\Tr}{\mathcal{T}}

Before proving \cref{prp:RLSLP}, we need to introduce several more compression schemes.
These concepts are not used any of the subsequent sections.

A \emph{straight-line grammar} is a tuple $\G=(\S, \Sigma, \rhs, S)$,
where $\S$ is a finite set of \emph{symbols}, $\Sigma \sub \S$ is a set of \emph{terminal symbols},
$\rhs : (\S\sm \Sigma) \to \S^*$ is the \emph{production} (or \emph{right-hand side}) function, and $S\in\S$ 
is the start symbol. We further require existence of an order $\prec$ on $\S$ such that $B\prec A$ if $B$ occurs in $\rhs(A)$. 

The \emph{expansion} function $\exp:\S \to \Sigma^+$ is defined recursively:
\[\exp(A) = \begin{cases}
    A & \text{if }A\in \Sigma,\\
    \exp(B_1)\cdots \exp(B_k) & \text{if }\rhs(A)= B_1\cdots B_k,
\end{cases}
\]
We say that $\G$ is a \emph{grammar-compressed representation} of $\exp(S)$.

The \emph{parse tree} $\Tr(A)$ of a symbol $A\in \S$ is a rooted ordered tree with each node $\nu$ associated to a symbol $s(\nu)\in \S$. 
The root of $\Tr(A)$ is a node $\rho$ with $s(\rho)=A$.
If $A \in \Sigma$, then $\rho$ has no children. 
Otherwise, if $\rhs(A)=B_1\cdots B_k$, then $\rho$ has $k$ children,
and the subtree rooted at the $i$th child is (a copy of) $\Tr(B_i)$.
The parse tree of a grammar $\G$ is defined as the parse tree $\Tr(S)$ of the starting symbol $S$,
and the height of $\G$ is defined as the height of $\Tr(S)$.

Each node $\nu$ of $\Tr(A)$ is associated with an occurrence $\exp(\nu)$ of $\exp(s(\nu))$ in $\exp(A)$.
For the root $\rho$, we define $\exp(\rho)=\exp(A)[1\dd |\exp(A)|]$ to be the whole $\exp(A)$.
Moreover, if $\exp(\nu)=\exp(A)[\ell\dd r)$, $\rhs(s(\nu))=B_1\cdots B_k$,
and $\nu_1,\ldots,\nu_k$ are the children of $\nu$, then $\exp(\nu_i)=\exp(A)[r_{i-1}\dd r_{i})$,
where $r_i = \sum_{j=1}^{i} |\exp(B_j)|$ for $0\le i \le k$.
This way, the fragments $\exp(\nu_i)$ form a partition of $\exp(\nu)$,
and $\exp(\nu_i)$ is indeed an occurrence $\exp(s(\nu_i))$ in $\exp(A)$.

A \emph{straight-line program} (SLP) is a straight-line grammar $\G=(\S,\Sigma,\rhs,S)$
such that $|\rhs(A)|=2$ holds for each \emph{nonterminal} $A\in \S \sm \Sigma$. 
In a \emph{run-length straight-line program} (RLSLP), productions of the form $\rhs(A)=B^k$, with $B\in \S$ and $k\in \mathbb{Z}_{\ge 2}$, are also allowed.

An RLSLP $\G$ of size $g$ (with $g$ symbols) representing a string of length $n$
can be stored in $\Oh(g)$ space ($\Oh(g \log (n+g))$ bits) with each non-terminal $A\in \S\sm \Sigma$ 
storing $\rhs(A)$ (represented by $(B,k)$ if $\rhs(A)=B^k$) and $|\exp(A)|$. 
This representation allows for efficiently traversing the parse tree~$\Tr_\G$: given a node $\nu$ represented as a pair $(s(\nu),\exp(\nu))$, it is possible to retrieve in constant time an analogous representation of any child of $\nu$.

\begin{fact}\label{fct:rlslp}
Let $\G$ be an RLSLP of size $g$ representing a string $X$.
\begin{enumerate}[label=(\alph*)]
    \item An LZ-like factorization of $X$ can be constructed in $\Oh(g)$ time.\label{it:LZ}
    \item For every fragment $X[i\dd j]$, an RLSLP representing $X[i\dd j]$ can be constructed in $\Oh(g)$ time.\label{it:sub}
\end{enumerate}
\end{fact}
\begin{proof}
    We traverse the parse tree $\Tr(S)$ in pre-order, skipping some subtrees depending on the application.

    \ref{it:LZ} For each symbol $A$, we maintain an already visited node $\nu$ with symbol $s(\nu)=A$, if any.
    When visiting a node $\nu$, we first check whether any node $\nu'$ with $s(\nu')=s(\nu)$ 
    has already been visited. If so, we retrieve the expansion $\exp(\nu')=X[i'\dd i'+\ell)$, and we output $(i',\ell)$ as a part of the LZ-like representation.
    Otherwise, we save $\nu$ as an already visited node with symbol $s(\nu)$ and proceed as follows:
    \begin{itemize}
        \item If $\nu$ is a leaf, then we output $s(\nu)\in \Sigma$ as a part of the LZ-like representation. 
        \item If $\nu$ has two children, we process them recursively.
        \item If $\nu$ has $k\ge 3$ children, we process the first child $\nu_1$ recursively,
        retrieve $\exp(\nu_1)= X[i\dd i+\ell)$, and output $(i,(k-1)\ell)$ as a part of the LZ-like representation.
    \end{itemize}
    The overall running time $\Oh(g)$ is amortized by the number of distinct symbols visited.
    
    \ref{it:sub} Here, an intermediate goal is to construct a sequence $(A_1,m_1),\ldots,(A_t,m_t)$ with $A_p\in \S$ and $m_p \in \mathbb{Z}_{\ge 1}$ such that $\exp(A_1)^{m_1}\cdots \exp(A_t)^{m_t}=X[i\dd j]$ and $t=\Oh(g)$.
    We only visit nodes $\nu$ such that $\exp(\nu)$ intersects $X[i\dd j]$.
    \begin{itemize}
        \item If $\exp(\nu)$ is contained in $X[i\dd j]$, we simply output $(s(\nu),1)$ and skip the subtree of $\nu$.
        \item Otherwise, we determine the children $\nu_\ell,\ldots,\nu_r$ of $\nu$ whose expansions intersect $X[i\dd j]$.
        Then, we recurse into $\nu_\ell$, output $(s(\nu_\ell),r-\ell-1)$ if $r > \ell+1$,
        and recurse into $\nu_r$ if $r > \ell$.
    \end{itemize}
    The number of visited nodes is proportional to the height of $\G$ and thus this traversal takes $\Oh(g)$ time.
    In the post-processing, for each pair $(A_p, m_p)$, we create a new symbol $B_p$ with $\rhs(B_p)=A_p^{m_p}$ (if $m_p \ge 2$) or set $B_p$ as an alias of $A_p$ (if $m_p=1$).
    Next, we set $C_1$ to be an alias of $B_1$ and, for $p\in [2\dd t]$,
    create a new symbol $C_p$ with $\rhs(C_p)=C_{p-1}B_{p-1}$.
    Finally, we return the extended RLSLP with $C_t$ as the new starting symbol.
\end{proof}

    \begin{theorem}[I~\cite{DBLP:conf/cpm/I17}]\label{thm:I}
        Let $\G$ be a size-$g$ SLP representing a string $X\in \Sigma^n$
        and let $g^*$ be the minimum size of an SLP representing $X$.
        Given $\G$, in $\Oh(g\log n)$ time, one can construct a size-$\Oh(g^* \log n)$ RLSLP
        representing $X$ and a size-$\Oh(g^* \log n)$ data structure that answers the following queries in $\Oh(\log n)$ time:
        \begin{description}
            \item[Access:] Given $i\in [1\dd n]$, retrieve $S[i]$;
            \item[LCE queries:] Given $i,j\in [1\dd n]$, compute $\LCE(S[i\dd n],S[j\dd n])$.
        \end{description}
    \end{theorem}

    \begin{proof}[Proof of \cref{prp:RLSLP}]
        We define $\CR(X)$ so that it consists of two instances of the data structure of \cref{thm:I}, one for $X$ and one for $\rev{X}$, including the RLSLPs representing $X$ and $\rev{X}$.
        Thus, the size of $\CR(X)$ is $\Oh(g^* \log n)$, where $g^*$ is the minimum size of an SLP representing $X$ (it is the same for $\rev{X}$).

        To implement~\ref{it:fromLZ}, we first use~\cite[Theorem 6.1]{Kempa2019} and build a size-$\Oh(f\log n)$ SLP representing $X$ or $\rev{X}$.
        By reversing the right-hand sides of all the productions, we obtain an analogous SLP representing $\rev{X}$ or $X$, respectively.
        Finally, we pass these SLPs to the construction algorithm of \cref{thm:I}.
        Overall, the running time is $\Oh(f\log^2 n)$.

        In~\ref{it:concat}, given $\CR(X[1\dd i))$ and $\CR(X[i\dd n])$, we convert the underlying RLSLPs to LZ-like representations (\cref{fct:rlslp}), concatenate them into an $\Oh(g^*\log n)$-phrase LZ-like representation of $X$, and apply~\ref{it:fromLZ}. Overall, this takes $\Oh(g^*\log^3 n)=\Oh(z\log^4 n)$ time.

        As for~\ref{it:length}, we note that $|X|$ can be trivially retrieved from the RLSLP representing $X$.
        Queries~\ref{it:access},~\ref{it:lce}, and~\ref{it:rlce} are implemented directly using \cref{thm:I}.

        In~\ref{it:lz}, we construct an index of~\cite[Theorem 6.11]{Kempa2019} that, given a fragment $X[i\dd i+\ell)$, in $\Oh(\log^3 n)$ time locates the leftmost position $i'$ with $X[i'\dd i'+\ell)=X[i\dd i+\ell)$.
        This index can be constructed in $\Oh(g^*\log^3 n)$ time: it is already based on the RLSLP of \cref{thm:I}, and the extra construction time is $\Oh(\log^2 n)$ per RLSLP symbol; see~\cite[Lemma 6.9]{Kempa2019}.
        This way, each phrase $F_j$ of the LZ77 parsing of $X$ can be constructed, by binary search, with $\Oh(\log n)$ queries to the index, for a total of $\Oh(z \log^4 n)$ time (see also~\cite[Section 5]{NISHIMOTO2020116}).
        The algorithm for~\ref{it:rlz} is symmetric.

        As for~\ref{it:extract}, we construct an $\Oh(g^* \log n)$-phrase LZ-like representation of $X[i\dd j]$ using \cref{fct:rlslp}, and then apply~\ref{it:fromLZ}. Overall, this takes $\Oh(g^*\log^3 n)=\Oh(z\log^4 n)$ time.
    \end{proof}

%% file: greedy.tex
\defgreedy*

\greedyoptimal*
\begin{proof}
  Let $\A = (x_t,y_t)_{t=1}^m$ be an optimal alignment maximizing the sum $\sum_{t=1}^m (x_t+y_t)$.
  For a proof by contradiction, suppose that $\A$ is not greedy.
  Then, there exists $i\in [1\dd m]$ such that $(x_i,y_i)\in \brkp(A)\cap ([1\dd |X|]\times [1\dd |Y|])$ and $X[x_i]=Y[y_i]$. Observe that $x_{i+1}=x_i$ or $y_{i+1}=y_i$ and, by symmetry, assume $x_{i+1}=x_i$ without loss of generality. Let us define $j> i$ so that $x_i = \cdots = x_{j}<x_{j+1}$
  and consider two cases depending on whether $y_j = y_{j+1}$:
  \begin{itemize}
    \item If $y_j = y_{j+1}$, we define an alignment $\A'$ obtained from $\A$ by deleting $(x_j,y_j)$.
    It is easy to see that $\A'$ is a valid alignment. Moreover, $\brkp(\A')\sub \brkp(\A)$ and $\brkp(\A)\setminus \brkp(\A') = \{(x_j,y_j)\}$, so $\cost(\A')<\cost(\A)$, contradicting the choice of $\A$ as an optimal alignment.
    \item If $y_j < y_{j+1}$, we define an alignment  $\A' = (x'_t,y'_t)_{t=1}^m$ obtained from $\A$
    by incrementing $x_t$ for $t\in (i\dd j]$.  Moreover, $\brkp(\A)\setminus \brkp(\A') \sub \{(x_t,y_t) : t\in [i\dd j)\}$ and $\brkp(\A')\setminus \brkp(\A) = \{(x'_t,y'_t) : t\in (i\dd j]\}$, so $\cost(\A')\le \cost(\A)$.
    Furthermore, $\sum_{t=1}^m (x'_t+y'_t) = j-i + \sum_{t=1}^m (x_t+y_t)$, contradicting the choice of $\A$
    as an optimal alignment maximizing $\sum_{t=1}^m (x_t+y_t)$.\qedhere
  \end{itemize}
\end{proof}

For strings $X,Y\in \Sigma^*$ and an integer $k\ge \ed(X,Y)$, we define a set \[\mtch_k(X,Y) = \bigcap_{\A \in \ga_k(X,Y)} \mtch_{X,Y}(\A)\] of \emph{common matches} of all alignments $\A\in \ga_k(X,Y)$;
note that it forms a non-crossing matching of~${X,Y}$.

\defnoncrossing*

Additionally, for a string $Z \in (\Sigma \cup \#)^*$, we define $\num(Z)$ as a string obtained from $Z$ by replacing the $i$th leftmost occurrence of $\#$ in it with a unique dummy symbol $\#_i \notin \Sigma$. 

\begin{fact}\label{fct:masking}
Let $X,Y\in \Sigma^*$. For each non-crossing matching $M$ of $X,Y$, we have $\ed(X,Y)\le \ed(\num(X^M),\num(Y^M))$.
\end{fact}
\begin{proof}
For every $(x,y)\in [1\dd |X|]\times [1\dd |Y|]$, if $\num(X^M)[x] = \num(Y^M)[y]$, then $X[x]=Y[y]$. Thus,
$\cost_{X,Y}(\A)\le \cost_{\num(X^M),\num(Y^M)}(\A)$ holds for all alignments $\A$ of $X,Y$.
\end{proof}

\lmmasking*
\begin{proof}
Denote $X' = \num(X^M)$ and $Y'=\num(Y^M)$, and let $\A\in \ga_k(X,Y)$. We shall first prove that $\mtch_{X,Y}(\A)\sub \mtch_{X',Y'}(\A)$.
Suppose that $(x,y)\in \mtch_{X,Y}(\A)$, i.e., $X[x]\simeq_{\A} Y[y]$.
Note that $\mtch_{X,Y}(\A)$ is a non-crossing matching of $X,Y$ and a superset of $M$,
so either $(x,y)\in M$, and then  $X'[x]=Y'[y]$ is a dummy symbol,
or $M$ contains no pair involving $x$ or $y$, and then $X'[x]=X[x]=Y[y]=Y'[y]$.
In both cases, we have  $X'[x]\simeq_{\A} Y'[y]$, i.e., $(x,y)\in \mtch_{X',Y'}(\A)$.
This completes the proof that $\mtch_{X,Y}(\A)\sub \mtch_{X',Y'}(\A)$,
from which we derive $\brkp_{X',Y'}(\A)\sub \brkp_{X,Y}(\A)$ and $\cost_{X',Y'}(\A)\le \cost_{X,Y}(\A)\le k$.

Next, we shall prove that $\A\in \ga_k(X',Y')$.
Suppose that $X'[x]=Y'[y]$ holds for some $(x,y)\in \A$.
Then, we have $X[x]=Y[y]$ because either $(x,y)\in M$
or $X[x]=X'[x]=Y'[y]=Y[y]$.
Since $\A$ is a greedy alignment of $X,Y$, this implies $(x,y)\in \mtch_{X,Y}(\A)$,
i.e., $X[x]\simeq_{\A} Y[y]$. 
Due to the assumption $X'[x]=Y'[y]$, we conclude that $X'[x]\simeq_{\A} Y'[y]$,
i.e., $(x,y)\in \mtch_{X',Y'}(\A)$.
This proves $\A\in \ga_k(X',Y')$.

Now, let $\A \in \ga_k(X',Y')$. We shall first prove that $\mtch_{X',Y'}(\A)\sub \mtch_{X,Y}(\A)$.
Suppose that $(x,y)\in \mtch_{X',Y'}(\A)$, i.e., $X'[x]\simeq_{\A} Y'[y]$.
In particular, $X'[x]=Y'[y]$, which implies $X[x]=Y[y]$ because either $(x,y)\in M$ or $X[x]=X'[x]=Y'[y]=Y[y]$.
Hence, $X[x]\simeq_{\A} Y[y]$, i.e., $(x,y)\in \mtch_{X,Y}(\A)$.
This completes the proof that $\mtch_{X',Y'}(\A)\sub \mtch_{X,Y}(\A)$,
from which we derive $\brkp_{X,Y}(\A)\sub \brkp_{X',Y'}(\A)$ and $\cost_{X,Y}(\A)\le \cost_{X',Y'}(\A)\le k$.

Next, we shall prove that $\A\in \ga_k(X,Y)$.
For a proof by contradiction, suppose that $X[x]=Y[y]$ holds for some $(x,y)\in \brkp_{X,Y}(\A)$; if there are several such breakpoints, let us choose the leftmost one.
Note that $X[1\dd x)\sim_{\A} Y[1\dd y)$ and $X[x\dd |X|] \sim_{\A} Y[y\dd |Y|]$.
Let us construct an alignment $\A'$ obtained from $\A$ by replacing the induced alignment $\A_{[x\dd |X|],[y\dd |Y|]}$
with an optimum greedy alignment of $X[x\dd |X|],Y[y\dd |Y|]$ (see \cref{obs:greedy_and_optimal}).
By the choice of $(x,y)$, the induced alignment $\A'_{[1\dd x),[1\dd y)}=\A_{[1\dd x),[1\dd y)}$ is a greedy alignment of $X[1\dd x),Y[1\dd y)$ and thus $\A'$ is a greedy alignment of $X,Y$.
Due to $\cost_{X,Y}(\A')\le \cost_{X,Y}(\A) \le k$, this implies $\A'\in \ga_k(X,Y)$.
Hence, $\mtch_{X,Y}(\A')$ is a non-crossing matching of $X,Y$ and a superset of $M$.
The construction of $\A'$ further guarantees $(x,y)\in \mtch_{X,Y}(\A')$,
so either $(x,y)\in M$, and then  $X'[x]=Y'[y]$ is a dummy symbol,
or $M$ contains no pair involving $x$ or~$y$, and then $X'[x]=X[x]=Y[y]=Y'[y]$.
In both cases, we have $X'[x]=Y'[y]$ and, since $\A$ is a greedy alignment of $X',Y'$,
we derive $X'[x]\simeq_{\A} Y'[y]$ and $(x,y)\in \mtch_{X,Y}(\A)$.
This contradicts the choice of $(x,y)$, completing the proof that $\A \in \ga_k(X,Y)$.

Overall, we conclude that $\ga_k(X,Y)=\ga_k(X',Y')$ and that every alignment $\A$ in this family satisfies both $\mtch_{X,Y}(\A)=\mtch_{X',Y'}(\A)$ and $\cost_{X,Y}(\A)=\cost_{X',Y'}(\A)$. Consequently, $\mtch_k(X,Y)=\mtch_k(X',Y')$ and, due to \cref{obs:greedy_and_optimal}, $\ed(X,Y)=\ed(X',Y')$.
\end{proof}

\newcommand{\RS}{\mathsf{RS}}
\begin{proposition}\label{prop:rank}
Consider a string $S\in (\Sigma\cup\{\#\})^n$ with $s$ dummy segments.
Given the sorted list of dummy segments in $S$, one can in $\Oh(s)$ time construct a data structure $\RS_\#(S)$ supporting the following queries in $\Oh(\log s)$ time: 
\begin{enumerate}
  \item $\mathbf\rank_\#(S,i)$: Given $i\in [1\dd n+1]$, return $|\{i' \in [1\dd i) : S[i'] = \#\}|$.
 \item $\mathbf\select_\#(S,j)$: Given $j\in [1\dd \rank_\#(S,n)]$, return the $j$th smallest element of $\{i \in [1\dd n] : S[i] = \#\}$.
\end{enumerate}
\end{proposition}
\begin{proof}
For each dummy segment $S[\ell\dd r)$, the data structure stores a tuple $(\ell,r,\rank_\#(S,\ell))$.
These tuples are stored in a sorted array (note that the order is the same for all coordinates).
To compute $\rank_\#(S,i)$, we binary search for the rightmost segment $S[\ell^*\dd r^*)$ such that $\ell^*\le i$.
If there is no such segment, we return $\rank_\#(S,i)=0$.
Otherwise, $\rank_\#(S,i)=\rank_\#(S,\ell^*) + \min(i,r^*)-\ell^*$.

To compute $\select_\#(S,j)$, we binary search for the rightmost segment $S[\ell^*\dd r^*)$ such that
$\rank_\#(S,\ell^*) < j$. Then, $\select_\#(S,j) = \ell^* + j-1 - \rank_\#(S,\ell^*)$.
\end{proof}

\begin{definition}\label{def:masked_encoding}
For a non-crossing matching $M$ of strings $X,Y\in \Sigma^*$
we set (cf.~\cref{prp:RLSLP}):  \[\Enc^{M}(X,Y):=\left(\CR(X^MY^M), \RS_\#(X^M), \RS_\#(Y^M) \right).\]
\end{definition}

\begin{lemma}\label{lm:eLCE}
  Let $X' = \num(X^M)$ and $Y' = \num(Y^M)$ for a non-crossing matching $M$ of strings $X,Y\in \Sigma^{\le n}$.
  The encoding $\Enc^{M}(X,Y)$ allows answering $\LCE$ queries
  on the suffixes of $X',Y'$ and on the suffixes of $\rev{X'},\rev{Y'}$
  in $\Oh(\log n)$ time.
\end{lemma}
\begin{proof}
  Consider a query asking for $\LCE(X'[x\dd ],Y'[y\dd ])$. Observe that 
  $\LCE(X'[x\dd ],Y'[y\dd ])=\LCE(X^M[x\dd ],Y^M[y\dd ])$ if $\rank_\#(X^M,x)=\rank_\#(Y^M,y)$.
  Otherwise, $\LCE(X'[x\dd ],Y'[y\dd ])=\min(\LCE(X^M[x\dd ],Y^M[y\dd ]), \select_\#(X^M,\rank_\#(X^M,x)+1)-x, \select_\#(Y^M,\rank_\#(Y^M,y)+1)-y)$.
  In either case, the query can be answered in $\Oh(\log n)$ time using \cref{prop:rank,prp:RLSLP}\ref{it:lce}.
  The values $\LCE(\rev{X'}[x\dd ],\rev{Y'}[y\dd ])$ are answered in a similar way.
\end{proof}

\begin{corollary}\label{cor:ED}
Given an integer $k> 0$ and the encoding $\Enc^{M}(X,Y)$ for a non-crossing matching $M$ of strings $X,Y\in \Sigma^{\le n}$, one can in $\Oh(k^2\log n)$ time compute a integer $d\in [0\dd k+1]$ such that
\[ d = \begin{cases}
  \ed(X,Y) & \text{if }\ed(X,Y)\le k\text{ and }M\sub \mtch_k(X,Y),\\
  k+1 & \text{if }\ed(X,Y)>k.
\end{cases}\]
\end{corollary}
\begin{proof}
We compute $d:=\min(k+1,\ed(\num(X^M),\num(Y^M)))$ using the Landau--Vishkin algorithm~\cite{LandauV86},
with $\Oh(k^2)$ $\LCE$ queries on suffixes of $\num(X^M),\num(Y^M)$ implemented using \cref{lm:eLCE}.
Due to \cref{fct:masking}, we have $d \ge \min(k+1,\ed(X,Y))$ and, in particular, $d=k+1$ if $\ed(X,Y)>k$. 
If $\ed(X,Y)\le k$ and $M\sub \mtch_k(X,Y)$, then \cref{lm:masking_does_not_change_ga}
yields $d = \min(k+1,\ed(X,Y))=\ed(X,Y)$.
\end{proof}

\subsection{Greedy Encoding and Its Size}
\begin{definition}\label{def:greedy_encoding}
	For strings $X,Y\in \Sigma^*$ and an integer $k$, we define the \emph{greedy encoding} 
	\[\gr_k(X,Y)=\begin{cases}
	  \Enc^{\mtch_k(X,Y)}(X,Y) & \text{if }k\ge \ed(X,Y),\\
	  \bot & \text{otherwise.}\end{cases}\]
\end{definition}

\lmgreedysize*
  \begin{proof}
  We show the claim of the lemma for $X^M$, the claim for $Y^M$ follows by symmetry. 
  
  As $\ed(X,Y) \le k$, by \cref{obs:greedy_and_optimal} $\ga_k(X,Y) \neq \emptyset$, and therefore there is an alignment $\A \in \ga_k(X,Y)$ of cost at most $k$ between $X$ and $Y$. We consider yet another graphical representation of an alignment. Namely, we represent $\A$ as a set of at most $k$ horizontal segments, where a horizontal segment $I = [i \dd j; \Delta]$ from $(i, \Delta)$ to $(j, \Delta)$ means that $X[i \dd j] \simeq_\A Y[i+\Delta \dd j + \Delta]$ (see Fig.~\ref{fig:alignments_horizontal}). This representation induces a partitioning of $X = f_1 \cdots f_z$, where $z = \Oh(k)$, and each factor $f_\ell$ is either a single character deleted under $\A$ or a fragment $X[i \dd j]$ corresponding to a horizontal segment $I = [i \dd j; \Delta]$. 
  
   \begin{figure}[ht!]
  \begin{center}
  \begin{tikzpicture}[scale=0.575]
  \begin{scope}
    \draw (1,0.7) node[above] {$X$:};
      \foreach \x/\c in {1/a,2/b,3/b,4/a,5/b,6/c,7/b,8/c,9/b,10/c}{
        \draw (\x+2,-0.7) node[above] {\tt \c};
        \draw (\x+2,-0.3) node[above] {\tiny{\color{gray} \x}};
      } 
      
      \draw (1,-0.6) node[above] {$Y$:};
      \foreach \x/\c in {1/a,2/b,3/a,4/b,5/a,6/b,7/c,8/b,9/c,10/b,11/c}{
        \draw (\x+2,0.7) node[above] {\tt \c};
        \draw (\x+2,1.1) node[above] {\tiny{\color{gray} \x}};
      }
      
      \draw[red,thick] (3,0.5)--(3,0.2);
      \draw[red,thick] (4,0.5)--(4,0.2);
      \draw[red,thick] (5,0.2)--(6,0.5);
      \draw[red,thick] (6,0.2)--(7,0.5);
      \draw[red,thick] (7,0.2)--(8,0.5);
      \draw[red,thick] (8,0.2)--(9,0.5);
      \draw[red,thick] (9,0.2)--(10,0.5);
      \draw[red,thick] (10,0.2)--(11,0.5);    
      \draw[red,thick] (11,0.2)--(12,0.5);
      \draw[red,thick] (12,0.2)--(13,0.5);
          
      \draw (7,-1.5) node {\small{Alignment $\A$.}}; 
   \end{scope}
  
  \begin{scope}[xshift = 14cm]
    \draw (1,0.7) node[above] {$X$:};
      \foreach \x/\c in {1/a,2/b,3/b,4/a,5/b,6/c,7/b,8/c,9/b,10/c}{
        \draw (\x+2,-0.7) node[above] {\tt \c};
        \draw (\x+2,-0.3) node[above] {\tiny{\color{gray} \x}};
      } 
      
      \draw (1,-0.6) node[above] {$Y$:};
      \foreach \x/\c in {1/a,2/b,3/a,4/b,5/a,6/b,7/c,8/b,9/c,10/b,11/c}{
        \draw (\x+2,0.7) node[above] {\tt \c};
        \draw (\x+2,1.1) node[above] {\tiny{\color{gray} \x}};
      }
      
      \draw[blue,thick] (3,0.5)--(3,0.2);
      \draw[blue,thick] (4,0.5)--(4,0.2);
      \draw[blue,thick] (6,0.2)--(5,0.5);
      \draw[blue,thick] (7,0.2)--(6,0.5);   
      \draw[blue,thick] (9,0.2)--(8,0.5);
      \draw[blue,thick] (10,0.2)--(9,0.5);     
      \draw[blue,thick] (11,0.2)--(10,0.5);    
      \draw[blue,thick] (12,0.2)--(11,0.5);        
          
      \draw (7,-1.5) node {\small{Alignment $\A''$.}}; 
   \end{scope}

  \begin{scope}[yshift=-5cm,xshift=2.5cm]
  \draw[thin,->] (0,0)--(24,0) node[below,pos=0.97] {$|X|$};
  \draw[thin,->] (0,-2.5)--(0,2.5) node[left,pos=0.9] {shift};
  \draw[thin] (-0.1,-1)--(0.1,-1) node[left] {$-1$}; 
  
  \foreach \x in {1,2,...,11} {
    \draw[thin] (2*\x,-0.05)--(2*\x,0.05) node[below] {\small{\color{gray} $\x$}};
  };
  
  \draw[thick,red] (2,0)--(4,0);
  \draw[thick,red] (8,-1)--(22,-1) node[below,pos=0.15] {\color{black}{\small $I = [4 \dd 11;-1]$}};
  
  \draw[thick,blue] (2,0.1)--(4,0.1);
  \draw[thick,blue] (6,1)--(8,1);
  \draw[thick,blue] (12,1)--(18,1);
  
  \end{scope}
  \end{tikzpicture}
  \end{center}
  
  \caption{Graphical representations of two greedy alignments $\A, \A'' \in \ga_5(X,Y)$, where $X = \mathtt{abababcbcbc}$ and $Y = \mathtt{abbabcbcbc}$. The alignment $\A$ (red) has cost $1$ (we delete $X[3]$), and the alignment $\A''$ cost $5$ (we delete $X[5]$, $X[10]$, $X[11]$, $Y[3]$, $Y[6]$). The alignments imply that $X[6 \dd 9] = Y[7\dd 10] = X[8 \dd 11]$.}
  \label{fig:alignments_horizontal}
  \end{figure}
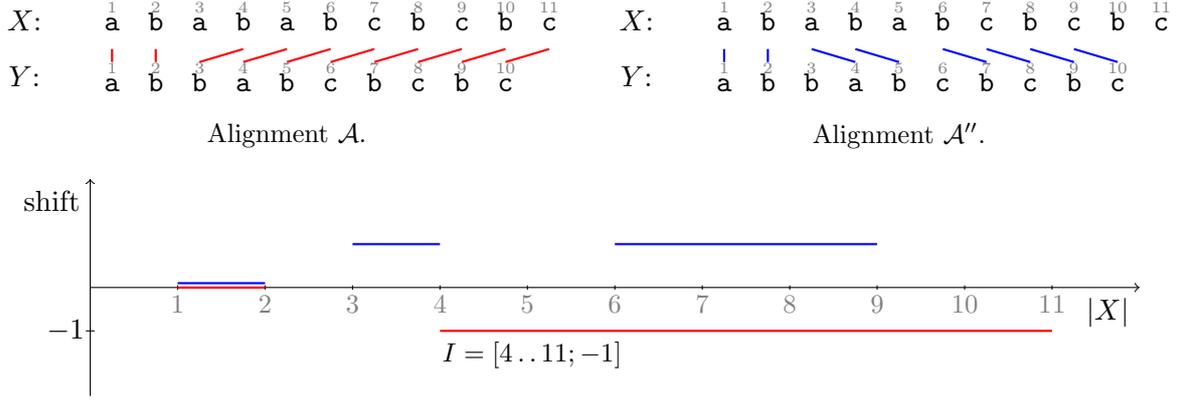
  
  If $|\ga_k(X, Y)| = 1$, then every factor $f_\ell$ with $|f_\ell| \ge 2$ is replaced with an equal length string of dummy symbols. The upper bound on $|\LZ(X^M)|$ follows by sub-additivity of the Lempel--Ziv encoding, and the upper bound on the number of dummy segments is trivial. 
  
  Suppose now that $|\ga_k(X, Y)| > 1$. We  show that in this case the Lempel--Ziv encoding of any factor $f_\ell$ has size $\Oh(k)$. Fix $\ell$. If $f_\ell$ is a single character, the claim obviously holds. Suppose now that the factor $f_\ell$ of the partitioning defined above corresponds to a horizontal segment $I = [i \dd j; \Delta]$. Let us show that there is $i \le q \le j$ such that $X^M[i \dd q] = X[i \dd q]$ and $X^M[q+1\dd j] = \#\#\ldots\#$. This follows immediately from the following observation. Let $\A' \neq \A$ be another alignment $\in \ga_k(X, Y)$. The alignment $\A'$ is greedy, and therefore, if for some $i \le p \le j$ the alignment $\A'$ matches $X[p]$ and $Y[p+\Delta]$, then it also matches $X[p+1]$ and $Y[p+1+\Delta]$, \ldots, $X[j]$ and $Y[j+\Delta]$. We obtain an upper bound on the number of dummy segments as an immediate corollary, but an upper bound on $\LZ(X^M)$ requires more work.

  Consider an alignment $\A'' \in \ga_k(X, Y)$ realising $q$. We have $X^M[q+1 \dd j] = \# \# \dd \#$, and therefore $|\LZ(X^M[q+1 \dd j])| = \Oh(1)$. Consider now $X^M[i \dd q] = X[i \dd q]$. Consider the set of the horizontal segments corresponding to $\A''$ that contain positions in $[i \dd q]$. We claim that either all segments in the set have height less than $\Delta$, or all segments have height larger than $\Delta$. Suppose that there are two consecutive horizontal segments $[i_1 \dd j_1;\Delta_1]$ and $[i_2 \dd j_2;\Delta_2]$ such that one of them is above $I$ and the other is below. $\A''$ must delete the characters $X[j_1+1], X[j_1+2], \ldots, X[i_2-1]$ and $Y[j_1+\Delta_1], Y[j_1+\Delta_1+1], \ldots, Y[i_2+\Delta_2-1]$ in some order. As it deletes the characters one by one, at some moment it arrives to a pair $X[x], Y[x+\Delta]$, where $i < x < q$. However, since $\A''$ is greedy, from this pair it must follow the segment $I$, a contradiction. 
  
  We now use this property to show that $|\LZ(X[i \dd q])| = \Oh(k)$. 
  First, consider the case when the segments corresponding to $\A''$ are below $I$. Let $I' = [i' \dd j'; \Delta']$ be one such segment, where $i \le i' \le j' \le q$, and $-k \le \Delta' < \Delta$. Let $i'' = \max\{i+2k, i'\}$. If $i'' \le j'$, we have $X[i'' \dd j'] = Y[i''+\Delta' \dd j'+\Delta']$, and the alignment $\A$ implies that $Y[i''+\Delta' \dd j'+\Delta'] = X[i''+(\Delta'-\Delta) \dd j'+(\Delta'-\Delta)]$. Consequently, $X[i'' \dd j'] = X[i''+(\Delta'-\Delta) \dd j'+(\Delta'-\Delta)]$ is a previous factor as $\Delta'-\Delta < 0$. It follows that $\A''$ defines a partitioning of $X[i\dd q] = t_1 \cdots t_z$, where $z = \Oh(k)$ and each factor $t_i$ is either a single character, or a previous factor. By Fact~\ref{fct:lz_properties}, we obtain that $|\LZ(X[i \dd q])| = \Oh(k)$. 
  
  The proof for the complementary case is similar. There are at most $k$ horizontal segments corresponding to $\A'$ that contain positions in $[i\dd q]$, denote them by $[i'_\ell \dd j'_\ell; \Delta'_\ell]$. For all $\ell$, we have $k \ge \Delta'_\ell > \Delta$. Let $j''_\ell = \min\{j'_\ell, q-2k\}$. If $i'_\ell \le j''_\ell$, we have $X[i'_\ell\dd j''_\ell] = Y[i'_\ell+\Delta'_\ell \dd j''_\ell+\Delta'_\ell]$ (from the alignment $\A''$), and $Y[i''_\ell+\Delta'_\ell \dd j''_\ell+\Delta'_\ell] = X[i'_\ell+(\Delta'_\ell-\Delta) \dd  j''_\ell+(\Delta'_\ell-\Delta)]$. Consequently, $X[i'_\ell+(\Delta'_\ell-\Delta) \dd j''_\ell+(\Delta'_\ell-\Delta)] = X[i'_\ell \dd j''_\ell]$ is a previous factor as $\Delta'_\ell-\Delta > 0$. As the cost of $\A''$ is bounded from above by $k$, we have $\bigl|[i \dd q] - \cup_\ell [i'_\ell \dd j''_\ell] \bigr | \le k$ and $\sum_\ell |\Delta'_\ell-\Delta'_{\ell+1}| \le k$. Therefore, $\bigl|[i \dd q] - \cup_\ell [i'_\ell+(\Delta'_\ell-\Delta) \dd j''_\ell+(\Delta'_\ell-\Delta)) \bigr | = \Oh(k)$, and we can conclude as above. 
  
  The claim follows.
  \end{proof}

  \begin{corollary}\label{cor:greedy_size}
    For every $X,Y\in \Sigma^{\le n}$ and $k\in \Zp$,
    the encoding $\gr_k(X,Y)$ takes $\Oh(k^2 \log^2 n)$ space.
  \end{corollary}
  \begin{proof}
    Follows immediately from \cref{prp:RLSLP,prop:rank,lm:greedy_size}.
  \end{proof}

\subsection{Construction of greedy encodings}\label{sec:greedyalg}
In this section, we give a deterministic approach to construction of $\gr_k(X,Y)$. We start with a combinatorial lemma that is a key to this approach.

\lmgreedyalg*
\begin{proof}[Proof\hspace{.1cm}\protect\footnotemark\hspace{-.1cm}]\footnotetext{The authors would like to thank Panagiotis Charalampopoulos for an early write-up of the proof of this lemma, which was originally meant to be included into~\cite{DBLP:journals/corr/abs-2004-08350}.}
We prove $\brkp_k(X,Y)\le 136k^5$ by induction on $|X|$.
In the base case of $|X|\le 42k^4$, we have $|\brkp_k(X,Y)|\le (2k+1)|X|\le 136k^5$ since each $(x,y) \in \brkp_k(X,Y)$ satisfies $|x-y|\le k$.
Thus, we assume that $|X|>42k^4$, $\ga_k(X,Y)\ne \emptyset$ (otherwise, $\brkp_k(X,Y)=\emptyset$),
and the claim holds for shorter strings.

\emph{Case 1}. If $\mtch_k(X,Y)$ contains two adjacent pairs $(x-1,y-1),(x,y)$, then
we construct strings $X^\ast$ and $Y^\ast$ by deleting $X[x]$ and $Y[y]$, respectively,
aiming to prove that $|\brkp_k(X,Y)|\le |\brkp_k(X^\ast,Y^\ast)|\le 136k^5$.
For this, consider an alignment $\A \in \ga_k(X,Y)$. Due to $(x-1,y-1),(x,y)\in \mtch(\A)$,
the pairs $(x-1,y-1)$, $(x,y)$, and $(x+1,y+1)$ are three consecutive elements of $\A$. Let $X^\ast$ be the string obtained by removing $X[x]$ from $X$, and $Y^\ast$ be the string obtained by removing $Y[y]$ from $Y$. 
We define an alignment $\A^\ast$ of $X^\ast,Y^\ast$ by removing $(x,y)$ and shifting all pairs to the right of $(x,y)$
by one position to the left. 
Each breakpoint $(x',y')\in\brkp(\A)$ is preserved (along with $X[x']\ne Y[y']$), if it is located to the left of $(x,y)$, or shifted by 1 position to the left (along with $X[x']\ne Y[y']$), if it is located to the right of $(x,y)$.
Consequently, $\A^\ast \in \ga_k(X^\ast,Y^\ast)$ and, since the shifts are independent of $\A$,
we derive $|\brkp_k(X^\ast,Y^\ast)|\ge |\brkp_k(X,Y)|$.

\emph{Case 2:} If $\mtch_k(X,Y)$ does not contain any two adjacent pairs,
we first prove the following claim:

\begin{claim}\label{claim:2align}
There are alignments $\A,\A'\in \ga_k(X,Y)$ such that 
$X[a \dd a+9k^2) \simeq_{\A} Y[i \dd i+9k^2) \simeq_{\A'} X[b \dd b+9k^2)$ holds for some positions $a \ne b$.
\end{claim}
\begin{proof} 
Consider $k+1$ disjoint fragments of $X$, each of length $21k^3 \le |X| / (k+1)$. 
Each $\A\in \ga_k(X,Y)$ aligns at least one fragment without mismatches: $X[s \dd t] \simeq_\A Y[s' \dd t']$ for some positions $s,t,s',t'$.
Due to $X[t]=Y[t']$, if $(t-1,t'-1)\in \mtch_k(X,Y)$, then $(t,t')\in \mtch_k(X,Y)$ holds by the greedy nature of the considered alignments.
Consequently, there is $\A' \in \ga_k(X,Y)$ such that $X[t-1] \not\simeq_{\A'} Y[t'-1]$.

Next, consider $k+1$ disjoint parts of $Y[s+k \dd t-k)$, each of length $9k^2\le (21k^3-1-2k)/(k+1)$.
By \cref{fact:edk2}, $\A$ aligns the entire $Y[s+k\dd t-k)$ without mismatches to a fragment of $X[s\dd t)$.
Moreover, $\A'$ aligns at least one part without mismatches. Denote this part $Y[i\dd i+9k^2)$,
and fix positions $a,b$ so that $X[a \dd a+9k^2) \simeq_{\A} Y[i \dd i+9k^2) \simeq_{\A'} X[b \dd b+9k^2)$.
Note that $a\in [s\dd t)$ and thus $X[a\dd t)\simeq_{\A} Y[i\dd t')$.
Now, if $a=b$, then we would have $X[a]\simeq_{\A'} Y[i]$ and the greedy nature of $\A'$ would guarantee $X[a\dd t)\simeq_{\A'}Y[i\dd t')$, contradicting $X[t-1] \not\simeq_{\A'} Y[t'-1]$.
Consequently, $a\ne b$.
\end{proof}

By \cref{fact:edk2}, we have $p := \per(Y[i\dd i+9k^2))\leq |a-b|\le |i-a|+|i-b|\leq 2k$. 
Moreover, due to $7k^2\ge 2p$,  the fragment $X[i+k^2 \dd i+8k^2)$ has the same period (up to cyclic rotation).

\begin{claim}\label{claim:nobreaks}
For every $\A'' \in \ga_k (X,Y)$, we have  $X[i+3k^2 \dd i+8k^2) \simeq_{\A''} Y[v \dd v+5k^2)$ for some $v \in [i+2k^2\dd i+4k^2]$.
\end{claim}
\begin{proof}
Consider $k+1$ disjoint parts of $X[i+k^2\dd i+3k^2 + p)$, each of length $p$ (note that $(k+1)p \le 2k^2+p$).
Each $\A''\in \ga_k(X,Y)$ matches at least one part without mismatches:
$X[j\dd j+p)\simeq_{\A''} Y[u\dd u+p)$ for some $j\in [i+k^2\dd i+3k^2]$ and $u\in [j-k\dd j+k]$.
Note that $X[j\dd i+8k^2)$ and $Y[u\dd i+9k^2)$ have, up to cyclic rotation,
the same string period of length~$p$.
This string period is primitive, so ${X[j\dd j+p)}=Y[u\dd u+p)$ implies 
that the periods synchronize, i.e., $X[j\dd i+8k^2)=Y[u\dd u+(i+8k^2-j)]$.
By the greedy nature of $\A''$, we must have $X[j\dd i+8k^2)\simeq_{\A''} Y[u\dd u+(i+8k^2-j)]$,
and therefore $X[i+3k^2\dd i+8k^2) \simeq_{\A''} Y[v\dd v+5k^2]$
for $v = u + i+3k^2-j \in [i+3k^2-k\dd i+3k^2+k]$.
\end{proof}

Consider fragments $Q:=X[i+3k^2 \dd i+8k^2)$ and 
$F:=Y[i+4k^2\dd i+7k^2)$, where $F$ is contained in all fragments $Y[v \dd v+5k^2)$ with $v\in [i+2k^2\dd i+4k^2]$.
We observe that $|Q|\ge |F| = 3k^2 > p$ and construct $X^\ast$ and $Y^\ast$ by removing $X[i+8k^2-p\dd i+8k^2)$
and $Y[i+7k^2-p\dd i+7k^2)$. 
For every $\A''\in \ga_k(X,Y)$, we derive $\A^\ast$ by removing the pairs $(x,y)$
with $y\in [i+7k^2-p\dd i+7k^2)$, and shifting all the subsequent pairs by $p$ positions to the left.
We conclude from \cref{claim:nobreaks} that the breakpoints $(x',y')\in \brkp(\A)$ are preserved (along with $X[x']\ne Y[y']$), if located to the left of $(i+3k^2,i+4k^2)$,
or shifted by $p$ positions to the left (along with $X[x']\ne Y[y']$), if located to the right of $(i+8k^2,i+7k^2)$. 
Consequently, $\A^\ast\in \ga_k(X^\ast,Y^\ast)$ and, since the shifts are independent of $\A''$,
we derive $|\brkp_k(X,Y)|\le |\brkp_k(X^\ast,Y^\ast)|\le 136k^5$ from the inductive assumption.
\end{proof}

Algorithm~\ref{alg:greedy_memoization} is the key procedure for computing greedy encodings. 

\newcommand{\extractmin}{\textsf{extract-min}}
\newcommand{\credit}{\textsf{credit}}
\begin{algorithm}
\newcommand\mycommfont[1]{\footnotesize\ttfamily\textcolor{blue}{#1}}
\SetCommentSty{mycommfont}
\SetAlgoNoLine
    
    \SetKwBlock{Begin}{}{end}
    \SetKwFunction{greedyalg}{Greedy}
	  \SetKwFunction{greedyprocedure}{GreedyMatch}
	  \SetKwFunction{push}{Push}

    $\greedyalg{k}$
    \Begin{
        $Q = \emptyset$\;
		\greedyprocedure{$Q$, $1$, $1$, $k$}\;\label{line:initialisation}	
    	\While{$Q$ is not empty}{
    		$(x,y), \credit = Q.\extractmin()$\;
           \If{$x \le |X|$}{
			     \greedyprocedure{$Q$, $x+1$, $y$, $\credit-1$}\tcp*{delete $X[x]$}\label{line:deletex}
		  } 
		  \If{$y \le |Y|$}{
			    \greedyprocedure{$Q$, $x$, $y+1$, $\credit-1$}\tcp*{delete $Y[y]$}\label{line:deletey}
		  }
		  \If{$x \le |X|$ \KwSty{and} \KwSty{and} $y \le |Y|$}{
			\greedyprocedure{$Q$, $x+1$, $y+1$, $\credit-1$}\tcp*{substitute $X[x]$ for $Y[y]$}\label{line:substitute}
		}
		}

    \BlankLine
    
    \greedyprocedure{$Q$, $x$, $y$, $\credit$}:
    \Begin{
    	\If{$\ed(X[x\dd], Y[y\dd]) \le \credit$}{
        	$\ell = \LCE(X[x\dd], Y[y\dd])$\tcp*{match $X[x]$ and $Y[y]$,\ldots, $X[x+\ell-1]$ and $Y[y+\ell-1]$} \label{line:match}
			$x = x + \ell$; $y = y + \ell$\;
			\push{$Q$, $(x, y)$, $\credit$}\;
			\lIf{$|Q| = 1$}{output $[x-\ell, x-1]$}
		}
	}
	
	\BlankLine
    
	\push{$Q$, $(x,y)$, $\credit$} 
	\Begin{
	  \lIf{$(x,y) \notin Q$}{$Q[(x,y)] = \credit$} 
	  \lElse(\tcp*[f]{points are ordered lexicographically.}){$Q[(x,y)] = \max\{Q[(x,y)], \credit\}$}
	}
   
  \caption{The algorithm receives as an input strings $X,Y $, and integer $k$. Lexicographic order on points is defined in the following way: $(x_1, y_1) < (x_2, y_2)$ if either $x_1 < x_2$ or $x_1 = x_2$ and $y_1 < y_2$.}\label{alg:greedy_memoization}}
\end{algorithm}

\begin{lemma}\label{lm:greedy_memoization}
Given two strings $X, Y\in \Sigma^\ast$ and integer $k\ge \ed(X,Y)$, \cref{alg:greedy_memoization} computes the dummy segments in $X^{\mtch_k(X,Y)}$.
Moreover, it can be implemented in $\Oh(|\brkp_k(X,Y)| (t + \log k))$ time and $\tOh(k)$ space,
provided with an oracle that, given $x\in [1\dd |X|+1]$ and $y\in [1\dd |Y|+1]$, in $\Oh(t)$ time
 computes $\LCE(X[x\dd],Y[y\dd])$ and $\min(k+1,\ed(X[x\dd], Y[y\dd]))$.
\end{lemma}
\begin{proof}
To show correctness of Algorithm~\ref{alg:greedy_memoization}, consider an (imaginary set) of alignments $\G$ generated by the algorithm. First, associate with each point in $[1\dd |X|] \times [1\dd |Y|]$ an empty set of alignments. Let $(x,y)$ be the latest point extracted from $Q$ and $\G[(x,y)]$ the set of alignments associated with it. In line~\ref{line:deletex} (if the algorithm executes it), we create a new set of alignments containing all alignments in $\G[(x,y)]$ extended by $(x+1,y)$ (deletion of $X[x]$), $(x+1,y), (x+2,y+1), \ldots (x+\ell,y+\ell-1)$, where $\ell = \LCE(X[x+1\dd], Y[y\dd])$ (matching of $X[x+1\dd x+\ell]$ and $Y[y \dd y+\ell-1]$). We union the resulting set of alignments with $\G[x+\ell,y+\ell-1]$ and push $(x+\ell,y+\ell-1)$ to $Q$. In lines~\ref{line:deletey} and~\ref{line:substitute} we process $\G[(x,y)]$ analogously. Finally, we define $\G = G[|X|,|Y|]$.  We claim that $\G = \ga_k(X,Y)$. First note that all alignments in $\G$ are greedy by construction: an alignment can only delete $X[x]$, $Y[y]$, or substitute $X[x]$ for $Y[y]$ if $X[x] \ne Y[y]$. In other words, for every breakpoint $(x,y)$ we have $X[x] \ne Y[y]$. Additionally, the cost of every alignment in $\G$ is at most $k$ as at each call of $\greedyprocedure$ $\credit$ decreases by one. Hence, $\G = \ga_k(X,Y)$.

We define successors and predecessors in the lexicographic order in a standard way. For an alignment $\A \in \G$ and a point $(x,y) \in \brkp(\A)$, define $\credit_\A (x,y)$ as the difference between $k$ and the number of edits that $\A$ makes to reach $(x,y)$.
Algorithm~\ref{alg:greedy_memoization} satisfies the following properties:
\begin{enumerate}
\item \label{property1} The set $\mathcal{P}$ of the points that are ever added to $Q$ is equal to $\brkp(\G)$, and when we process a point $(x,y)$, the credit associated with it equals $\max_{\A \in \G : (x,y) \in \brkp(\A)} \credit_\A (x,y)$;
\item \label{property2} When we process a point $(x,y)$, $Q = \cup_{\A \in \G} \{(x',y')  \text{ is the successor of } (x,y) \text{ in } \brkp(\A)\}$.
\item \label{property3} If $Q$ contains points $p_1 = (x_1,y_1)$ and $p_2 = (x_2, y_2)$ such that $x_1-y_1 = x_2-y_2$, then $p_1 = p_2$.
\end{enumerate}

To show the properties, we exploit the fact that the points are processed in the lexicographical order: When we process a point $q \in Q$, it is the lexicographically smallest point in $Q$, and all the points that we add to $Q$ while processing $q$ are larger than it. 

\underline{Property~\ref{property1}.} We first prove that $\mathcal{P} \supseteq \brkp(\G)$ and that when we process a point $(x,y)$, the credit associated with it is at least $\max_{\A \in \G : (x,y) \in \brkp(\G)} \credit_\A (x,y)$. 
Let $(x,y) \in \brkp(\G)$ be the lexicographically smallest point which either does not belong to $\mathcal{P}$ or such that the credit associated with it is smaller than $\max_{\A \in \G : (x,y) \in \brkp(\A)} \credit_\A (x,y)$. Consider an alignment $\A \in \G$ such that $(x,y) \in \brkp(\A)$ and let $(x',y')$ be the predecessor of $(x,y)$ in $\brkp(\A)$. When we process $(x',y')$, which happens before we process $(x,y)$ as we process the points in the lexicographic order, the credit associated with it is at least $\credit_\A (x',y') \ge \ed(X[x'\dd], Y[y'\dd])$, and therefore $(x,y)$ is added to $Q$ with credit at least $\credit_\A (x',y')-1=\credit_\A(x,y)$. As it holds for any $\A \in \G$, we obtain a contradiction. 

We now show the opposite direction. Consider a point $(x,y) \in Q$ associated with the edit distance credit $c$. Suppose that $c$ is achieved when we process a point $(x',y')$. By the induction assumption, there is a greedy alignment $\A \in \G$ such that the credit $c+1$ associated with $(x',y')$ equals $\credit_\A(x',y')$. Consider a greedy alignment $\A'$ that behaves as $\A$ until $(x',y')$ ($k-\credit_\A(x',y')$ edits), then proceeds to $(x',y')$ (one edit), and ends as an arbitrary optimal greedy  alignment between $X[x'\dd]$ and $Y[y'\dd]$. The cost of the latter is $\ed(X[x'\dd],Y[y'\dd]) \le c$. Therefore, the cost of $\A'$ is at most $k$ and it is in $\G$. It follows that $(x',y')\in \brkp(\A')$ and that $c = \credit_{\A'}(x',y')$, finishing the proof.

\underline{Property~\ref{property2}.} Let $(x,y)$ be the lexicographically smallest point in $\brkp(\G)$ such that when we process it $Q \ni (u,v) \notin \cup_{\A \in \G} \{(x',y')  \text{ is the successor of } (x,y) \text{ in } \brkp(\A)\}$. Suppose that $(u,v)$ was added to $Q$ when we processed a point $(u',v')$. Similarly to above, we can construct an alignment $\A \in \G$ such that $(u',v'), (u,v) \in \brkp(\A)$. As we processed $(u',v')$ before $(x,y)$, we must have that $(u',v')$ is lexicographically smaller than $(x,y)$, while $(u,v)$ is lexicographically larger than $(x,y)$, and therefore $(u,v)$ is the successor of $(x,y)$ in $\brkp(\A)$, a contradiction. On the other hand, if $(u,v)  \text{ is the successor of } (x,y) \text{ in } \brkp(\A)$ for some $\A \in \G$, then it must be in $Q$ when we process $(x,y)$: if $(u',v')$ is the predecessor of $(u,v)$ in $\brkp(\A)$, then $(u',v')$ is at most $(x,y)$, and when we process it, we add $(u,v)$ to $Q$.
 
\underline{Property~\ref{property3}.} We show the property by induction. It is true at initialisation. Suppose that the property is violated when we process a point $(x,y) \in Q$. Recall that we extract $(x,y)$ and call \greedyprocedure for a subset of points $\{(x+1,y+1), (x+1,y), (x,y+1)\}$. The point $(x+1,y+1)$ is in the same diagonal as $(x,y)$ that we extract from $Q$, and therefore the call for it does not violate the property. Consider now the point $(x+1,y)$ (resp., $(x,y+1)$) and assume that the diagonal containing it also contains a point $(x',y') \in Q$. If we added $(x',y')$ to $Q$ when we called \greedyprocedure for a point $(x'',y'')$, we have that $(x'',y'')$ is in the same diagonal, $X[x'' \dd x'-1] = Y[y''\dd y'-1]$, and $X[x'] \neq Y[y']$. As we process the points in the lexicographic order, at least one of the points $(x''-1,y''-1)$, $(x''-1,y'')$, $(x'',y''-1)$ is strictly smaller than $(x,y)$. By considering each of the cases, from the definition we obtain that $(x'',y'')$ is at most $(x+1,y)$ (resp., $(x,y+1)$). As all three points $(x'',y'')$, $(x+1,y)$ (resp., $(x,y+1)$), $(x',y')$ are on the same diagonal, it follows that \greedyprocedure for $(x+1,y)$ outputs $(x',y')$.

\underline{Correctness and complexity.} $[x,x']$ is a dummy segment in $X^{\mtch_k(X,Y)}$ iff there exists $-k \le \delta \le k$ such that the following three conditions are satisfied, where $y = x+\delta$ and $y' = x'+\delta$:
\begin{enumerate}
\item $\brkp(\G)$ does not contain any points with the $x$-coordinate in $[x \dd x']$ or the $y$-coordinate in~$[y \dd y']$;
\item $\brkp(\G)$ contains $q = (x'+1,y'+1)$;
\item The predecessor $p$ of $q$ in $\brkp(\G)$ is one of the points $(x-1,y-1), (x-1,y), (x,y-1)$. 
\end{enumerate}

Therefore, if $[x,x']$ is a maximal dummy segment, then by Property~\ref{property2}, after we have processed $p$, $Q$ contains a single point $q$ that is obtained by calling \greedyprocedure for the point $(x,y)$. Hence, the implementation creates the dummy segment $[x,x']$. On the other hand, if the implementation creates a segment $[x,x']$, then $\brkp(\G)$ satisfies all three conditions and hence $[x,x']$ is a maximal dummy segment.

Algorithm~\ref{alg:greedy_memoization} implements the set $Q$ as a binary search tree. 
By~\cref{lm:greedyalg} and Property~\ref{property1}, we process $|\brkp(\G)| = |\brkp_k(X,Y)|$ points, spending $\Oh(\log k + t)$ time per point, which gives the desired time complexity. To show the space complexity, it suffices to show that at any moment $|Q| = \Oh(k)$. By Fact~\ref{fact:edk2}, for any $(x,y) \in Q$ we have $y = x+\delta$ for some $-k \le \delta \le k$, in other words, the point belongs to one of $2k+1$ diagonals of the grid.  By Property~\ref{property3}, each of the diagonals contains at most one point from $Q$, which concludes the proof.
\end{proof}

We proceed with an auxiliary claim that allows constructing an edit distance oracle. 

\begin{claim}\label{claim:LMS}
Given two strings $U,V$, and a data structure of size $s$ that can answer the $\LCE$ queries on the suffixes of $\rev{U}, \rev{V}$ in $t$ time. We can build in $\Oh(k^2 t)$ time a data structure that occupies $\Oh(k^2)$ space and can retrieve $\min(k+1,\ed(U[u\dd], V[v\dd]))$ in $\Oh(\log n)$ time.
\end{claim}
\begin{proof}
The data structure is based on the Landau--Myers--Schmidt algorithm~\cite{LMS}. Consider a table $D$ of size $|U| \times |V|$ such that $D[i,j] = \ed(U[|U|-i+1\dd], V[|V|-j+1\dd]) = \ed(\rev{U}[1\dd i], \rev{V}[1 \dd j])$. By Fact~\ref{fact:edk2}, if $D[i,i+\delta] \le h$, then $-h \le \delta \le h$. 
Let $L^h(\delta) = \max\{i: D[i, i+\delta] \le h\}$. The data structure consists of $2k+1$ sorted arrays containing $L^0(\delta), L^1(\delta), \ldots, L^k(\delta)$, $-k \le \delta \le k$. 
The arrays occupy $\Oh(k^2)$ space and allow to retrieve $\ed(U[u\dd], V[v\dd])$ in $\tOh(1)$ time by simple binary search. 
The data structure can be computed via the following recurrence, where $\textsc{Slide}_\delta(u) = u+\LCE(\rev{U}[u\dd], \rev{V}[u+\delta\dd])$. 

\begin{equation}\label{eq:LMS}
L^h(\delta) = \textsc{Slide}_\delta \left( \max
\begin{cases}
L^{h-1}(\delta-1), & \delta > -h;\\
L^{h-1}(\delta)+1, & \text{always};\\
L^{h-1}(\delta+1)+1, & \delta < h.\\
\end{cases}
 \right)
\end{equation}

Therefore, given $L^{h-1}$, $L^h(\delta)$ can be computed via three $\LCE$ queries on the suffixes of $\rev{U}, \rev{V}$. As a corollary, the arrays $\{L^0(\delta), L^1(\delta), \ldots, L^k(\delta)\}$, $-k \le \delta \le k$, can be built in $\tOh(k^2)$ time.
\end{proof}

\begin{proposition}\label{prp:greedify}
Consider a non-crossing matching $M$ of strings $X,Y\in \Sigma^{\le n}$ and an integer $k\in \Zp$
such that $M\sub \mtch_k(X,Y)$ holds if $\ed(X,Y)\le k$.
Given $k$ and $\Enc^M(X,Y)$, the greedy encoding $\gr_k(X,Y)$ can be computed in $\tOh(zk+k^2 + |\brkp_k(X,Y)|)$ time and $\tOh(z+k^2)$ space, where $z = |\LZ(X^MY^M)|$.
\end{proposition}
\begin{proof}
First, we pass $k$ and $\Enc^M(X,Y)$ to the procedure of \cref{cor:ED}.
Observe that the returned value $d$ satisfies $d\le k$ if and only if $\ed(X,Y)\le k$.
If $\ed(X,Y)>k$, then we simply return $\gr_k(X,Y)=\bot$. In this case, the running time and the space complexity are $\tOh(k^2)$.

Otherwise, our strategy is to compute $\mtch_k(X,Y)$ and then mask out the corresponding characters of $X^M$ and $Y^M$
to obtain $X^{\mtch_k(X,Y)}$ and $Y^{\mtch_k(X,Y)}$.
By \cref{lm:masking_does_not_change_ga}, we have $\mtch_k(X,Y)=\mtch_k(X',Y')$, where $X' = \num(X^M)$ and $Y'=\num(Y^M)$. 
Hence, we shall use \cref{lm:greedyalg} to compute the dummy segments corresponding to $\mtch_k(X',Y')$.
For this, we need to support $\LCE$ queries and edit distance queries on the suffixes of $X',Y'$.
As for $\LCE$ queries, we rely on \cref{lm:eLCE}, which provides $\Oh(\log n)$-time $\LCE$ queries on the suffixes of $X',Y'$ and on the suffixes of $\rev{X'},\rev{Y'}$.
We use the latter queries to build a component of \cref{claim:LMS} for the edit distance queries.
After $\Oh(k^2 \log n)$-time preprocessing, these queries can be answered in $\Oh(\log k)$ time.
Overall, constructing the dummy segments corresponding to $\mtch_k(X',Y')$ costs $\tOh(k^2 + |\brkp_k(X,Y)|)$ time and $\Oh(k^2)$ working space.

Our next goal is to convert $\CR(X^MY^M)$ to $\CR(X^{\mtch_k(X,Y)}Y^{\mtch_k(X,Y)})$.
For this, we need to place $\#$s within each of the computed dummy segments.
Consider updating a working string $Z$ by setting $Z[\ell\dd r):=\#^{r-\ell}$.
To implement this operation, we build $\LZ(\#^{r-\ell})$ (of size at most 2),
 derive $\CR(\#^{r-\ell})$ (\cref{prp:RLSLP}\ref{it:fromLZ}),
 extract $\CR(Z[1\dd \ell))$ and $\CR(Z[r\dd |Z|])$ (\cref{prp:RLSLP}\ref{it:extract}),
 and finally concatenate into $\CR(Z[1\dd \ell)\#^{r-\ell}Z[r\dd |Z|])$ (\cref{prp:RLSLP}\ref{it:concat}).
 Overall, each dummy segment is processed in $\Oh((z+k)\log^4 n)$ time and space,
 for a total of $\tOh(zk+k^2)$ time and $\tOh(z+k)$ space across all segments.

Finally, $\RS_\#(X^{\mtch_k(X,Y)})$ and $\RS_\#(Y^{\mtch_k(X,Y)})$ are constructed in $\Oh(k)$ time and space using \cref{prop:rank}.
\end{proof}

\begin{corollary}\label{cor:greedy_short}
  Given an integer $k\in \Zp$ and two strings $X, Y \in \Sigma^{\le k^2}$, one can compute $\gr_k(X, Y)$ in $\tOh(k^3)$ time and $\tOh(k^2)$ space.
  \end{corollary}
  \begin{proof}
  We shall construct $\Enc^\emptyset(X,Y)$ and then derive $\gr_k(X,Y)$ from \cref{prp:greedify}.
  Hence, we build a trivial LZ-like representation of $XY$ (with length-1 phrases)
  and derive $\CR(XY)=\CR(X^\emptyset Y^\emptyset)$ using \cref{prp:RLSLP}\ref{it:fromLZ}.
  This costs $\Oh(k^2\log^2 k)$ time and space.
  The components $\RS_\#(X^\emptyset)$ and $\RS_\#(Y^\emptyset)$ are trivial (there are no dummy segments).
  Using  \cref{prp:greedify} costs $\tOh(k^3+|\brkp_k(X,Y)|)$ time and $\tOh(k^2)$ space. To finish the proof, note that $\brkp_k(X,Y)$ can only contain pairs $(x,y)$ such that $1\le x \le |X|$, $1 \le y \le |Y|$, and $|x-y| \le k$, which implies $|\brkp_k(X,Y)| = \Oh(|X|k) = \Oh(k^3)$.
  \end{proof}

\subsection{Concatenations and Quasi-greedy Alignments}\label{sec:quasi}

\begin{definition}[Quasi-greedy alignment]\label{def:quasi_greedy_alignment}
  We say that an alignment $\A$ of two strings $X, Y\in \Sigma^*$ is \emph{quasi-greedy} if it satisfies at least one of the following symmetric conditions:
  \begin{itemize}
    \item there exists $\delta \in [1\dd |X|+1]$ such that $(\delta,1)\in \A$ and $X[x] \neq Y[y]$ holds for every $(x,y) \in \brkp(\A) \cap ([\delta\dd|X|]\times [1\dd |Y|])$;
    \item there exists $\delta \in [1\dd |Y|+1]$ such that $(1,\delta)\in \A$ and $X[x] \neq Y[y]$ holds for every $(x,y) \in \brkp(\A) \cap ([1\dd|X|]\times [\delta\dd |Y|])$;
  \end{itemize}
Given $k\ge 0$, we denote by $\qga_k(X,Y)$ the set of all quasi-greedy alignments $\A$ of $X,Y$ with $\cost(\A)\le k$.
\end{definition}

\begin{lemma}\label{lem:concat}
Consider an alignment $\A\in \qga_k(X_pX_s,Y_pY_s)$, where $X_p,X_s,Y_p,Y_s\in \Sigma^*$ and $k\in \Zz$.
Then, $\A_{[1\dd |X_p|+1),[1\dd |Y_p|+1)}\in \qga_{k+||X_s|-|Y_s||}(X_p,Y_p)$
and $\A_{[|X_p|+1\dd ],[|Y_p|+1\dd ]}\in \qga_{k+||X_p|-|Y_p||}(X_s,Y_s)$.
\end{lemma}
\begin{proof}
For brevity, let $X = X_p X_s$, $Y = Y_p Y_s$; $n_p = |X_p|$, $n_s = |X_s|$, $n = |X|$, $m_p = |Y_p|$, $m_s = |Y_s|$, $m = |Y|$. 

Let $\A=(x_t,y_t)_{t=1}^q$ and let $(x_{i_p},y_{i_p})$ be the leftmost pair $(x,y)\in \A$ such that $x>n_p$ or $y>m_p$.
By symmetry, we assume without loss of generality that $x_{i_p} = n_p+1$;
see \cref{fig:alignment_repartition}.
Observe that $\A_p:=\A_{[1\dd n_p+1),[1\dd m_p+1)}$ is the union of $(x_t,y_t)_{t=1}^{i_p-1}$ and $(n_p+1,y)_{y=y_{i_p}}^{m_p+1}$. 
Furthermore, $\brkp(\A_p)= \{(x_t,y_t)\in \brkp(\A) : t\in [1\dd p)\}\cup \{(n_p+1,y) : y\in [y_p\dd m_p+1]\}$.
Let $k_p$ be the size of the former component.
Then, $k\ge \cost(\A) \ge k_p + \ed(X_s,Y_p[y_{i_p}\dd]Y_s) \ge k_p + |Y_p[y_{i_p}\dd]|+m_s-n_s$,
so that $\cost(\A_p) = k_p + |Y_p[y_p\dd ]| \le k +n_s - m_s \le k - \big||X_s|-|Y_s|\big|$.
It remains to prove that $\A_p$ is a quasi-greedy alignment.
By \cref{def:greedy_alignment} applied to $\A$, we have two possibilities:
\begin{itemize}
  \item There exists $\delta\in [1\dd n+1]$ such that $(\delta,1)\in \A$ and $X[x]\ne Y[y]$ holds for every $(x,y)\in \brkp(\A)\cap ([\delta\dd n]\times [1\dd m])$.
  If $\delta \le n_p+1$, then we also have that $(\delta,1)\in \A_p$ and $X_p[x]\ne Y_p[y]$ holds for every $(x,y)\in \brkp(\A_p)\cap ([\delta\dd n_p]\times [1\dd m_p])$.
  Otherwise, $(x_{i_p},y_{i_p})=(n_p+1,1)\in \A_p$, so $\A_p$ is trivially quasi-greedy.
  \item There exists $\delta\in [1\dd m+1]$ such that $(1,\delta)\in \A$ and $X[x]\ne Y[y]$ holds for every $(x,y)\in \brkp(\A)\cap ([1\dd n]\times [\delta \dd m])$.
  If $\delta \le y_{i_p}$, then we also have that $(1,\delta)\in \A_p$ and $X_p[x]\ne Y_p[y]$ holds for every $(x,y)\in \brkp(\A_p)\cap ([1\dd n_p]\times [\delta\dd m_p])$.
  Otherwise, $n_p=0$ because $(1,\delta)$ cannot cross $(n_p+1,y_p)$.
  Hence, $\A_p$ is trivially quasi-greedy.
\end{itemize}

Now, let $(x_{i_s},y_{i_s})$ be the leftmost pair $(x,y)\in \A$ such that $x>n_p$ and $y>m_p$.
By symmetry, we assume without loss of generality that $y_{i_s}= m_p+1$;
see \cref{fig:alignment_repartition}.
Observe that $\A_s := \A_{[n_p+1\dd ],[m_p+1\dd ]}$ is the union of  $(x,1)_{x=1}^{x_{i_s}-n_p-1}$ and $(x_t-n_p,y_t-m_p)_{t={i_s}}^q$.
Furthermore, $\brkp(\A_s)= \{(x_t-n_s,y_t-m_s)\in \brkp(\A) : t\in [i_s\dd q]\}\cup \{(x,1) : x\in [1\dd x_{i_s}-n_p)\}$.
Let $k_s$ be the size of the former component.
Then, $k\ge \cost(\A) \ge k_s-1 + \ed(X_pX_s[1\dd x_{i_s}-n_p),Y_p) \ge k_s-1 + x_{i_s}-m_p$,
so that $\cost(\A_s) = k_s-1 + x_{i_s}-n_p\le k + m_p-n_p \le k - \big||X_p|-|Y_p|\big|$.
It remains to prove that $\A_s$ is a quasi-greedy alignment.
By \cref{def:greedy_alignment} applied to $\A$, we have two possibilities:
\begin{itemize}
  \item There exists $\delta\in [1\dd n+1]$ such that $(\delta,1)\in \A$ and $X[x]\ne Y[y]$ holds for every $(x,y)\in \brkp(\A)\cap ([\delta\dd n]\times [1\dd m])$.
  If $\delta \le x_{i_s}$, then $(x_{i_s}-n_p,1)\in \A_s$ and $X_s[x]\ne Y_s[y]$ holds for every $(x,y)\in \brkp(\A_s)\cap ([x_{i_s}-n_p\dd n_s]\times [1\dd m_s])$.
  Otherwise, $(\delta-n_p,1)\in \A_s$ and $X_s[x]\ne Y_s[y]$ holds for every $(x,y)\in \brkp(\A_s)\cap ([\delta-n_p\dd n_s]\times [1\dd m_s])$.
  \item There exists $\delta\in [1\dd m+1]$ such that $(1,\delta)\in \A$ and $X[x]\ne Y[y]$ holds for every $(x,y)\in \brkp(\A)\cap ([1\dd n]\times [\delta \dd m])$.
  If $\delta \le m_p+1$, then $(x_{i_s}-n_p,1)\in \A_s$ and $X_s[x]\ne Y_s[y]$ holds for every $(x,y)\in \brkp(\A_s)\cap ([x_{i_s}-n_p\dd n_s]\times [1\dd m_s])$.
  Otherwise, $x_{i_s} = 1$ and $n_p=0$ because $(1,\delta)$ cannot cross $(x_{i_s},m_p+1)$.
  Hence, $(1,\delta-m_p)\in \A$ and $X_s[x]\ne Y_s[y]$ holds for every $(x,y)\in \brkp(\A_s)\cap ([1\dd n_s]\times [\delta-m_p \dd m_s])$.\qedhere
\end{itemize}
\end{proof}

  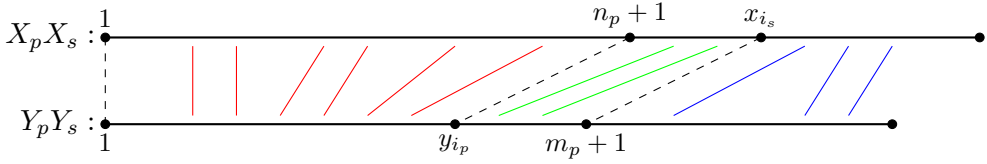
\begin{figure}[th]
  \begin{center}
  \begin{tikzpicture}[scale=0.575]
  \draw[thick] (0,0)--(20,0);
  \filldraw[fill=black] (12,0) circle (0.1) node[above] {\small{$n_p+1$}};
  \filldraw[fill=black] (15,0) circle (0.1) node[above] {\small{$x_{i_s}$}};
  \filldraw[fill=black] (0,0) circle (0.1) node[above] {\small{$1$}};
  \filldraw[fill=black] (20,0) circle (0.1);
  \draw (0,0)  node[left] {$X_pX_s:$};
  
  \draw[thick] (0,-2)--(18,-2);
  \filldraw[fill=black] (0,-2) circle (0.1) node[below] {\small{$1$}};
  \filldraw[fill=black] (11,-2) circle (0.1) node[below] {\small{$m_p+1$}};
  \filldraw[fill=black] (8,-2) circle (0.1) node[below] {\small{$y_{i_p}$}};
  \filldraw[fill=black] (18,-2) circle (0.1);
  \draw (0,-2)  node[left] {$Y_pY_s:$};
  
  \draw[dashed] (12,0)--(8,-2);
  \draw[dashed] (15,0)--(11,-2);
  \draw[dashed] (0,0)--(0,-2);
  
  \draw[red] (2,-0.2)--(2,-1.8);
  \draw[red] (3,-0.2)--(3,-1.8);
  \draw[red] (5,-0.2)--(4,-1.8);
  \draw[red] (6,-0.2)--(5,-1.8);
  \draw[red] (8,-0.2)--(6,-1.8);
  \draw[red] (10,-0.2)--(7,-1.8);
  
  \draw[green] (13,-0.2)--(9,-1.8);
  \draw[green] (14,-0.2)--(10,-1.8);
  
  \draw[blue] (16,-0.2)--(13,-1.8);
  \draw[blue] (17,-0.2)--(16,-1.8);
  \draw[blue] (18,-0.2)--(17,-1.8);
  \end{tikzpicture}
  \end{center}
  
  \caption{Alignment $\A \in \qga_k(X_pX_s,Y_pY_s)$.}\label{fig:alignment_repartition}
  \end{figure} 

  For strings $X,Y\in \Sigma^*$ and an integer $k\ge \ed(X,Y)$, we define a set \[\qmtch_k(X,Y) = \bigcap_{\A \in \qga_k(X,Y)} \mtch_{X,Y}(\A)\] of \emph{common matches} of all alignments $\A\in \qga_k(X,Y)$.

  \begin{definition}
    For strings $X,Y\in \Sigma^*$ and an integer $k$, we define the \emph{quasi-greedy encoding} 
    \[\qgr_k(X, Y) = \begin{cases} \Enc^{\qmtch_k(X,Y)}(X,Y)&\text{if }\ed(X,Y)\le k,\\
      \bot &\text{otherwise.}\end{cases}\]
  \end{definition}

  \begin{corollary}\label{cor:qgr_to_gr}
    Consider strings $X,Y \in \Sigma^{\le n}$. Define $X' = \$_1 X$ and $Y' = \$_2 Y$, where $\$_1 \neq \$_2$ are special symbols not in $\Sigma$. Let $M' = \mtch_{k+1}(X',Y')$ and $M = \qmtch_k(X,Y)$. The following properties hold:
\begin{enumerate}[label=\textrm{(\alph*)}]
\item Given the dummy segments for one of the strings $X^{M}, (X')^{M'}$, the dummy segments for the other string can be constructed in $\Oh(k)$ time and space;\label{it:dummy} 
\item  Given one of the encodings $\qgr_k(X, Y), \gr_{k+1}(X', Y')$, the other encoding can be constructed in $\Oh(k^2 \log^4 n)$ time and space;\label{it:encodings}
\item The quasi-greedy encoding $\qgr_k(X, Y)$ occupies $\Oh(k^2 \log^4 n)$ space. \label{it:qga_space} 
\end{enumerate}
\end{corollary}
\begin{proof}
    Let us first show that $\A' \mapsto \A'_{[2\dd],[2\dd]}$ bijectively maps $\ga_{k+1}(X',Y')$ to $\qga_k(X,Y)$. 
    
    Consider a greedy alignment $\A' \in \ga_{k+1}(X', Y')$ and let $\A=\A'_{[2\dd],[2\dd]}$. \cref{lem:concat} yields $\A\in \qga_{k+1}(X,Y)$. However, we actually have $\cost(\A) \le k$
    because $\mtch(\A') = \{(x,y) : (x+1,y+1) \in \mtch(\A)\}$ holds due to $X'[1]\ne Y'[1]$. 
    
    We now show the opposite direction. Let $\A \in \qga_{k}(X, Y)$. Suppose that there is $\delta \in [1\dd |Y|]$ such that $(1,\delta) \in \A$ and $X[x] \neq Y[y]$ holds for every $(x,y) \in \brkp(\A) \cap ([1\dd|X|] \times [\delta \dd |Y|])$ (the other case is symmetrical). It follows in particular that $\A$ contains elements $\{(1,1), (1,2), \ldots, (1,\delta)\}$. We start by replacing each element $(x,y)$ of $\A$ with $(x+1,y+1)$. We add an element $(1,1)$ to $\A$ and replace each element $(2,i)$, $2\le i \le \delta+1$, with $(1,i)$. Finally, we add an element $(2,\delta+1)$. We claim that the resulting alignment $\A'$, that we treat as an alignment of $X'$ and $Y'$, is greedy. For every element $(x,y) \in \brkp(\A') \cap  ([2\dd|X'|] \times [\delta+1 \dd |Y'|])$ we have $(x-1,y-1) \in \brkp(\A) \cap ([1\dd|X|] \times [\delta \dd |Y|])$ and hence from quasi-greediness of $\A$ we obtain $X'[x] = X[x-1] \ne Y[y-1] = Y'[y]$. For every element $(1,i) \in \A'$ we have $X'[1] = \$_1 \neq Y[i]$. It follows that $\A'$ is greedy. In addition, the cost of $\A'$ equals the cost of $\A$ plus one and hence is bounded by $k+1$. Finally, note that $\mtch(\A') = \{(x,y) : (x-1,y-1) \in \mtch(\A)\}$. 
    
Let $M' = \mtch_{k+1}(X',Y')$ and $M = \qmtch_k(X,Y)$. From above, we obtain $M' = \{(x+1,y+1) : (x,y) \in M\}$. Therefore, $(X')^{M'}[2\dd] = X^M$ and $(Y')^{M'}[2\dd] = Y^M$. 
    
It follows that we can obtain the dummy segments for $X^M$ by subtracting one from the endpoints of each dummy segment of $(X')^{M'}$, and analogously for $Y^M$. The reverse claim is obtained analogously, by adding one to the endpoints. 
As the number of the dummy segments in $X^M$ and $(X')^{M'}$ is $\Oh(k)$ by Lemma~\ref{lm:greedy_size}, \ref{it:dummy} follows. 

From~\ref{it:dummy} it follows that given the rank data structures for one pair of strings $X^M, Y^M$ and $(X')^{M'}, (Y')^{M'}$, the rank data structures for the other pair can be built in $\Oh(k)$ time and space. We now must show that given the $\CR$ data structure for one of the two strings $X^M Y^M, (X')^{M'} (Y')^{M'}$, we can construct the $\CR$ data structure for the other string efficiently. Assume that we are given $\CR((X')^{M'}(Y')^{M'})$, the other case is analogous. We extract all non-dummy segments using \cref{prp:RLSLP}\ref{it:extract} and concatenate them using \cref{prp:RLSLP}\ref{it:concat}. In total, it takes $\Oh(k^2 \log^4 n)$ time and space, giving \ref{it:encodings}. 
    
\ref{it:qga_space} follows immediately from \ref{it:encodings}.
\end{proof}

From~\cref{cor:greedy_short,cor:qgr_to_gr} it immediately follows that:

\begin{claim}\label{claim:short_strings_qga}
  Given an integer $k\in \Zp$ and strings $X, Y \in \Sigma^{\le k^2}$, one can compute $\qgr_k(X, Y)$ in $\tOh(k^3)$ time and $\tOh(k^2)$ space.
\end{claim}

\begin{corollary}\label{cor:qgreedify}
  Consider a non-crossing matching $M$  of strings $X,Y\in \Sigma^{\le n}$ and an integer $k\in \Zp$
  such that $M\sub \qmtch_k(X,Y)$ if $\ed(X,Y)\le k$. 
  Given $k$ and $\Enc^M(X,Y)$, the quasi-greedy encoding $\qgr_k(X,Y)$ can be computed in $\tOh(k^5+zk)$ time and $\tOh(k^2+z)$ space, where $z = |\LZ(X^M Y^M)|$.
  \end{corollary}
  \begin{proof}
    Let $\$_1\ne \$_2$ be auxiliary symbols not in $\Sigma\cup\{\#\}$.
    We construct $M'=\{(x+1,y+1): (x,y)\in M\}$ and $\Enc^{M'}(\$_1X,\$_2Y)$.
    Note that $M'$ is a non-crossing matching of $\$_1X,\$_2Y$.
    Moreover, if $M\sub \qmtch_k(X,Y)$, then $M'\sub \mtch_{k+1}(\$_1X,\$_2Y)$ holds
    by \cref{cor:qgr_to_gr}. Hence, we can construct $\gr_{k+1}(\$_1X,\$_2Y)$ using \cref{prp:greedify}. 
    Then, we derive $\qgr_k(X,Y)$ using \cref{cor:qgr_to_gr} again. By~\cref{lm:greedyalg}, we have $|\brkp_{k+1}(\$_1X,\$_2Y)| = \Oh(k^5)$, which gives the desired time bound.
  \end{proof}
  
\begin{observation}\label{obs:qgr_larger}
Let $k' \le k$. Given $\qgr_{k}(X,Y)$, the encoding $\qgr_{k'}(X,Y)$ can be computed in $\tOh(k^5)$ time and $\tOh(k^2)$ space.
\end{observation} 
\begin{proof}
For $k' \le k$, we have $\qmtch_k(X,Y) \sub \qmtch_{k'}(X,Y)$ by definition. The claim follows from~\cref{cor:qgreedify}. 
\end{proof}

We now show that the quasi-greedy encodings are \emph{concatenatable}.

\begin{lemma}\label{lm:greedy_concatenation}
Consider strings $X_p, Y_p, X_s, Y_s \in \Sigma^{\le n}$ and $k\in \Zp$. Assume $\max\{\bigl| |X_p|-|Y_p| \bigr|, \bigl| |X_s|-|Y_s| \bigr| \} =\Oh(k)$. Given $\qgr_{k+||X_s|-|Y_s||}(X_p, Y_p)$ and $\qgr_{k+||X_p|-|Y_p||}(X_s, Y_s)$, one can compute $\qgr_k(X_pX_s, Y_pY_s)$ in $\tOh(k^5)$ time and $\tOh(k^2)$ space.
\end{lemma}
\begin{proof}
For brevity, let $M_p = \qmtch_{k+||X_s|-|Y_s||}(X_p, Y_p)$ and $M_s = \qmtch_{k+||X_p|-|Y_p||}(X_s, Y_s)$. Let also $X = X_pX_s$, $Y=Y_pY_s$.  

By~\cref{lem:concat}, if $\ed(X,Y)\le k$, then $M:=M_p \cup \{(x+|X_p|,y+|Y_p|) : (x,y)\in M_s\} \sub \qmtch_k(X, Y)$. Hence, we shall construct $\Enc^{M}(X,Y)$ and then apply \cref{cor:qgreedify}.

For this, we extract $\CR(X_p^{M_p}), \CR(X_s^{M_s}), \CR(Y_p^{M_p}), \CR(Y_s^{M_s})$ using \cref{prp:RLSLP}\ref{it:extract}, and then concatenate them to $\CR(X^MY^M)=\CR(X_p^{M_p}X_s^{M_s}Y_p^{M_p}Y_s^{M_s})$ using \cref{prp:RLSLP}\ref{it:concat}. Overall, this takes $\Oh(k^2\log^4 n)$ time.
Then, we use \cref{prop:rank} to build $\RS_\#(X^M)$ and $\RS_\#(Y^M)$ in $\tOh(k)$ time.

Finally, we note that using \cref{cor:qgreedify} costs $\tOh(k^5)$ time and $\tOh(k^2)$ space
(here we use the fact that $\bigl| |X_p|-|Y_p| \bigr| = \Oh(k)$).
\end{proof}

Finally, as a corollary we derive an algorithm that can compute the quasi-greedy encoding of arbitrarily long strings. 

\begin{corollary}\label{cor:greedy_long}
Assuming constant-time random access to a string $X \in \Sigma^{\ell}$ and streaming access to a string $Y \in \Sigma^{\ell}$, where $\ell \le n$, there is an algorithm that constructs $\qgr_k(X,Y)$ with a delay of at most $k^2$ characters in $\tOh(k^3)$ amortised time per character and $\tOh(k^2)$ space. The delay means that at the moment when the $y$-th character of $Y$ arrives, the algorithm knows $\qgr_k(X[\dd y'], Y[\dd y'])$, where $|y-y'| \le k^2$.
\end{corollary}
\begin{proof}
If $\ell \le k^2$, construct $\qgr_k(X, Y)$ via~\cref{claim:short_strings_qga} in $\tOh(k^3)$ (total) time and $\tOh(k^2)$ space. Otherwise, partition the strings into non-overlapping blocks $X = X_1 \cdots X_p$ and $Y = Y_1  \cdots Y_p$ so that $|X_i| = |Y_i| = k^2$ for all $i\in [1\dd p)$ and $|X_p| = |Y_z| = \ell \bmod k^2$.
Suppose that we have computed  $\qgr_{k}(X_1 \cdots X_{i-1}, Y_1 \cdots Y_{i-1})$ for $i\in [1\dd p)$. Compute $\qgr_{k}(X_1 \cdots X_{i}, Y_1 \cdots Y_{i})$ in the following manner: 
first, compute $\qgr_{k}(X_{i}, Y_{i})$ in $\Oh(k^5)$ time and $\Oh(k^2)$ space via~\cref{claim:short_strings_qga}, and then compute $\qgr_{k}(X_1 \cdots X_{i}, Y_1 \cdots Y_{i})$ in $\tOh(k^5)$ time and $\tOh(k^2)$ space via \cref{lm:greedy_concatenation}.
\end{proof}

\subsection{Products of Greedy Alignments}\label{sec:product}
\defalignmentsproduct*

\lmgreedyproduct*
\begin{proof}
We proceed by induction on $|X|+|Y|+|Z|$. In the base case, when at least one of the strings $X,Y,Z$ is empty, we set $\A^{X,Y}$ and $\A^{Y,Z}$ to be any greedy optimal alignments of $X,Y$ and $Y,Z$, respectively,
so that $\cost(\A^{X,Y})=\ed(X,Y)$ and $\cost(\A^{Y,Z})=\ed(Y,Z)\le \ed(X,Y)+\ed(X,Z) \le \ed(X,Y)+\cost(\A^{X,Z})$.
Moreover, it easy to check that $\A^{X,Z}$ is a product of $\A^{X,Y}$ and $\A^{Y,Z}$ because
two out of these three alignments simply delete all characters of the non-empty string.

In the inductive step, we assume that all strings $X,Y,Z$ are non-empty, and we consider several cases.
\begin{enumerate}
\item $\mathbf{X[1] = Z[1] = Y[1]}$. \\We recurse on $X'=X[2\dd]$, $Y'=Y[2\dd]$, $Z'=Z[2\dd]$, and $\A^{X',Z'
}=\A^{X,Z}_{[2\dd],[2\dd ]}$, which is greedy due to $X[1]\simeq_{\A^{X,Z}} Z[1]$. 
This yields greedy alignments $\A^{X',Y'}\!\!, A^{Y',Z'}$ of cost at most $d' = 2\cost(\A^{X,Z})+\ed(X,Y)=d$.
We extend them so that $X[1]\simeq_{\A^{X,Y}} Y[1]$ and $Y[1]\simeq_{\A^{Y,Z}} Z[1]$,
obtaining alignments of cost up to $d$.

\item $\mathbf{X[1] = Z[1] \neq Y[1]}$.\\ In this case, we have $\ed(X,Y)>\min(\ed(X[2\dd],Y[2\dd]),\ed(X[2\dd],Y),\ed(X,Y[2\dd]))$.
\begin{enumerate}[label=(\alph*)]
 \item If $\ed(X,Y)>\ed(X[2\dd],Y[2\dd])$, we recurse on $X'=X[2\dd]$, $Y'=Y[2\dd]$, $Z'=Z[2\dd]$, and $\A^{X',Z'}=\A^{X,Z}_{[2\dd ],[2\dd ]}$, which is greedy due to $X[1]\simeq_{\A^{X,Z}} Z[1]$.
This yields greedy alignments $\A^{X',Y'}\!\!, A^{Y',Z'}$ of cost at most $d' = 2\cost(\A^{X,Z})+\ed(X,Y)-1 = d-1$.
We extend them so that $X[1]\sim_{\A^{X,Y}} Y[1]$ and $Y[1]\sim_{\A^{Y,Z}} Z[1]$, obtaining alignments of cost up to $d$.

 \item If $\ed(X,Y)>\ed(X[2\dd],Y)$, we recurse on $X'=X[2\dd]$, $Y'=Y$, $Z'=Z[2\dd]$, and $\A^{X',Z'}=\A^{X,Z}_{[2\dd ],[2\dd ]}$, which is greedy due to $X[1]\simeq_{\A^{X,Z}} Z[1]$.
 This yields greedy alignments $\A^{X',Y'}\!\!, A^{Y',Z'}$ of cost at most $d' = 2\cost(\A^{X,Z})+\ed(X,Y)-1 = d-1$.
 We extend them so that $\A^{X,Y}$ deletes $X[1]$ and $\A^{Y,Z}$ deletes $Z[1]$, obtaining alignments of cost up to $d$.

 \item If $\ed(X,Y)>\ed(X,Y[2\dd])$, we recurse on $X'=X$, $Y'=Y[2\dd]$, $Z'=Z$, and $\A^{X',Z'}=\A^{X,Z}$.
 This yields greedy alignments $\A^{X',Y'}\!\!, A^{Y',Z'}$ of cost at most $d' = 2\cost(\A^{X,Z})+\ed(X,Y)-1 = d-1$.
We extend them so that $\A^{X,Y}$ and $\A^{Y,Z}$ both delete $Y[1]$, 
obtaining alignments of cost up to $d$.
\end{enumerate}

\item $\mathbf{X[1]\ne Z[1] = Y[1]}$.
\begin{enumerate}[label=(\alph*)]
\item If $\A^{X,Z}$ deletes $X[1]$, we recurse on $X'=X[2\dd]$, $Y'=Y$, $Z'=Z$, and $\A^{X',Z'}=\A^{X,Z}_{[2\dd ],[1\dd ]}$. 
This yields greedy alignments $\A^{X',Y'}\!\!, A^{Y',Z'}$ of cost $\le d' = 2\cost(\A^{X,Z})-2+\ed(X',Y')\le d-1$. 
We derive $\A^{Y,Z}=\A^{Y',Z'}$ and extend $\A^{X',Y'}$ so that $\A^{X,Y}$ deletes $X[1]$,
obtaining an alignment of cost up to $d$.

\item If $\A^{X,Z}$ deletes $Z[1]$, we recurse on $X'=X$, $Y'=Y[2\dd]$, $Z'=Z[2\dd]$, and $\A^{X',Z'}=\A^{X,Z}_{[1\dd],[2\dd]}$. 
This yields greedy alignments $\A^{X',Y'}\!\!, A^{Y',Z'}$ of cost $\le d' = 2\cost(\A^{X,Z})-2+\ed(X',Y')\le d-1$. 
We extend them so that $\A^{X,Y}$ deletes $Y[1]$ and $Y[1]\simeq_{\A^{Y,Z}} Z[1]$, obtaining alignments of cost up to $d$.

\item If $X[1]\sim_{\A^{X,Z}} Z[1]$, we recurse on $X'=X[2\dd]$, $Y'=Y[2\dd]$, $Z'=Z[2\dd]$, and $\A^{X',Z'}=\A^{X,Z}_{[2\dd ],[2\dd ]}$.
This yields greedy alignments $\A^{X',Y'}\!\!, A^{Y',Z'}$ of cost $\le d' = 2\cost(\A^{X,Z})-2+\ed(X',Y')\le d-1$. 
We extend them so that  $X[1]\sim_{\A^{X,Y}} Y[1]$ and $Y[1]\simeq_{\A^{Y,Z}} Z[1]$,
obtaining alignments of cost up to $d$.
\end{enumerate}

\item $\mathbf{X[1] \ne Z[1] \ne Y[1]}$. (Note that this case allows both $X[1]=Y[1]$ and $X[1]\ne Y[1]$.)
\begin{enumerate}[label=(\alph*)]
\item If $\A^{X,Z}$ deletes $X[1]$, we recurse on $X'=X[2\dd]$, $Y'=Y[2\dd]$, $Z'=Z$, and $\A^{X',Z'}=\A^{X,Z}_{[2\dd ],[1\dd ]}$.
This yields greedy alignments $\A^{X',Y'}\!\!, A^{Y',Z'}$ of cost at most $d' = 2\cost(\A^{X,Z})-2+\ed(X',Y')\le d-2$. 
We extend them so that $X[1]\sim_{\A^{X,Y}} Y[1]$ and $\A^{Y,Z}$ deletes $Y[1]$, 
obtaining alignments of cost up to $d-1$.

\item If $\A^{X,Z}$ deletes $Z[1]$, we recurse on $X'=X$, $Y'=Y$, $Z'=Z[2\dd]$, and $\A^{X',Z'}=\A^{X,Z}_{[1\dd],[2\dd]}$. 
This yields greedy alignments $\A^{X',Y'}\!\!, A^{Y',Z'}$ of cost $\le d' = 2\cost(\A^{X,Z})-2+\ed(X,Y)= d-2$. 
We derive $\A^{X,Y}=\A^{X',Y'}$ and extend $\A^{Y',Z'}$ so that $\A^{Y,Z}$ deletes $Z[1]$,
obtaining an alignment of cost up to $d-1$.

\item If $X[1]\sim_{\A^{X,Z}} Z[1]$, we recurse on $X'=X[2\dd]$, $Y'=Y[2\dd]$, $Z'=Z[2\dd]$, and $\A^{X',Z'}=\A^{X,Z}_{[2\dd ],[2\dd ]}$.
This yields greedy alignments $\A^{X',Y'}\!\!, A^{Y',Z'}$ of cost $\le d' = 2\cost(\A^{X,Z})-2+\ed(X',Y')\le d-2$. 
We extend them so that $X[1]\sim_{\A^{X,Y}} Y[1]$ and $Y[1]\sim_{\A^{Y,Z}} Z[1]$, 
obtaining alignments of cost up to $d-1$.
\end{enumerate}
\end{enumerate}

In all the cases above, $\A^{X,Y}$ is greedy because $\A^{X',Y'}$ is greedy and $\A^{X,Y}$ matches $X[1]$ with $Y[1]$ whenever $X[1]=Y[1]$. 
Similarly, $\A^{Y,Z}$ is greedy because $\A^{Y',Z'}$ is greedy and $\A^{Y,Z}$ matches $Y[1]$ with $Z[1]$  whenever $Y[1]=Z[1]$.
Moreover, $\A^{X,Z}$ is a product of $\A^{X,Y}$ and $\A^{Y,Z}$ because each alignment starts with $(1,1)$ and since $\A^{X',Z'}$ is a product of $\A^{X',Y'}$ and $\A^{Y',Z'}$.
\end{proof}

Assume that we are given three strings $X, Y, Z$. Let $d = \ed(X,Y)+2k$, and define $\G_X = \gr_{d}(X,Y)$ and $\G_Z=\gr_d(Y,Z)$. We show that given $\G_X$ and $\G_Z$, we can compute an optimal alignment between $X$ and $Z$ efficiently if its cost is at most $k$. Let $M = \{(x,z) : \exists y \text{ such that } (x,y) \in \mtch_d(X,Y) \text{ and } (y,z) \in \mtch_d(Y,Z)\}$. 

\begin{lemma}\label{lm:greedy_encodings_product}
If $\ed(X,Z)\le k$, then $M \subseteq \mtch_k(X,Z)$. 
\end{lemma}
\begin{proof}
Suppose that $(x,y) \in \mtch_d(X,Y)$ and $(y,z) \in \mtch_d(Y,Z)$. 
For a proof by contradiction, suppose that $X[x] \not\sim_{\A} Z[z]$ for some $\A \in \ga_k(X,Z)$.
Then, there exists $(x',z')\in \A$ such that either $x'\le x$ and $z'>z$, or $x'>x$ and $z'\le z$.
By symmetry, without loss of generality, we consider the first alternative.
By \cref{lm:greedy_product}, $\A$ is a product $\A^{X,Y} \in \ga_d(X,Y)$ and $\A^{Y,Z} \in \ga_d(Y,Z)$.
According to \cref{def:alignments_product}, this means that there exists $y'$ such that 
$(x',y')\in \A^{X,Y}$ and $(y',z')\in \A^{Y,Z}$. If $y'\le y$, then $(y,z),(y',z')\in \A^{Y,Z}$
implies that $\A^{Y,Z}$ deletes $Z[z]$, contradicting $(y,z)\in \mtch(\A^{Y,Z})$.
Similarly, if $y'>y$, then $(x,y),(x',y')\in \A^{X,Y}$ implies that $\A^{X,Y}$ deletes $Y[y]$,
contradicting $(x,y)\in \mtch(\A^{X,Y})$.
This completes the proof that $X[x]\sim_{\A} Z[z]$ for every $\A \in \ga_k(X,Z)$.
Due to $X[x]=Y[y]=Z[z]$, we also have $X[x]\simeq_{\A} Z[z]$ for every $\A \in \ga_k(X,Z)$,
i.e., $(x,y)\in \mtch_k(X,Z)$ holds as claimed.
\end{proof}

\begin{lemma}\label{lm:RLSLPs_to_RLSLP}
If $\ed(X,Y) = \Oh(k)$ and $\G_X, \G_Z \neq \bot$, $\Enc^M(X,Z)$ can be computed in $\tOh(k^3)$ time and $\tOh(k^2)$ space.
\end{lemma}
\begin{proof}
Let us first explain how $X^M$ can be constructed.  
Consider a position $x\in [1\dd |X|]$ such that $(x,z) \notin M$ for every $z\in [1\dd |Z|]$. If $X^{\mtch_d(X,Y)}[x] \neq \#$, then $X^M[x] = X[x] = X^{\mtch_d(X,Y)}[x]$. Otherwise, $(x,y) \in \mtch_d(X,Y)$ for some $y\in [1\dd |Y|]$. By~\cref{lm:greedy_encodings_product}, we have $(y,z) \notin \mtch_d(Y,Z)$ for every $z\in [1\dd |Z|]$. Hence, $Y^{\mtch_d(Y,Z)}[y] \neq \#$, and we have $X[x] = Y[y] = Y^{\mtch_d(Y,Z)}[y]$. In words, non-dummy characters of $X^M$ can be retrieved from non-dummy characters of $X^{\mtch_d(X,Y)}$ and non-dummy characters of $Y^{\mtch_d(Y,Z)}$.

Thus, for each dummy segment $[\ell\dd r)$ in $X^{\mtch_d(X,Y)}$ and its counterpart $[\ell'\dd r')$ in $Y^{\mtch_d(X,Y)}$, we need to set $X^M[\ell \dd r):=Y^{\mtch_d(Y,Z)}[\ell'\dd r')$. Such a copy-paste operation can be implemented using \cref{prp:RLSLP}\ref{it:extract}\ref{it:concat}, in $\Oh(k^2 \log^4 n)$ time per dummy segment.
We can also keep track of the dummy segments in $X^M$ within the same procedure.
The algorithm for $Z^M$ is symmetric, and thus we can construct $\CR(X^MZ^M)$ along with the dummy segments
in $\tOh(k^3)$ time and $\tOh(k^2)$ space. Finally, we build $\RS_\#(X^M)$ and $\RS_\#(Z^M)$ in $\Oh(k)$ time.
\end{proof}

\begin{corollary}\label{cor:greedy_to_ed}
Given $\G_X$ and $\G_Z$. If $\ed(X,Y) = \Oh(k)$, then we can compute $\min(k+1,\ed(X,Z))$ in $\tOh(k^3)$ time and $\tOh(k^2)$ space.
\end{corollary}
\begin{proof}
  If $\G_X$ or $\G_Z$ equals to $\bot$, then $\ed(X,Z)>k$ by \cref{lm:greedy_product}. Otherwise, construct $\Enc^{M}(X,Z)$ using \cref{lm:RLSLPs_to_RLSLP}. Finally, we pass $k$ and $\Enc^{M}(X,Z)$ to the algorithm of \cref{cor:ED},
  and we return the resulting value $d$.

  Note that $M$ is a non-crossing matching of $X,Z$, so $d = k+1$ is correctly returned if $\ed(X,Z)>k$.
  If $\ed(X,Z)\le k$, then \cref{lm:greedy_encodings_product} implies $M\sub \mtch_k(X,Z)$,
  and thus $d = \ed(X,Z)$ holds as claimed.
  
  The overall runtime and space complexity are dominated by the procedure of \cref{lm:RLSLPs_to_RLSLP}. 
\end{proof}

\begin{remark}
Using \cref{prp:greedify} instead of \cref{cor:ED}, we could construct $\gr_k(X,Z)$.
\end{remark}

\begin{corollary}\label{cor:qgr_to_ed}
Consider three strings $X, Y, Z$. Let $d = \ed(X,Y)+2k$, and assume that we are given $\qgr_{d}(X,Y)$ and $\qgr_d(Y,Z)$. If $\ed(X,Y) = \Oh(k)$, then we can compute $\min(k+1,\ed(X,Z))$ in $\tOh(k^3)$ time and $\tOh(k^2)$ space.
\end{corollary}
\begin{proof}
We compute $\G_X = \gr_{d+1}(\$_1 X, \$_2 Y)$ and $\qgr_{d+1}(\$_2Y, \$_3 Z)$ in $\tOh(k^2)$ time and space via~\cref{cor:qgr_to_gr}, and apply~\cref{cor:greedy_to_ed} to compute $\min(k+2,\ed(\$_1X,\$_2 Z))-1=\min(k+1,\ed(X,Z))$ in $\tOh(k^3)$ time and $\tOh(k^2)$ space. 
\end{proof}

%% file: sketches.tex
\subsection{CGK embedding}\label{sec:CGK}
In this section, we prove \cref{prp:alg} based on the CGK embedding introduced in~\cite{CGK}. Recall that the Hamming distance between the embeddings of two strings $X,Y\in \Sigma^{\le n}$ is bounded in terms of the edit distance $\ed(X,Y)$, which allows using Hamming distance sketches to approximate edit distance.

\newcommand{\Hash}{\mathcal{H}}
\begin{definition}[CGK embedding~\cite{CGK}]
Consider an alphabet $\Sigma$, a sentinel character $\bot\notin\Sigma$, and a 2-independent family $\Hash$ of hash functions $h:\Sigma \to \{0,1\}$.
For an integer $n\in \Zp$, a (uniformly random) sequence $R\in \Hash^{3n}$, and a string $S\in \Sigma^{\le n}$,
the \emph{CGK walk} $W_{\CGK}(S)=(s_t)_{t=1}^{3n+1}$
and the \emph{CGK embedding} $\CGK(S)\in (\Sigma\cup\{\bot\})^{3n}$ are defined by the following algorithm:
\begin{algorithm}
	\SetKwInput{KwData}{Randomness}
    \SetAlgoNoLine
	\SetKwBlock{Begin}{}{end}
	\KwIn{An integer $n\in \Zp$, a string $S\in \Sigma^{\le n}$.}
	\KwData{A sequence $R\in \Hash^{3n}$ of 2-independent hash functions $\Sigma \to \{0,1\}$.}
	\KwOut{The CGK walk $W_{\CGK}(S)=(s_t)_{t=1}^{3n+1}$ and the CGK embedding $\CGK(S)$.}
	\vspace{.25cm}
	$s_1 := 1$\;
	\For{$t := 1$ \KwTo $3n$} {
		\If{$s_t \le |S|$}{
			$\CGK(S)[t] := S[s_t]$\;
			$s_{t+1} := s_t+R_t(S[s_t])$\;
			}
		\Else{
			$\CGK(S)[t] := \bot$\;
			$s_{t+1}:=s_t$\;
		}
	} 
	\caption{The CGK algorithm}\label{alg:CGK}
\end{algorithm}
\end{definition}

Recall from \cref{def:walk} that an $m$-step walk over $S$ is complete if $s_{m+1}=|S|+1$.
\begin{fact}[{\cite[Theorem 4.1]{CGK}}]\label{fct:walks}
	Each $S\in \Sigma^{\le n}$ satisfies $\Pr_{R}[W_{\CGK}(S)\text{ is complete}] \ge 1-e^{-\Omega(n)}$.
\end{fact}

The following result summarizes the central property of the CGK embedding.
\begin{fact}[{\cite[Theorem 4.3]{CGK}}]\label{fct:cgk}
	For every $X,Y\in \Sigma^{\le n}$ and every constant $c>0$, the embeddings $\CGK(X)$, $\CGK(Y)$
	satisfy $\Pr_{R\in \Hash^{3n}}[\hd(\CGK(X),\CGK(Y)) > c\cdot \ed(X,Y)^2] < \frac{12}{\sqrt{c}}$.
\end{fact}

If $W_{\CGK}(X)$ and $W_{\CGK}(Y)$ are complete, \cref{def:zip} yields an edit-distance alignment of $X,Y$,
which we call the \emph{CGK alignment} of $X,Y$.
that they induce an edit-distance alignment 

\begin{fact}[see also {\cite[Theorem 4.2]{CGK}}]\label{fct:greedy}
If $W_{\CGK}(X)$ and $W_{\CGK}(Y)$ are complete for some $X,Y\in\Sigma^{\le n}$ and $R\in \Hash^{3n}$,
then the CGK alignment of $X,Y$ belongs to $\ga_{\hd(\CGK(X), \CGK(Y))}(X,Y)$.
\end{fact}
\begin{proof}
Let $\A$ be the CGK alignment of $X,Y$, i.e., the zip alignment of $W_{\CGK}(X)=(x_t)_{t=1}^{3n+1}$ and $W_{\CGK}(Y)=(y_t)_{t=1}^{3n+1}$.
Consider $(x_t,y_t)\in \brkp(\A)$, with $t\in [1\dd 3n+1]$
chosen so that $t=3n+1$ or $(x_{t+1},y_{t+1})\ne (x_t,y_t)$.
If $t=3n+1$, then $(x_t,y_t)=(|X|+1,|Y|+1)$ by the assumption that $W_{\CGK}(X)$ and $W_{\CGK}(Y)$ are complete.
If $(x_{t+1},y_{t+1})=(x_t+1,y_t+1)$, then we have
$\CGK(X)[t]=X[x_t]\ne Y[y_t]=\CGK(Y)[t]$.
The remaining possibility $x_{t+1}-x_t \ne y_{t+1}-y_t$ holds only in the following three cases:
if $x_t = |X|+1$ and $y_t \le |Y|$ (when $\CGK(X)[t]=\bot \ne \CGK(Y)[t]=Y[y_t]$),
if $x_t \le |X|$ and $y_t= |Y|+1$ (when $X[x_t]=\CGK(X)[t] \ne \bot = \CGK(Y)[t]$),
or if $x_t \le |X|$, $y_t \le |Y|$, and $R_t(X[x_t])\ne R_t(Y[y_t])$
(when $\CGK(X)[t]=X[x_t]\ne Y[y_t] = \CGK(Y)[t]$).
Overall, we have $X[x_t]\ne Y[y_t]$ whenever $(x_t,y_t)\in [1\dd |X|]\times [1\dd |Y|]$, and $\CGK(X)[t]\ne \CGK(Y)[t]$ whenever $t\in [1\dd 3n]$.
Consequently, $\A\in \ga_{\hd(\CGK(X), \CGK(Y))}(X,Y)$.
\end{proof}

The final property of the CGK alignment required in this work
is that its width is within $\Oh(\ed(X,Y))$ with good probability.
For this, we first prove a fact about random walks:
\begin{fact}\label{fct:walk}
Let $(w_i)_{i\ge 0}$ be an unbiased lazy random walk (that is, $w_0 = 0$, and, for every $i\ge 1$, we have $\Pr[w_{i+1}=w_{i} \mid w_0,\ldots, w_{i}]=\frac12$, $\Pr[w_{i+1}=w_{i}+1 \mid w_0,\ldots, w_{i}]=
\Pr[w_{i+1}=w_{i}-1 \mid w_0,\ldots, w_{i}]=\frac14$).
Then, for every $m,\ell \in \mathbb{Z}_+$,
we have $\Pr[\max_{i=0}^m |w_i| \ge \ell] \le \frac{m}{\ell^2}$.
\end{fact}
\begin{proof}
By the reflection principle~\cite[Excercise 2.10]{Levin2017},
we have $\Pr[\max_{i=0}^m |w_i| \ge \ell] \le 2\Pr[|w_m| \ge \ell]
= 2 \Pr[w_m^2 \ge \ell^2]$. By Markov's inequality,
\[ 2 \Pr[w_m^2 \ge \ell^2] \le \frac{2 \Exp[w_m^2]}{\ell^2}
= \frac{2\Exp[(w_1-w_0)^2+\cdots + (w_m-w_{m-1})^2]}{\ell^2}=\frac{m}{\ell^2}.\qedhere\]
\end{proof}

\begin{corollary}\label{cor:CGK}
For every constant $\delta\in (0,1)$, there exists a constant $c$ such that for sufficiently large $n$, all strings $X,Y\in \Sigma^{\le n}$, and their CGK alignment $\A$, the probability over $R\in \Hash^{3n}$ that  $W_{\CGK}(X),W_{\CGK}(Y)\text{ are complete, } \A\in \ga_{c\cdot \ed(X,Y)^2}(X,Y)\text{, and }\width(\A)\le c\cdot\ed(X,Y)$ is at least $1-\delta$.
\end{corollary}
\begin{proof}
By \cref{fct:walks},
$\Pr[W_{\CGK}(S)\text{ is incomplete}]\le \frac\delta4$ holds for every $S\in \Sigma^{\le n}$ and sufficiently large $n = \Omega(\log\frac1\delta)$.
By \cref{fct:cgk}, there is a constant $c$ such that 
$\Pr[\hd(\CGK(X),\CGK(Y))> c\cdot \ed(X,Y)^2] \le \frac{\delta}{4}$.
We may take arbitrarily large $c$, so we shall also assume that $c \ge \frac{4}{\delta}$.
Moreover, by \cref{fct:greedy}, if $W_{\CGK}(X)$ and $W_{\CGK}(Y)$ are complete and $\hd(\CGK(X),\CGK(Y))\le c\cdot \ed(X,Y)^2$, then $\A\in \ga_{c\cdot \ed(X,Y)^2}(X,Y)$.

To bound $\width(\A)$, let us analyze the CGK walks $W_{\CGK}(X)=(x_t)_{t=1}^{3n+1}$ and $W_{\CGK}(Y)=(y_t)_{t=1}^{3n+1}$.
Let $T$ be the first step such that $x_T = |X|+1$, $y_T = |Y|+1$, or $T=3n+1$. Note that it suffices to bound $\Pr[\max_{t=1}^{T} |x_t - y_t| \ge c k]$, where $k = \ed(X,Y)$.

Let $A = \{1\}\cup \{t+1 : t\in[1\dd T)\text{ and }X[x_t]\ne Y[y_t]\}$. Observe that $x_{t}-y_{t}=x_{t-1}-y_{t-1}$ holds
for $t\in [1\dd T]\setminus A$, so we may focus on bounding $\Pr[\max_{t\in A} |x_t - y_t| \ge c k]$.
Let $A = \{t_0,\ldots, t_a\}$ with $t_0 < \cdots < t_a$,
and let $d_i = x_{t_i}-y_{t_i}$ for $i\in [0\dd a]$.
Observe that $(d_i)_{i=0}^a$ is an $a$-step unbiased lazy random walk and that $a\le \hd(\CGK(X),\CGK(Y))$.
Moreover, if we extend $(d_i)_{i=0}^a$ to an infinite unbiased lazy random walk $(d_i)_{i=0}^\infty$,
then \cref{fct:walk} yields $\Pr[\max_{i=0}^{\lfloor c k^2\rfloor} |d_i| \ge c k] \le \frac{c k^2}{c^2k^2}=\frac{1}{c}\le \frac\delta4$. 
Therefore,
\[\Pr\left[\max_{i=0}^{a} |d_i| \ge c k\right] \le  \Pr\left[\hd(\CGK(X),\CGK(Y))> c k^2\right]  + \Pr\left[{\max_{i=0}^{\lfloor c k^2\rfloor} |d_i| \ge c k}\right].\]
Consequently, if $\CGK(X)$ and $\CGK(Y)$ are complete, $\hd(\CGK(X),\CGK(Y)) \le c k^2$,
and the random walk satisfies $\Pr\big[{\max_{i=0}^{\lfloor c k^2\rfloor} |d_i| < c k}\big]$, then the alignment $\A$ satisfies the lemma. The total probability of the complementary events is at most $\delta$, so this completes the proof.
\end{proof}

To complete the proof of \cref{prp:alg} (repeated below), it remains to reduce the number of random bits using Nisan's pseudorandom generator~\cite{Nisan}.

\prpalg*

\begin{proof}
Let $c_{\ref{cor:CGK}}$ and  $n_{\ref{cor:CGK}}$ be the constant and threshold (respectively) of \cref{cor:CGK}
for $\delta_{\ref{cor:CGK}}=\frac{\delta}{2}$.

If $n < \max(\frac{10}{\delta},n_{\ref{cor:CGK}})$, we set $\Alg(n,r,S)=(\min(t,|S|+1))_{t=1}^{3n+1}$ for all $S\in \Sigma^{\le n}$ 
so that  $\Alg(n,r,S)$  is trivially a $3n$-step complete walk over $S$.
Now, consider strings $X,Y\in \Sigma^{\le n}$ and the zip alignment $\A_\Alg$ of  $\Alg(n,r,X)$ and $\Alg(n,r,Y)$.
Observe that $\A_\Alg\in \ga_{0}(X,Y)$ if $X=Y$ and $\A_\Alg\in \ga_n(X,Y)$ otherwise.
Moreover, $\width(\A_\Alg)=\big||X|-|Y|\big|\le \ed(X,Y)$. Consequently, the claimed conditions are (deterministically) satisfied for $c \ge n_{\ref{cor:CGK}}$.
The construction algorithm uses $\Oh(\log n)$ bits and $\Oh(n)$ time.

If $n \ge \max(\frac{10}{\delta},n_{\ref{cor:CGK}})$, we first develop an algorithm $\Alg'$ that uses $\Theta(n\log \sigma)$ random bits, interpreted as a sequence $R\in \Hash^{3n}$. Specifically, each hash function $R_t$ is specified by an $\lceil\log \sigma\rceil$-bit integer $h_t$ so that $R_t(a)$ counts (modulo 2) the set bits in $a \textsf{ xor } h_t$.

We set $\Alg'(n,R,S)=(\max(t-3n+|S|,s_t))_{t=1}^{3n+1}$
based on the CGK walk $W_{\CGK}(S)=(s_t)_{t=1}^{3n+1}$ for all $S\in \Sigma^{\le n}$.
Observe that this modification guarantees that $\Alg'(n,R,S)$ is a complete walk over $S$
and that $\Alg'(n,R,S)=W_{\CGK}(S)$ if $W_{\CGK}(S)$ is already complete.
Consequently, by \cref{cor:CGK}, the claimed conditions are satisfied with probability at least $1-\delta_{\ref{cor:CGK}}=1-\frac\delta2$ for $c \ge c_{\ref{cor:CGK}}$.
The construction algorithm uses $\Oh(\log n)$ bits, costs $\Oh(n)$ time, reads the random bits $(h_t)_{t=1}^{3n+1}$ from left to right, and outputs the elements of $\Alg'(n,R,S)$ when required.

In order to use Nisan's pseudorandom generator~\cite{Nisan}, we need to argue that we only care about properties testable using $\Oh(\log n)$-bit algorithms with one-way access to the randomness~$R$.
As in~\cite{CGK}, our testers are given two (read-only) strings $X,Y\in \Sigma^{\le n}$
and a stream of random bits representing $R$ (the sequence $(h_t)_{t=1}^{3n}$).
Observe that the following properties of the zip alignment $\A'$ of  $\Alg'(n,R,X)=(x_t)_{t=1}^{3n+1}$ and $\Alg'(n,R,Y)=(y_t)_{t=1}^{3n+1}$ are testable in $\Oh(\log n)$ bits and $\Oh(n)$ time:
\begin{itemize}
	\item whether $\A'$ is greedy;
	\item whether $\cost(\A')\le k$ for a given integer $k\in [0\dd n]$;
	\item whether $\width(\A')\le w$ for a given integer $w\in [0\dd n]$.
\end{itemize}
All these testers simply construct triples $(t,x_t,y_t)$ for subsequent $t\in [1\dd 3n+1]$.

\newcommand{\PRG}{\mathsf{PRG}}

In this setting, Nisan's pseudorandom generator~\cite{Nisan}, given a sequence $r$ of $\Oh(\log^2 n)$ random bits,
constructs a sequence $\PRG(r)$ of pseudorandom bits in $\Oh(1)$ amortized time per bit, using $\Oh(\log^2 n)$ bits of working space.
Moreover, for every tester, the probabilities of accepting a given input $\mathcal{I}$ with randomness $R$
and $\PRG(r)$ differ by at most $\frac{1}{n^2}$~\cite[Theorem 5]{CGK}.
Consequently, setting $\Alg(n,r,S)=\Alg'(n,\PRG(r),S)$, we can guarantee that the claimed 
condition is satisfied with probability at least $1-\frac12\delta - \frac{(2n+3)}{n^2} \ge 1-\delta$.
The streaming construction algorithm takes $\Oh(\log^2 n)$ bits and $\Oh(n\log n)$ time, dominated by the generation of $\PRG(r)$.
\end{proof}

\newcommand{\Fr}{\mathcal{F}}
\newcommand{\Cov}{\mathsf{Cov}}
\newcommand{\Pos}{\mathsf{Pos}}

\subsection{Context encoding}\label{sec:cover}
\newcommand{\Cv}{\mathcal{C}}

For a string $S\in \Sigma^*$, define $\maxLZ(S) = \max_{[\ell\dd r)\sub [1\dd |S|]} |\LZ(\rev{S[\ell \dd r)})|$. 
\begin{observation}\label{obs:maxLZ}
	Consider a string $S \in \Sigma^*$ and an integer $k\in \Zp$. If $\maxLZ(S[\ell \dd r))\le k$ holds for a non-empty fragment $S[\ell\dd r)$, then:
	\begin{enumerate}[label=(\alph*)]
		\item $\maxLZ(S[\ell\dd r]) \le k$ if and only if $\LZ(\rev{S[\ell \dd r]})\le k$;\label{it:test}
		\item $\maxLZ(S[\ell+1\dd r))\le k$;\label{it:ppp}
		\item $\maxLZ(S[\ell\dd r]) \le k+1$.\label{it:tpp}
	\end{enumerate}
	\end{observation}
	\begin{proof}
	The first two parts follow from the fact that the size of the $\LZ$-factorisation of a prefix of a string is bounded by the size of the $\LZ$-factorisation of the string itself. The third claim follows from the optimality of the LZ77 parsing among all LZ-like parsings (\cref{fct:lz_properties}).
\end{proof}

\newcommand{\Ctx}[3]{C_{#1}(#2,#3)}
\newcommand{\Dtx}[3]{C^2_{#1}(#2,#3)}

\begin{definition}[(Double) Context]
	Consider a string $S\in \Sigma^*$ and an integer $k\in \Zp$.
	For a position $p\in [1\dd |S|]$, we define the \emph{context} $\Ctx{k}{S}{p}$
	as the longest prefix $S[p\dd q)$ of $S[p \dd |S|]$ such that $\maxLZ(S[p\dd q))\le k$
	and the \emph{double context} $\Dtx{k}{S}{p}$ as the longest prefix $S[p\dd q)$ of $S[p\dd |S|]$ such that $\maxLZ(S[p\dd r))\le k$ and $\maxLZ(S[r\dd q))\le k$ for some $r\in [p\dd q]$.
\end{definition}

For integers $t<v$, let $\mu(t,v)$ denote the integer $w\in [t\dd v)$ divisible by the largest power~of~2.

\begin{algorithm}[ht]
	\SetKwInput{KwData}{Randomness}
    \SetAlgoNoLine
	\SetKwBlock{Begin}{}{end}
	\KwIn{An integer $k\in \Zz$, a string $S\in \Sigma^*$, and a complete walk $(s_t)_{t=1}^{m+1}$ over $S$.}
	\KwOut{The string $\bCGK_k(W)[1\dd m]$.}
	\vspace{.25cm}
	$\bCGK_k(W) := \bot^{m}$, where $\bot = (\LZ(\eps),0)$\;
	$s_0 : =0;\; \mu_0 := 0;\; u := 0$\;
	\For{$t := 1$ \KwTo $m$} {
		\If{$s_t \le |S|$} {
			$\mu_t := \mu(t,\min\{v\in [1\dd m+1] : s_v = s_t + |\Ctx{k}{S}{s_t}|\})$\;
			\If{$\mu_t > \mu_{t-1}$}{
				$\bCGK_k(W)[t] := \left(\LZ\left(\rev{\Dtx{k}{S}{s_t}}\right),s_{t}-s_{u}\right)$\;
				$u := t$\;
			}
		}
	}
	\Return{$\bCGK_k(W)$}\;
	\caption{Function $\bCGK$}\label{alg:bCGK}
\end{algorithm}

We say that a position $S[s]$ is covered by a fragment $S[\ell\dd r)$ if $\ell\le s < r$.

\begin{lemma}\label{lem:cover}
Consider $\bCGK_{k}(W)$ constructed for an integer $k\in \Zz$ and an $m$-complete walk $W=(s_t)_{t=1}^{m+1}$ over $S\in \Sigma^*$. Each position $s\in [1\dd |S|]$ satisfies
\begin{multline*}1\le |\{t\in [1\dd m] : \bCGK_k(W)[t]\ne \bot\text{ and }\Ctx{k}{S}{s_t}\text{ covers }S[s]\}| \le 	\\ \le |\{t\in [1\dd m] : \bCGK_k(W)[t]\ne \bot\text{ and }\Dtx{k}{S}{s_t}\text{ covers }S[s]\}| \le \Oh(\log m)
\end{multline*}
\end{lemma}
\begin{proof}
Let us first bound the covering number from below.
Consider an index $t\in [1\dd m]$ such that $s=s_t$,
and let $t'\in [1\dd m]$ be the smallest index such that $\mu_{t'}=\mu_t$.
Note that $s_{\mu_{t'}} < s_{t'}+|\Ctx{k}{S}{s_{t'}}|$ and $s_{\mu_t}\ge s_t$ by definition of $\mu$
and monotonicity of the walk $W$.
Due to $\mu_{t'}=\mu_t$, this implies $s_{t'} \le s_t < s_{t'}+|\Ctx{k}{S}{s_{t'}}|$,
i.e., that $S[s_t]$ is covered by $\Ctx{k}{S}{s_{t'}}$.
Due to $\mu_{t'-1}\ne \mu_t$, this guarantees $\bCGK_{k}(S)[t']\ne \bot$.

As for the upper bound, note that the indexes $t\in [1\dd m]$ with $\bCGK_{k}(S)[t]\ne \bot$
have distinct values $\mu_t$. Let us further classify them into $\Oh(\log n)$ groups
depending on the largest power of two dividing $\mu_t$.
Consider two indexes $t<t'$ in the same group.
Since $\mu_{t} < \mu_{t'}$ are divisible by the same largest power of two,
there is a number $\nu\in (\mu_t\dd \mu_{t'})$ divisible by a strictly larger power of two.
By definition of $\mu_t$, we have $s_\nu \ge s_t + |\Ctx{k}{S}{s_{t}}|$
and, by definition of $\mu_{t'}$, we have $s_\nu < s_{t'}$.
Consequently, $s_{t'}>s_t + |\Ctx{k}{S}{s_{t}}|$, i.e., the contexts $\Ctx{k}{S}{s_{t}}$ and $\Ctx{k}{S}{s_{t'}}$ are disjoint.
By monotonicity of $\maxLZ$, this also means that $s_t + |\Dtx{k}{S}{s_{t}}| \le  s_{t'} + |\Ctx{k}{S}{s_{t'}}|$.
In particular, each position $s\in [1\dd |S|]$ is covered by at most one context $\Ctx{k}{S}{s_{t}}$ and at most two double contexts $\Dtx{k}{S}{s_{t}}$ for $t\in [1\dd m]$ belonging to a single group.
\end{proof}

Let $\A$ be an alignment between $X,Y\in \Sigma^*$.
We define $\brkp_X(\A) = \{x : (x,y)\in \brkp(\A)\}\cap [1\dd |X|]$.
Moreover, for two alignments $\A, \A'$ between $X,Y\in \Sigma^*$, we define $\Delta_X(\A,\A')$
to contain $x\in [1\dd |X|]$ unless $(x,y)\in \mtch(\A)\cap \mtch(\A')$ for some $y\in [1\dd |Y|]$.

\newcommand{\Opt}{\mathcal{O}}

\begin{lemma}\label{lem:close}
Let $b\in \mathbb{Z}_{+}$ and $\A,\A'$ be greedy alignments of strings $X$ and $Y$.
The positions in $\Delta_X(\A,\A')$ can be covered by at most $\cost(\A)+\frac1b\cost(\A')$ contexts
$\Ctx{\width(\A)+\width(\A')+2b}{X}{\cdot}$. 

Moreover, if $b > \cost(\A')$, then the positions in $\Delta_X(\A,\A')$ can be covered 
by contexts $\Ctx{\width(\A)+\width(\A')+2b}{X}{x}$ with $x\in \brkp_X(\A)$.
\end{lemma}
\begin{proof}
Let $w = \width(\A)$ and $w'=\width(\A')$.
Without loss of generality, we may trim the longest common prefix of $X$ and $Y$; 
this is feasible because both $\A$ and $\A'$ match the common prefix, so
$\Delta_X(\A,\A')$ only contains positions following the prefix. 
Moreover, we assume that $X\ne \eps$; otherwise $\Delta_X(\A,\A')=\emptyset$ and the lemma is trivial.

In the remaining case, we are guaranteed that $1\in \brkp_X(\A)$.
Let $\brkp_X(\A)=\{x_1,\ldots,x_m\}$, where $1=x_1<\cdots < x_m$, and let $x_{m+1}=|X|+1$.
This yields a decomposition $X = X_1\cdots X_m$ into $m$ non-empty substrings $X_i := X[x_i\dd x_{i+1})$.
The alignment $\A'$ further yields a decomposition $Y = Y_1\cdots Y_m$ into $m$ (possibly empty) substrings
$Y_i:=Y[y_i\dd y_{i+1})$ so that $X_i \sim_{\A'} Y_i$ and $|x_i-y_i|\le w'$ for $i\in [1\dd m]$.
Moreover, the cost of $\A'$ can be expressed as $\cost(\A') = \sum_{i=1}^m c'_i$, where $c'_i$
denotes the cost of $\A'$ restricted to $X_i$ and $Y_i$.

\begin{claim}
For every $i\in [1\dd m]$, the set $\Delta_X(\A,\A')\cap [x_i \dd x_{i+1})$ can be covered by at most $\lceil\frac{1}{b}(c'_i+1)\rceil$ contexts $\Ctx{w+w'+2b}{X}{\cdot}$ including $\Ctx{w+w'+2b}{X}{x_i}$.
\end{claim}
\begin{proof}
Let $q_i = \max(\Delta_X(\A,\A')\cap [x_i \dd x_{i+1}))$. 
We may assume that $q_i \ge  x_i + w + w'$; otherwise, the claim holds trivially.
We shall prove that $X[x_i+1\dd q_i-(w+w')]$ can be decomposed into at most $2c'_i+1$ phrases,
each of which is a single character or has another occurrence at most $w+w'$ positions to the right.
Accounting for $X[x_i]$, this yields a decomposition of $X[x_i\dd q_i-(w+w')]$ 
into at most $2c'_i+2$ such phrases.
The $\maxLZ$ measure of the concatenation of every $2b$ subsequent phrases and $w+w'$ following single characters 
does not exceed $w+w'+2b$.
Hence, $X[x_i\dd q_i]$ can be covered by $\lceil\frac{1}{b}(c'_i+1)\rceil$ contexts $\Ctx{w+w'+2b}{X}{\cdot}$ including $\Ctx{w+w'+2b}{X}{x_i}$.

Thus, it remains to construct the decomposition of $X[x_i+1\dd q_i-(w+w')]$ into phrases. By definition of $\brkp_X(\A)$, we have ${X[x_i+1\dd x_{i+1}) \simeq_{\A} Y[x_i+1+d\dd x_{i+1}+d)}$ for some shift $d\in [-w\dd w]$. We consider two cases.

In the first case, we assume that $(x_i+1,\bar{y})\in \A'$ holds for some $\bar{y}\ge x_i+1+d$.
The greedy nature of $\A'$ then guarantees that if $(x,y)\in \A'$ with $x\in [x_i+1\dd q_i)$, then $y>x+d$.
We decompose $X[x_i+1\dd q_i-(w+w')]$ (which is contained in $X_i$) into maximal phrases $X[x\dd x']$ satisfying $X[x\dd x'] \simeq_{\A'} Y[x+d'\dd x'+d']$ for some $d' \in (d\dd w']$ and remaining single characters (deleted or substituted by $\A'$).
Observe that this yields at most $1+c'_i$ phrases and at most $c'_i$ single characters.
Moreover, we have $X[x\dd x'] \simeq_{\A'} Y[x+d'\dd x'+d'] \simeq_{\A} X[x+(d'-d)\dd x'+(d'-d)]$
due to $x_i+1 \le x \le x+(d'-d)$ and $x'+(d'-d) \le x'+(w'+w) \le q_i$. Hence,
each phrase $X[x\dd x']$ has another occurrence located $d'-d\in [1\dd w+w']$ positions to the right.

In the second case, we assume that $(x_i+1,\bar{y})\in \A'$ holds for some $\bar{y}\le x_i+1+d$.
The greedy nature of $\A'$ then guarantees that if $(x,y)\in \A'$ with $x\in [x_i+1\dd q_i)$, then $y<x+d$. 
We decompose $Y[x_i+1+d\dd q_i-(w+w')+d]$ (which is contained in $Y_i$) into maximal phrases $Y[y\dd y']$ satisfying $Y[y\dd y'] \simeq_{\A'} X[y-d'\dd y'-d']$ for some $d'\in [-w\dd d)$ and remaining single characters (deleted or substituted by $\A'$).
Observe that this yields at most $1+c'_i$ phrases and at most $c'_i$ single characters.
Moreover, we have $Y[y\dd y'] \simeq_{\A'} X[y-d'\dd y'-d'] \simeq_{\A} Y[y+(d-d')\dd y+(d-d')]$
due to $x_i+1 \le x_i+1 + d-d' \le y-d'$ and $y'-d' \le q_i-(w+w')+d-d' \le q_i$.
Hence, each phrase $Y[y\dd y']$ has another occurrence located $d-d' \in [1\dd w+w']$ characters to the right.
Since $Y[x_i+1+d\dd q_i+d] = X[x_i+1\dd q_i]$, this decomposition of $Y[x_i+1+d\dd q_i-(w+w')+d]$
gives an analogous decomposition of $X[x_i+1\dd q_i-(w+w')]$.
\end{proof}

The set $\Delta_X(\A,\A')$ can be covered by $\sum_{i=1}^m \lceil\frac{1}{b}(c'_i+1)\rceil \le m + \frac{\cost(\A')}{b}\le \cost(\A)+\frac{\cost(\A')}{b}$ contexts $\Ctx{w+w'+2b}{X}{\cdot}$. 
If $b \ge \cost(\A')+1$, then the contexts starting at positions in $\brkp_X(\A)$ are sufficient.
\end{proof}

\begin{lemma}\label{lem:bcgk}
Let $\A_W$ be the zip alignment of $m$-complete walks $W_X = (x_t)_{t=1}^{m+1}$, $W_Y=(y_t)_{t=1}^{m+1}$ over strings $X,Y\in \Sigma^*$, and let $k'\in \Zp$.
If $\A_W$ is greedy, then $\brkp_X(\A_W)\sub P_X$, where $P_X$ is the set of positions $x\in [1\dd |X|]$ such that $X[x]$ is covered by $\Dtx{k'}{X}{x_t}$ for some $t\in [1\dd m]$ with $\bot \ne \bCGK_{k'}(W_X)[t] \ne \bCGK_{k'}(W_Y)[t]$.
Moreover, for every positive integer $k\ge \ed(X,Y)$ such that $k'\ge \width(\A_W)+5k$:
\begin{itemize}
	\item every $\A\in\ga_k(X,Y)$ satisfies $\Delta_X(\A,\A_W)\sub P_X$,
	\item $\hd(\bCGK_{k'}(X),\bCGK_{k'}(Y)) \le c(k+\frac1k\cost(\A_W))\log m$ holds for a sufficiently large constant $c$.
\end{itemize}
\end{lemma}
\begin{proof}
Consider a position $x\in \brkp_X(\A_W)$. By \cref{lem:cover}, there is an index $t\in [1\dd m]$ 
such that $\bCGK_{k'}(W_X)[t]\ne \bot$ and $\Dtx{k'}{X}{x_t}$ covers $X[x]$.
To conclude that $x\in P_X$, it remains to prove $\bCGK_{k'}(W_X)[t]\ne \bCGK_{k'}(W_Y)[t]$.
For a proof by contradiction, suppose that $\bCGK_{k'}(W_X)[t]=\bCGK_{k'}(W_Y)[t]$.
This implies $\Dtx{k'}{X}{x_t}=\Dtx{k'}{Y}{y_t}$. Since $\A_W$ is greedy, we derive $X[x_t\dd x_t+|\Dtx{k'}{X}{x_t}|) \simeq_{\A_W} Y[y_t\dd y_t + |\Dtx{k'}{Y}{y_t}|)$, which contradicts $x\in \brkp_X(\A_W)$.

In the remainder of the proof, we assume that $k\ge \ed(X,Y)$ and $k'\ge \width(\A_W)+5k$.
Next, consider a greedy alignment $\A$ with $\cost(\A)\le k$ and a position $x\in \Delta_X(\A,\A_W)$.
Due to $k' \ge \width(\A_W)+\width(\A)+4k$, \cref{lem:close} shows that there exists a position $r\in \brkp_X(\A_W)$ such that $\Ctx{k'}{X}{r}$ covers $X[x]$
and, by \cref{lem:cover}, there exists a position $t\in [1\dd m]$ such that $\bCGK_{k'}(X)[t]\ne \bot$ and $\Ctx{k'}{X}{x_t}$ covers $X[r-1]$.
Since $\maxLZ(X[x_t\dd r)),\maxLZ(X[r\dd x])\le k'$, we conclude that both $r$ and $x$ are covered by $\Dtx{k'}{X}{x_t}$.
We shall prove that $\bCGK_{k'}(X)[t] \ne \bCGK_{k'}(Y)[t]$.
For a proof by contradiction, suppose that $\bCGK_{k'}(W_X)[t]=\bCGK_{k'}(W_Y)[t]$,
which implies $\Dtx{k'}{X}{x_t}=\Dtx{k'}{Y}{y_t}$.
Let us fix the smallest $t'\in [t\dd m+1]$ such that $(x_{t'},y_{t'})\in \brkp(\A_W)$.
Observe that $X[x_t\dd x_{t'}) \simeq_{\A_{W}} Y[y_t \dd y_{t'})$ and, by the greedy nature of $\A_{W}$,
$X[x_t\dd x_{t'})=Y[y_t \dd y_{t'})$ is the longest common prefix of $X[x_t\dd]$ and $Y[y_t\dd]$.
At the same time, due to $x_t  < r \in \brkp_X(\A_W)$, we have 
$x_{t'}\le r$, so $X[x_t\dd x_{t'}]$ is a prefix of $\Dtx{k'}{X}{x_t}$.
However, $X[x_t\dd x_{t'}]$ is not a prefix $Y[y_t\dd ]$, so $\Dtx{k'}{X}{x_t}$
is not a prefix of $Y[y_t\dd ]$ and  $\Dtx{k'}{X}{x_t} \ne \Dtx{k'}{Y}{y_t}$.
The contradiction completes the proof.

Finally, we bound $\hd(\bCGK_{k'}(X),\bCGK_{k'}(Y))$ using several claims.
Consider a set $M$ of indices $t\in [1\dd m]$
such that $\bCGK_{k'}(X)[t]$ and $\bCGK_{k'}(Y)[t]$ differ on the first coordinate.
\begin{claim}\label{clm:bcgk}
If $t\in M$ satisfies $\bCGK_{k'}(X)[t] \ne \bot$,
then \[[x_{t-1}\dd x_t+|\Dtx{k'}{X}{x_t}|]\cap (\brkp_X(\A_W)\cup \{|X|+1\})\ne \emptyset.\]
\end{claim}
\begin{proof}
For a proof by contradiction, suppose that $[x_{t-1}\dd x_t+|\Dtx{k'}{X}{x_t}|]\cap (\brkp_X(\A_W)\cup \{|X|+1\})=\emptyset$.
Let $L$ be the length of the longest common prefix of $X[x_t\dd]$ and $Y[y_t\dd ]$.
Note that $(x_t+L,y_t+L)\in \brkp(\A_W)$, so $x_t + L \in \brkp_X(\A_W)\cup \{|X|+1\}$.
Consequently, $L > |\Dtx{k'}{X}{x_t}| \ge |\Ctx{k'}{X}{x_t}|$,
and therefore $X[x_t\dd x_t+|\Dtx{k'}{X}{x_t}|]=Y[y_t\dd y_t+|\Dtx{k'}{Y}{y_t}|]$ and 
$X[x_t\dd x_t+|\Ctx{k'}{X}{x_t}|]=Y[y_t\dd y_t+|\Ctx{k'}{Y}{y_t}|]$.
In particular, $\Dtx{k'}{X}{x_t}=\Dtx{k'}{Y}{y_t}$.

If $t=1$, then we note that $\bCGK_{k'}(X)[t]\ne \bot \ne \bCGK_{k'}(Y)[t]$,
so $\Dtx{k'}{X}{x_t}=\Dtx{k'}{Y}{y_t}$ contradicts $t\in M$.

In the following, we assume that $t\ge 2$. 
Due to $(x_{t-1},y_{t-1})\notin \brkp(\A_W)$, either $(x_{t-1},y_{t-1})=(x_t,y_t)$,
or $(x_{t-1},y_{t-1})=(x_t-1,y_t-1)$ and $X[x_t-1]=Y[y_t-1]$.
In either case, we have $X[x_{t-1}\dd x_{t-1}+|\Ctx{k'}{X}{x_{t-1}}|]=Y[y_{t-1}\dd y_{t-1}+|\Ctx{k'}{Y}{y_{t-1}}|]$
due to $X[x_{t-1}\dd x_{t}+|\Ctx{k'}{X}{x_t}|]=Y[y_{t-1}\dd y_t+|\Ctx{k'}{Y}{y_t}|]$ and by \cref{obs:maxLZ}.
Consequently, $\min\{u\in [1\dd m] : x_u = x_{t-1} + |\Ctx{k'}{X}{x_{t-1}}|\}
=\min\{u\in [1\dd m] : y_u = y_{t-1} + |\Ctx{k'}{Y}{y_{t-1}}|\}$ and,
by a similar reasoning, $\min\{u\in [1\dd m] : x_u = x_{t} + |\Ctx{k'}{X}{x_t}|\}=
\min\{u\in [1\dd m] : y_u = y_{t} + |\Ctx{k'}{Y}{y_t}|\}$.
Hence, the assumption  $\bCGK_{k'}(X)[t] \ne \bot$ implies $\bCGK_{k'}(Y)[t]\ne \bot$.
Therefore, $\Dtx{k'}{X}{x_t}=\Dtx{k'}{Y}{y_t}$ contradicts $t\in M$.
\end{proof}

\begin{claim}
There are $\Oh((k+\frac1k\cost(\A_W))\log m)$ positions $t\in M$ such that $\bCGK_{k'}(X)[t]\ne \bot$.
\end{claim}
\begin{proof}
Let $\Opt$ be an optimum greedy alignment between $X$ and $Y$. 
Due to $k' \ge \width(\Opt)+\width(\A_W)+4k$, from~\cref{lem:close} it follows that $\Delta_X(\Opt,\A_{W})$
can be covered by $\Oh(\cost(\Opt)+\frac{1}{k}\cost(\A_W))=\Oh(k+\frac1k\cost(\A_W))$ contexts $\Ctx{k'}{X}{\cdot}$. Let us fix such the smallest family $\Cv$ of contexts covering $\brkp_X(\A_{W})\sub \Delta_X(\Opt,\A_{W})$.

Define a set $R_X = \{|X|\} \cup \bigcup_{X[q\dd r]\in \Cv} \{q-1,r+1\}$
and note that $|R_X| = \Oh(|\Cv|)$.
We shall prove that, if $t\in M$ and $\bCGK_{k'}(X)[t]\ne \bot$, then $\Dtx{k'}{X}{x_t}$ covers at least one position in $R_X$. This is sufficient to derive the claim because, by \cref{lem:cover}, every position in $X$ is covered by at most $\Oh(\log m)$ double contexts $\Dtx{k'}{X}{x_t}$ with $\bCGK_{k'}(X)[t]\ne \bot$.

By \cref{clm:bcgk}, we have $[x_{t-1}\dd x_t+|\Dtx{k'}{X}{x_t}|] \cap (\brkp_X(\A_{W})\cup \{|X|+1\})\ne \emptyset$. If $ \{|X|+1\}\in [x_{t-1}\dd x_t+|\Dtx{k'}{X}{x_t}|]$, then $|X|+1=  x_t+|\Dtx{k'}{X}{x_t}|$.
Hence, $\Dtx{k'}{X}{x_t}$ covers the position $|X|\in R_X$.
Thus, we may assume that $[x_{t-1}\dd x_t+|\Dtx{k'}{X}{x_t}|]$ contains a position in $\brkp_X(\A_W)$.
Note that the fragment $X[q\dd r]\in \Cv$ covering that position in $\brkp_X(\A_W)$
satisfies $r \ge x_{t-1}$ and $q \le x_t+|\Dtx{k'}{X}{x_t}|$.
In particular, $r+1 \ge x_t$ and $q-1 < x_t+|\Dtx{k'}{X}{x_t}|$.
Now, if $\Dtx{k'}{X}{x_t}$ covers $q-1$ or $r+1$, then we are done.
Otherwise, $q-1 < x_t$ and $r+1 \ge x_t+|\Dtx{k'}{X}{x_t}|$,
so $q \le x_t$ and $r \ge  x_t+|\Dtx{k'}{X}{x_t}|$, i.e., $\Dtx{k'}{X}{x_t}$ is contained in $X[q\dd r]$
and, since $\maxLZ$ is monotone, $\maxLZ(\Dtx{k'}{X}{x_t})\le k'$.
Consequently, $\Dtx{k'}{X}{x_t}=\Ctx{k'}{X}{x_t}$ is a suffix of $X$, so $\Dtx{k'}{X}{x_t}$ covers the position $|X|\in R_X$.
\end{proof}

A symmetric argument shows that there are $\Oh((k+\frac1k\cost(\A_W))\log m)$ positions $t\in M$ such that $\bCGK_{k'}(Y)[t] \ne \bot$. Consequently, $|M|=\Oh((k+\frac1k\cost(\A_W))\log m)$. 
We bound $\hd(\bCGK_{k'}(X),\bCGK_{k'}(Y))$ using the following claim. 
\begin{claim}
	$\hd(\bCGK_{k'}(X),\bCGK_{k'}(Y)) \le 2|M|$.
\end{claim}
\begin{proof}
 Let $M' = \{t\in [1\dd m] :\bCGK_{k'}(X)[t] \ne \bCGK_{k'}(Y)[t]\}$.
 Note that $M\sub M'$, so the claim is equivalent to $|M'\sm M| \le |M|$.
 For every  $t\in M'\sm M$, we shall prove that $t$ is not the leftmost position in $M'$ and that the preceding position in $M'$ belongs to $M$.

	The assumption $t\in M'\sm M$ implies $\bCGK_{k'}(X)[t] \ne \bot \ne \bCGK_{k'}(Y)[t]$.
	We define $t'$ as the largest position in $[1\dd t)$ such that $\bCGK_{k'}(X)[t']\ne \bot$ or $\bCGK_{k'}(Y)[t']\ne \bot$.
	To see that $t'$ is well defined, note that $t>1$ and $\bCGK_{k'}(X)[1]\ne \bot$.

	Now, suppose that $t'\notin M$. Consequently, we have $\bCGK_{k'}(X)[t'] \ne \bot \ne \bCGK_{k'}(Y)[t']$.
	Due to $\bCGK_{k'}(X)[t'+1\dd t)=\bot^{t-t'-1}$, \cref{lem:cover} implies that $\Ctx{k'}{X}{x_{t'}}$
	covers $X[x_t-1]$, that is, $\maxLZ(X[x_{t'}\dd x_t))\le k'$, and
	hence $X[x_{t'}\dd x_t]$ is a prefix of $\Dtx{k'}{X}{x_{t'}}$.
	Moreover, $\Dtx{k'}{X}{x_{t'}} = \Dtx{k'}{Y}{y_{t'}}$
	implies $X[x_{t'}\dd x_t] \simeq_{\A_{W}} Y[y_{t'}\dd y_{t'}+x_t - x_{t'}]$ by the greedy nature of $\A_{W}$.
	As $(x_t,y_t)\in \A_{W}$, we conclude that $y_t = y_{t'}+x_t - x_{t'}$.
	At the same time, since $\bCGK_{k'}(X)[t'+1\dd t) = \bot^{t-t'-1} = \bCGK_{k'}(Y)[t'+1\dd t)$,
	we have $\bCGK_{k'}(X)[t] = (\LZ(\rev{\Dtx{k'}{X}{x_{t}}}),x_t-x_{t'})$ and $\bCGK_{k'}(Y)[t] = (\LZ(\rev{\Dtx{k'}{Y}{y_{t}}}),y_t-y_{t'})$.
	Thus, $\bCGK_{k'}(X)[t]= \bCGK_{k'}(Y)[t]$, which contradicts $t\in M'$. 

	Consequently, $t'\in M$.
	In this case, $M'\cap [t'\dd t] = \{t',t\}$, so $t'$ is the position of $M'$ preceding~$t$.
	Our goal was to show that such a position exists and belongs to $M$, so this completes the proof of the claim.
\end{proof}
Overall, we conclude that $\hd(\bCGK_{k'}(X),\bCGK_{k'}(Y)) \le 2|M| = \Oh((k+\frac1k\cost(\A_W))\log m)$.
\end{proof}

\begin{lemma}\label{cor:bcgkp_construction}
	Given integers $n\ge k \ge 1$, a seed $r$ of $\Oh(\log^2 n)$ random bits, and streaming access to a string $S\in \Sigma^{\le n}$, the string $\bCGK_k(\Alg(S,n,r))$ can be computed in $\tOh(k)$ space and $\tOh(nk)$ time.
	\end{lemma}
	\begin{proof}
	Let us start with an auxiliary subroutine:
	\begin{claim}\label{clm:streaming_single}
		There is a streaming algorithm that computes $\CR(\Ctx{k}{S}{p})$ for subsequent positions $p\in [1\dd |S|]$. The algorithm uses $\tOh(k)$ space and $\tOh(nk)$ time. 
		Moreover, if $\Ctx{k}{S}{p}=S[p\dd q)$, then $\CR(\Ctx{k}{S}{p})$ is reported while the algorithm processes $S[q]$ (or the end-of-string token if $q=|S|+1$).
	\end{claim}
	\begin{proof}
	The algorithm maintains a fragment $S[p\dd q)$ satisfying $\maxLZ(S[p\dd q))\le k$
	and the encoding $\CR(S[p\dd q))$. Initially, $S[p\dd q)=S[1\dd 1)=\eps$.
	
	In each iteration, we read $S[q]$, compute $\CR(S[p\dd q])$ (using \cref{prp:RLSLP}\ref{it:concat}),
	and check whether $|\LZ(\rev{S[p\dd q]})|\le k$ (using \cref{prp:RLSLP}\ref{it:rlz}).
	By \cref{obs:maxLZ}\ref{it:test}, this condition is equivalent to $\maxLZ(S[p\dd q])\le k$.
	If $\maxLZ(S[p\dd q])\le k$, we discard $\CR(S[p\dd q))$ and increment $q$.
	Otherwise, we are guaranteed that $S[p\dd q)=\Ctx{k}{S}{p}$, so we output $\CR(S[p\dd q))$,
	compute $\CR(S[p+1\dd q))$ (using \cref{prp:RLSLP}\ref{it:extract}), discard $\CR(S[p\dd q))$ and $\CR(S[p\dd q])$,
	and increment $p$.
	By \cref{obs:maxLZ}\ref{it:ppp}, we are guaranteed that the invariant $\maxLZ(S[p\dd q))\le k$ remains satisfied.
	In the special case of $q=|S|+1$, we proceed as if $\maxLZ(S[p\dd q])> k$.
	
	By \cref{obs:maxLZ}\ref{it:tpp}, we store $\CR(X)$ only for strings $X$ satisfying $|\maxLZ(X)|\le k+1$,
	so the space usage and the per-iteration running time is $\tOh(k)$.
	Each iteration increments either $p$ or $q$, so the algorithm reads $S$ in a streaming fashion and 
	the amortized running time is $\tOh(k)$ per character.
	\end{proof}

	We maintain two instances of the algorithm of \cref{clm:streaming_single} and two instances of the algorithm of \cref{prp:alg}.
	We feed the first instance of \cref{clm:streaming_single} with the input stream $S$, obtaining 
	$\CR(\Ctx{k}{S}{q})$ for subsequent positions $q\in [1\dd |S|]$. 
	Upon retrieving $\CR(\Ctx{k}{S}{q})$, we extract $S[q]$ using \cref{prp:RLSLP}\ref{it:access}
	and forward $S[q]$ to the first instance of \cref{prp:alg}, which lists indices $v\in [1\dd 3n]$ such that $s_v = q$.
	We also forward $S[q]$ to the second instance of \cref{clm:streaming_single}, obtaining  $\CR(\Ctx{k}{S}{p})$ for all subsequent positions $p$ such that $\CR(\Ctx{k}{S}{p})=S[p\dd q)$.
	Upon retrieving $\CR(\Ctx{k}{S}{p})$, we extract $S[p]$ using \cref{prp:RLSLP}\ref{it:access}
	and forward $S[p]$ to the second instance of \cref{prp:alg}, which lists indices $t\in [1\dd 3n]$ such that $s_t = q$.
	For each such position $t$, we have $\Dtx{k}{S}{s_t} = \Ctx{k}{S}{p}\Ctx{k}{S}{q}$.
	We compute $\mu_t = \mu(t,\min\{v\in [1\dd 3n] : s_v = q\})$ based on the output of the first instance of \cref{prp:alg}. If $\mu_t > \mu_{t-1}$, we construct $\LZ\left(\rev{\Dtx{k}{S}{s_t}}\right)$
	using \cref{prp:RLSLP}\ref{it:concat} and \cref{prp:RLSLP}\ref{it:rlz}.
	
	By \cref{prp:RLSLP,prp:alg,clm:streaming_single}, the algorithm uses $\tOh(k)$ space and $\tOh(nk)$ time.
\end{proof}

\subsection{Applications of Hamming distance sketches}

Let us start by reminding the fingerprints (sketches) for testing string equality.

 
\begin{fact}[see e.g.~\cite{DBLP:journals/ibmrd/KarpR87}]\label{fact:kr}
There exists a fingerprint $\psi$ (parameterized by an integer $n\in \Zp$, a threshold $\delta$ with $1 \ge \delta \ge n^{-\Oh(1)}$, an alphabet $\Sigma=[0\dd n^{\Oh(1)})$, and a seed of $\Oh(\log n)$ random bits) such that:
\begin{enumerate}
	\item The fingerprint $\psi(S)$ of a string $S\in \Sigma^{\le n}$ takes $\Oh(\log \delta^{-1})$ bits.
	Given streaming access to $S$, it can be constructed in $\Oh(|S|)$ time using $\Oh(\log n)$ bits of space.
	\item For all strings $X,Y\in \Sigma^{\le n}$, we have $\Pr[\psi(X)= \psi(Y)] \le \delta$ if $X\ne Y$ (and $\psi(X)=\psi(Y)$ otherwise).
\end{enumerate}
\end{fact}

For two equal-length strings $X,Y$, the set of \emph{mismatch positions} is defined as $\MP(X,Y)=\{i\in [1\dd |X|]: X[i]\ne Y[i]\}$
and the \emph{mismatch information} 
$\MI(X,Y) = \{(i,X[i],Y[i]) : i\in \MP(X,Y)\}$.
Below, we adapt the Hamming sketches of~\cite{clifford2018streaming} to large alphabets.

\begin{theorem}\label{thm:skH}
	For every constant $\delta\in (0,1)$, there exists a sketch $\skH_k$ (parameterized by integers $n\ge k\ge 1$,
	an alphabet $\Sigma=[0\dd \sigma)$, and a seed of $\Oh(\log (n\log \sigma))$ random bits) such that:
	\begin{enumerate}
		\item The sketch $\skH_k(S)$ of a string $S\in \Sigma^{n}$ takes $\Oh(k\log (n\sigma))$ bits. Given streaming access to $S$, it can be constructed in $\Oh(n\log (n\sigma) \log (n\log \sigma))$ time using $\Oh(k\log (n\sigma))$ bits of space.
		\item There exists a decoding algorithm that, given $\skH_k(X)$ and $\skH_k(Y)$ for strings $X,Y\in \Sigma^{n}$, with probability at least $1-\delta$ either returns $\MI(X,Y)$ or certifies that $\hd(X,Y)>k$.
		The algorithm uses $\Oh(k \log (n\sigma) \log^2 (n\log \sigma))$ time and $\Oh(k \log (n\sigma))$ bits of space. 
	\end{enumerate}
\end{theorem}
\begin{proof}
The construction of~\cite{clifford2018streaming} satisfies the required conditions provided that $\sigma = n^{\Oh(1)}$. Henceforth, we assume without loss of generality that $\sigma$ is a power of two satisfying $\sigma \ge n\log \sigma$.

We interpret each character in $\Sigma$ as a block of $b=\lceil \log_{n \log \sigma} \sigma \rceil$
characters in $[0\dd n\log \sigma)$.
Moreover, we consider the fingerprints $\psi$ of \cref{fact:kr} with $\delta_{\ref{fact:kr}}=\frac{\delta}{2n}$ and  $n_{\ref{fact:kr}}=\sigma_{\ref{fact:kr}}=n\log \sigma$.
This construction uses  $\Oh(\log n_{\ref{fact:kr}}) = \Oh(\log (n \log \sigma))$ random bits
and produces fingerprints of $\Oh(\log \delta^{-1}_{\ref{fact:kr}})=\Oh(\log n)$ bits.

Given a string $S\in \Sigma^{n}$, we define a string $\bar{S}[1\dd nb]$ so that 
$\bar{S}[ib-j+1]=(S[i][j], \psi(S[i]))$ for $i\in [1\dd |S|]$ and $j\in [1\dd b]$.
Consequently, we set $\skH_k(S):=\skH_{\bar{k}}(\bar{S})$ using the the sketch of~\cite{clifford2018streaming}
with $\bar{\delta}=\frac\delta2$, $\bar{n}=nb$, $\bar{k}=kb$, and $\bar{\sigma}\le n^{\Oh(1)}\log \sigma$.
Note that this is feasible since $\bar{\sigma}\le n^{\Oh(1)} \le \bar{n}^{\Oh(1)}$ if $n \le \log \sigma$ and $\bar{\sigma} \le \log^{\Oh(1)}\sigma \le (\log_{\log^2\sigma} \sigma)^{\Oh(1)}\le b^{\Oh(1)}\le\bar{n}^{\Oh(1)}$ otherwise.
This construction uses $\Oh(\log \bar{n})=\Oh(\log (n\log \sigma))$ further random bits
and produces a sketch of $\Oh(\bar{k} \log \bar{n}) = \Oh(kb \log (nb)) = \Oh(k \log_{n\log \sigma} \sigma \log (n\log \sigma) = \Oh(k \log \sigma)$ bits.

The auxiliary string $\bar{S}$ is constructed in $\Oh(\bar{n}) = \Oh(n \log \sigma)$ time
using $\Oh(\log (n\sigma))$ bits of space.
The stream representing $\bar{S}$ is passed to the encoding algorithm of~
\cite{clifford2018streaming},
which takes $\Oh(\bar{n}\log^2 \bar{n}) =\Oh(n\log_{n \log \sigma} \sigma \log^2 (n\log \sigma)) = \Oh(n\log(n\sigma) \log (n\log\sigma))$ time and $\Oh(\bar{k} \log \bar{n}) = \Oh(k\log \sigma)$ bits of space.

The decoding algorithm, given $\skH_{\bar{k}}(\bar{X})$ and $\skH_{\bar{k}}(\bar{Y})$,
runs the decoding algorithm of~\cite{clifford2018streaming}.
If the latter certifies $\hd(\bar{X},\bar{Y})>\bar{Y}$, we certify that $\hd(X,Y)>k$.
Otherwise, we interpret the output as $\MI(\bar{X},\bar{Y})$. For each position $i\in [1\dd |X|]$ such that $(ib-b\dd ib]\sub \MP(\bar{X},\bar{Y})$
we retrieve $X[i][j]$ and $Y[i][j]$ for each $j\in [1\dd b]$ from $(ib-j+1,\bar{X}[ib-j+1],\bar{Y}[ib-j+1])\in \MI(\bar{X},\bar{Y})$, and then we output $(i,X[i],Y[i])\in\MI(X,Y)$.

As for correctness, with at most $\delta_{\ref{fact:kr}}+n\bar{\delta}=\delta$ probability loss, we may assume that the decoder of~\cite{clifford2018streaming} is successful and that, for all $i\in [1\dd |X|]$,
we have $\psi(X[i])=\psi(Y[i])$ if and only if $X[i]=Y[i]$.
The latter assumption yields $\MP(\bar{X},\bar{Y}) = \bigcup_{i\in \MP(X,Y)}(ib-b\dd ib]$
and thus $\hd(\bar{X},\bar{Y})=b\cdot \hd(X,Y)$. Hence, we correctly certify $\hd(X,Y)>k$ if $\hd(\bar{X},\bar{Y})>\bar{k}$, and we correctly reconstruct $\MI(X,Y)$ otherwise.

The decoder uses $\Oh(\bar{k}\log^3 \bar{n})  =\Oh(k\log_{n\log \sigma} \sigma \log^3(n\log \sigma))=\Oh(k\log \sigma \log^2 (n\log \sigma))$ time and $\Oh(\bar{k}\log\bar{n}) = \Oh(k\log \sigma)$ bits of space.
\end{proof}

Next, consider an alphabet $\hat{\Sigma}:=\Sigma\times \Zz$ 
and a function $\pi : \hat{\Sigma} \to \Zz$ defined so that $\pi((a,v))=v$ for $(a,v)\in \hat{\Sigma}$
and $\pi(S)=\sum_{i=1}^{|S|} \pi(S[i])$ for $S\in \hat{\Sigma}^*$.  
For two strings $X,Y\in \hat{\Sigma}^*$ of the same length, we define the \emph{prefix mismatch information}
\[\PMI(X,Y) = \{(i,\pi(X[1\dd i)),\pi(Y[1\dd i))) : i\in \MP(X,Y)\}.\]

\begin{proposition}\label{prp:mpi}
For every constant $\delta\in (0,1)$, there exists a sketch $\skL_k$ (parameterized by integers $n\ge k\ge 1$, an alphabet $\hat{\Sigma}=[1\dd \sigma]\times [0\dd n^{\Oh(1)}]$, and a seed of $\Oh(\log (n\log \sigma))$ random bits) such that:
\begin{enumerate}
	\item The sketch $\skL_k(S)$ of a string $S\in \hat{\Sigma}^n$ takes $\Oh(k\log^2 n)$ bits. Given streaming access to $S$, it can be constructed in $\Oh(n \log n \log(n\log \sigma) + n\log^2 n)$ time using $\Oh(k\log^2 n + \log n\log(n\log \sigma))$ bits of space.
	\item There exists a decoding algorithm that, given the sketches $\skL_k(X)$ and $\skL_k(Y)$ of strings $X,Y\in \hat{\Sigma}^{n}$
	satisfying $\pi(X),\pi(Y)< n$, with probability at least $1-\delta$ either returns $\PMI(X,Y)$ or certifies that $\hd(X,Y)>k$.
	The algorithm uses $\Oh(k \log^4 n)$ time and $\Oh(k \log^2 n)$ bits of space. 
\end{enumerate}
\end{proposition}
\begin{proof}
Let us fix a complete binary tree $\Tr$ with $n$ leaves, numbered with $[1\dd n]$ in the left-to-right order,
and let $v_1,\ldots,v_{2n-1}$ denote the nodes of $\Tr$ in the pre-order. For each node $v_i\in \Tr$, let $[p_i\dd q_i)\sub [1\dd n]$ be the indices of the leaves in the subtree of $v_i$.
Consider the fingerprints $\psi$ of \cref{fact:kr} parameterized by $\delta_{\ref{fact:kr}}=\frac{\delta}{4n-2}$, $n_{\ref{fact:kr}}= \Oh(n\log (n\sigma))$, and $\sigma_{\ref{fact:kr}}=2$:
Given a string $S\in \hat{\Sigma}^n$, we define a string $\Tr(S)[1\dd 2n)$
so that
$\Tr(S)[i] = (\pi(S[p_i\dd q_i)), \psi(S[p_i\dd q_i)))$ for every $i\in [1\dd 2n)$, where $\psi$ expands each character of $\hat{\Sigma}$ into a sequence of $\Oh(\log (n\sigma))$ bits.
We set $\skL_k(S):=\skH_{k_{\ref{thm:skH}}}(\Tr(S))$,
where $\skH_{k_{\ref{thm:skH}}}$ is the sketch of \cref{thm:skH}
with $\delta_{\ref{thm:skH}}=\frac{\delta}{2}$, $n_{\ref{thm:skH}}=2n-1$, $k_{\ref{thm:skH}}=\min(k\lceil{\log (2n)}\rceil,n_{\ref{thm:skH}})$, and $\sigma_{\ref{thm:skH}}\le 2^{\Oh(\log (n / \delta_{\ref{fact:kr}}))}\le n^{\Oh(1)}$.

This construction uses $\Oh(\log n_{\ref{fact:kr}} + \log (n_{\ref{thm:skH}}\log \sigma_{\ref{thm:skH}})) = \Oh(\log (n\log \sigma))$ random bits and produces sketches of $\Oh(k_{\ref{thm:skH}}\log(n_{\ref{thm:skH}}\sigma_{\ref{thm:skH}}))=\Oh(k\log^2 n)$ bits.

The encoding algorithm transforms $S$ into $\Tr(S)$ and feeds $\Tr(S)$ into the encoding procedure of \cref{thm:skH}. 
Construction of $\Tr(S)$ is organized in $\Oh(\log n)$ layers,
each responsible for nodes $v_i\in \Tr$ at a fixed level.
The intervals $[p_i\dd q_i)$ corresponding to these nodes are disjoint so, at any time, a layer produces a single character $\Tr(S)[i]$ and, by \cref{fact:kr}, spends $\Oh(n\log \sigma)$ amortized time and uses $\Oh(\log n)$ bits of space to process a single character $S[j]$.
Since the tree $\Tr$ is linearized in the post-order fashion, all levels can read the string $S$ with a common left-to-right pass and outputting $\Tr(S)$ does not require any buffers.
Overall, construction of $\Tr(S)$ takes $\Oh(n \log n \log(n\log \sigma))$ time and uses $\Oh(\log n \log(n\log \sigma))$ bits of space.
The encoder of \cref{thm:skH} takes $\Oh(n_{\ref{thm:skH}}\log(n_{\ref{thm:skH}}\sigma_{\ref{thm:skH}})\log(n_{\ref{thm:skH}}\log \sigma_{\ref{thm:skH}})) = \Oh(n\log^2 n)$ time and uses $\Oh(k_{\ref{thm:skH}}\log(n_{\ref{thm:skH}}\sigma_{\ref{thm:skH}}))=\Oh(k\log^2 n)$ bits of space.

The decoding algorithm first retrieves $\MI(\Tr(X),\Tr(Y))$ from $\skH_{k_{\ref{thm:skH}}}(\Tr(X))$ and $\skH_{\ref{thm:skH}}(\Tr(Y))$ using the decoder of \cref{thm:skH}.
If the output of that subroutine can be interpreted as $\MI(\Tr(X),\Tr(Y))$, then we 
use \cref{clm:retr} below to retrieve $\MI(X,Y)$. 
If the size of the obtained set does not exceed $k$, we return the set.
In the remaining cases, we certify that $\hd(X,Y)>k$. 

\begin{claim}\label{clm:retr}
	For every $X,Y\in \hat{\Sigma}^n$, the set
	$\PMI(X,Y)$ can be extracted from $\MI(\Tr(X),\Tr(Y))$ in time and space $\Oh(\hd(X,Y)\log n)$  with success probability at least $1-\frac\delta2$.
\end{claim}
\begin{proof}
	We assume that, for every $i\in [1\dd 2n)$, the equality $\psi(X[p_i\dd q_i))=\psi(Y[p_i\dd q_i))$ implies $X[p_i\dd q_i) = Y[p_i\dd q_i)$.
	By \cref{fact:kr}, each of the implications fails with probability at most  $\delta_{\ref{fact:kr}}$,
	so the overall failure probability can be bounded by $\delta_{\ref{fact:kr}}\cdot (2n-1)= \frac\delta2$.

	Our assumption yields \[\MP(\Tr(X),\Tr(Y))=\{i\in [1\dd 2n) : [p_i\dd q_i) \cap \MP(X,Y)\ne \emptyset\}.\]
	In particular, $\MP(X,Y) = \{p_i : v_i \text{ is a leaf and } i \in \MP(\Tr(X),\Tr(Y))\}$.
	Thus, it remains to describe how to extract $\pi(X[1\dd p))$ and $\pi(Y[1\dd p))$ for each $p\in \MP(X,Y)$;
	by symmetry, we focus on $\pi(X[1\dd p))$.
	For this, we process subsequent nodes $v_i$ on the path from the root of $\Tr$ to the leaf $v_j$ such that $\{p\}=[p_j\dd q_j)$	maintaining $\pi(X[1\dd p_i))$. Note that the values $\pi(X[p_i\dd q_i))$ can be extracted 
	from $\MI(\Tr(X),\Tr(Y))$ due to $p\in [p_i\dd q_i)$.
	
	If $v_i$ is the root, then  $\pi(X[1\dd p_i)) = \pi(\eps) =0$.
	If $v_{i}$ is the left child of $v_{i'}$, then $p_i = p_{i'}$, so $\pi(X[1\dd p_i))=\pi(X[1\dd p_{i'}))$
	has already been computed. 
	If $v_{i}$ is the right child of $v_{i'}$, then $q_i = q_{i'}$, so $\pi(X[1\dd p_i)) + \pi(X[p_i\dd q_i))
	= \pi(X[1\dd q_i)) = \pi(X[1\dd q_{i'})) = \pi(X[1\dd p_{i'}))+ \pi(X[p_{i'}\dd q_{i'}))$. 
	Consequently, $\pi(X[1\dd p_i))$ can be retrieved from $\pi(X[1\dd p_{i'}))$, which has already been computed,
	as well as $\pi(X[p_i\dd q_i))$ and $\pi(X[p_{i'}\dd q_{i'}))$, which are available in $\MI(\Tr(X),\Tr(Y))$.
	When this process reaches $v_j$, it results in the sought value $ \pi(X[1\dd p)) = \pi(X[1\dd p_j))$.
\end{proof}

It remains to analyze correctness of the decoding algorithm.
The decoder of \cref{thm:skH} fails with probability at most $\delta_{\ref{thm:skH}}=\frac{\delta}{2}$
and the procedure of \cref{clm:retr} fails with probability at most $\frac{\delta}{2}$.
Thus, with at most $\delta$ probability loss, we may assume that both calls are successful.
 
If $\hd(\Tr(X),\Tr(Y))\le k_{\ref{thm:skH}}$, then the decoder of \cref{thm:skH} retrieves $\MI(\Tr(X),\Tr(Y))$
and the procedure of \cref{clm:retr} results in $\MI(X,Y)$.
Depending on whether $\hd(X,Y)=|\MI(X,Y)|\le k$ or not, our decoding algorithm thus correctly returns $\MI(X,Y)$ or certifies that $\hd(X,Y)>k$.

Otherwise, the algorithm of \cref{thm:skH} certifies $\hd(\Tr(X),\Tr(Y))> k_{\ref{thm:skH}}$,
and the whole decoding procedure certifies $\hd(X,Y)>k$. This is correct because of $\hd(\Tr(X),\Tr(Y))\le \hd(X,Y)\lceil{\log (2n)}\rceil$.  To see this, observe that $[p_i\dd q_i)\cap \MP(X,Y) = \emptyset$ implies $X[p_i\dd q_i)=Y[p_i\dd q_i)$ and $\Tr(X)[i] = \Tr(Y)[i]$.
However, every leaf of $\Tr$ has at most $\lceil{\log (2n)}\rceil$ ancestors (including itself).
Consequently, for every $j\in \MP(X,Y)$, there are at most $\lceil{\log (2n)}\rceil$ nodes $v_i$
such that $j\in [p_i\dd q_i)$.
\end{proof}

\subsection{Edit Distance Sketches}\label{sec:actualsketches}
We are now ready to show the main result of this section.

\thskE*
\begin{proof}
Let $c_{\ref{prp:alg}}$ be the constant of \cref{prp:alg} for $\delta_{\ref{prp:alg}}=\frac{\delta}{3}$,
and let $c_{\ref{lem:bcgk}}$ be the constant of \cref{lem:bcgk}.
We shall use $\bCGK_{k'}(W)$ with $k'=(c_{\ref{prp:alg}}+5)k$ and $W = \Alg(n,r,S)$, where $r$ is a random seed of $\Oh(\log^2 n)$ bits.
Observe that the alphabet of $\bCGK_{k'}(W)$ can be interpreted as $[0\dd n^{\Oh(k)})\times [0\dd n]$
because the $|\LZ(\rev{\Dtx{k'}{S}{s}})|=\Oh(k)$ for each $s\in [1\dd |S|]$.
We shall use \cref{thm:skH,prp:mpi} with $n_{\ref{thm:skH}}=n_{\ref{prp:mpi}}=3n$, $\delta_{\ref{thm:skH}}=\delta_{\ref{prp:mpi}}=\tfrac{\delta}{3}$, $k_{\ref{thm:skH}}=k_{\ref{prp:mpi}}=\min(n_{\ref{thm:skH}},\lfloor c_{\ref{lem:bcgk}}(1+c_{\ref{prp:alg}})k\log(3n)\rfloor)$, and $\sigma_{\ref{thm:skH}}=\sigma_{\ref{prp:mpi}} = n^{\Oh(k)}.$

The sketch $\skE_k(S)$ of a string $S\in \Sigma^{\le n}$ consists of $|S|$
as well as the sketches $\skH_{k_{\ref{thm:skH}}}(\bCGK_{k'}(W))$ and $\skL_{k_{\ref{prp:mpi}}}(\bCGK_{k'}(W))$.
This construction uses \[\Oh(\log^2 n+\log(n_{\ref{thm:skH}}\log \sigma_{\ref{thm:skH}})+\log (n_{\ref{prp:mpi}}\log \sigma_{\ref{prp:mpi}})) = \Oh(\log^2 n + \log(nk\log n))=\Oh(\log^2 n)\]
random bits and produces sketches of bit-size
\[\Oh(k_{\ref{thm:skH}}\log(n_{\ref{thm:skH}}\sigma_{\ref{thm:skH}}) + k_{\ref{prp:mpi}}\log^2 n_{\ref{prp:mpi}})
=\Oh(k\log n \log(n^{\Oh(k)})+k\log^3 n)=\Oh(k^2\log^3 n).\]

The encoding algorithm uses \cref{cor:bcgkp_construction} to transform the input stream representing $S$ to an auxiliary stream representing $\bCGK_{k'}(W)$, which we forward to the encoders constructing $\skH_{k_{\ref{thm:skH}}}(\bCGK_{k'}(S))$ and $\skL_{k_{\ref{prp:mpi}}}(\bCGK_{k'}(W))$.
Thus, it takes $\tOh(nk' + n_{\ref{thm:skH}}\log\sigma_{\ref{thm:skH}} + n_{\ref{prp:mpi}}) = \tOh(nk)$ time and uses $\tOh(k' + k_{\ref{thm:skH}}\log \sigma_{\ref{thm:skH}} + k_{\ref{prp:mpi}})=\tOh(k^2)$
space.

\paragraph*{Decoding Algorithm}
Let $\A_\Alg$ be the zip alignment of walks $W_X = (x_t)_{t=1}^{3n+1} = \Alg(n,r,X)$ and $W_Y = (y_t)_{t=1}^{3n+1} = \Alg(n,r,Y)$ over strings  $X,Y \in \Sigma^{\le n}$.

Given sketches  $\skE_{k}(X), \skE_{k}(Y)$,
we run the decoders for $\skH_{k_{\ref{thm:skH}}}(\bCGK_{k'}(W_X)),\skH_{k_{\ref{thm:skH}}}(\bCGK_{k'}(W_Y))$ and $\skL_{k_{\ref{prp:mpi}}}(\bCGK_{k'}(W_X)),\skL_{k_{\ref{prp:mpi}}}(\bCGK_{k'}(W_Y))$. We certify $\ed(X,Y)>k$ if either procedure certifies $\hd(\bCGK_{k'}(W_X),\bCGK_{k'}(W_Y))>k_{\ref{thm:skH}}=k_{\ref{prp:mpi}}$. Otherwise, the decoder interprets the outputs as $\MI(\bCGK_{k'}(W_X),\bCGK_{k'}(W_Y))$ and $\PMI(\bCGK_{k'}(W_X),\bCGK_{k'}(W_Y))$, respectively.
In this case, we construct $\Enc^M(X,Y)$ for $M = \{(x,y)\in \mtch_{X,Y}(\A_\Alg) : x\in P_X\text{ or }y\in P_Y\}$,
where $P_X$ and $P_Y$ are defined in the statement of \cref{lem:bcgk}. Finally, we compute $\gr_k(X,Y)$ using \cref{prp:greedify} or $\min(\ed(X,Y),k+1)$ using \cref{cor:ED}.

\begin{claim}\label{clm:decoder}
Given  $\MI(\bCGK_{k'}(W_X),\bCGK_{k'}(W_Y))$ and $\PMI(\bCGK_{k'}(W_X),\bCGK_{k'}(W_Y))$,
the encoding $\Enc^M(X,Y)$ can be constructed in $\tOh(k^3)$ time using $\tOh(k^2)$ space.
Moreover, $|\LZ(X^MY^M)|=\tOh(k^2)$.
\end{claim}
\begin{proof}
We proceed in two phases. In the first phase, we construct the compressed representation $\CR(X')$ of a string 
$X'\in (\Sigma\cup \{\#\})^{|X|}$ such that $X'[x]=X[x]$ if $x\in P_X$ and $X'[x]=\#$ otherwise.
We initialize $X':=\#^{|X|}$ (via $\LZ(\#^{|X|})$ using \cref{prp:RLSLP}\ref{it:fromLZ}).
Next, we iterate over positions $t\in [1\dd 3n]$ such that $\bot \ne \bCGK_{k'}(W_X)[t]\ne \bCGK_{k'}(W_Y)[t]$.
We retrieve $x_t$ and $\LZ\big(\rev{\Dtx{k'}{X}{x_t}}\big)$ from the mismatch information,
convert $\LZ\big(\rev{\Dtx{k'}{X}{x_t}}\big)$ to $\CR(\Dtx{k'}{X}{x_t})$ (\cref{prp:RLSLP}\ref{it:fromLZ}), and update $\CR(X')$, setting $X'[x_t\dd x_t+|\Dtx{k'}{X}{x_t}|):=\Dtx{k'}{X}{x_t}$ (\cref{prp:RLSLP}\ref{it:extract}\ref{it:concat}). We also symmetrically construct the compressed representation $\CR(Y')$ of a string 
$Y'\in (\Sigma\cup \{\#\})^{|Y|}$ such that $Y'[y]=Y[y]$ if $y\in P_Y$ and $Y'[y]=\#$ otherwise.

In the second phase, we convert $\CR(X')$ to $\CR(X^M)$. Here, the goal is make sure that $X^M[x]=\#$ 
not only for $x\in P_X$, but also when $X[x]\simeq_{\A_\Alg} Y[y]$ for some $y\in P_Y$.
For this, we iterate over dummy segments $Y'[y\dd y')$.
By \cref{lem:bcgk}, $\A_\Alg$ matches $Y[y\dd y')$ to a fragment of $X[x\dd x')$.
Hence, we shall identify the shift $x-y$ and set $X^M[x\dd x') :=\#^{x'-x}$ (\cref{prp:RLSLP}\ref{it:fromLZ}\ref{it:extract}\ref{it:concat}).
Let $\MP_{\bCGK} :=  \MP(\bCGK_{k'}(W_X),\bCGK_{k'}(W_Y))$.
Define $u,v\in [1\dd 3n+1]$ so that $u=v=3n+1$ if  $\pi(\bCGK_{k'}(W_Y)[1\dd t]) \ne y'$ for all $t\in \MP_{\bCGK}$; otherwise, $u$ is the smallest index  in $\MP_{\bCGK}$
with $\pi(\bCGK_{k'}(W_Y)[1\dd u])=y'$, whereas $v$ is the smallest index $\MP_{\bCGK}$ with $\pi(\bCGK_{k'}(W_Y)[1\dd v))=\pi(\bCGK_{k'}(W_Y)[1\dd u))$.
We claim that $x-y=\pi(\bCGK_{k'}(W_X)[1\dd v))-\pi(\bCGK_{k'}(W_Y)[1\dd v))$.

If $y'\le |Y|$, then $y'-1 \notin P_Y$ and $y'\in P_Y$. Hence, $Y[y']$ is covered by $\Dtx{k'}{Y}{y_t}$ for some $t\in \MP_{\bCGK}$ with $\bCGK_{k'}(W_X)[t]\ne \bot$, whereas $Y[y'-1]$ is not covered by $\Dtx{k'}{Y}{y_t}$.
Hence, there is $t\in \MP_{\bCGK}$ with $y_t = y'$, and therefore $y_u = y'$.
Similarly, $y_u=y'$ holds if $y'=|Y|+1$.
Let $w\in [1\dd 3n]$ be the smallest index such that $\bCGK_{k'}(W_Y)[w]\ne \bot$
and $\Dtx{k'}{Y}{y_t}$ covers $Y[y'-1]$ (such an index exists by \cref{lem:cover}).
Due to $y'-1\in P_Y$, we have $\bCGK_{k'}(W_X)[w]=\bCGK_{k'}(W_Y)[w]$, and hence $x-y=x_w-y_w$.
Definition of $w$ further yields $\bCGK_{k'}(W_Y)[w+1\dd v) = \bot^{u-w-1}$,
and thus $\pi(\bCGK_{k'}(W_Y)[1\dd v))=y_w$.
Moreover, since $\pi(\bCGK_{k'}(W_Y)[1\dd t))=y_w$ for $t\in [w+1\dd v]$,
we have $\pi(\bCGK_{k'}(W_X)[1\dd v))=x_w$ by definition of $v$.
This completes the proof that $x-y=\pi(\bCGK_{k'}(W_X)[1\dd v))-\pi(\bCGK_{k'}(W_Y)[1\dd v))$.

To derive $\Enc^M(X,Y)$, it suffices to convert  $\CR(Y')$ to $\CR(Y^M)$ (symmetrically),
and to construct $\CR(X^MY^M)$ using \cref{prp:RLSLP}\ref{it:concat}.
Since $\maxLZ(\Dtx{k'}{X}{x_t}),\maxLZ(\Dtx{k'}{Y}{y_t})=\Oh(k)$ holds for all $t\in [1\dd 3n+1]$
and since $|\MP_{\bCGK}|=\Oh(k\log n)$, the $\maxLZ(\cdot)$ measure of all intermediate
strings is $\tOh(k^2)$. Consequently, the $\tOh(k)$ applications of \cref{prp:RLSLP}
cost $\tOh(k^3)$ time and use $\tOh(k^2)$ space.
\end{proof}

To complete the complexity analysis, observe that the decoding procedure of \cref{thm:skH}  uses $\Oh(k_{\ref{thm:skH}} \log (n_{\ref{thm:skH}}\sigma_{\ref{thm:skH}}) \log^2(n_{\ref{thm:skH}}\log \sigma_{\ref{thm:skH}})) = \tOh(k^2)$ time and $\Oh(k_{\ref{thm:skH}} \log (n_{\ref{thm:skH}}\sigma_{\ref{thm:skH}})) = \tOh(k^2)$ bits of space.
The procedure of \cref{prp:mpi} uses $\Oh(k_{\ref{prp:mpi}} \log^2 n_{\ref{prp:mpi}}) = \tOh(k)$ bits of space and costs $\Oh(k_{\ref{prp:mpi}}  \log^4 n_{\ref{prp:mpi}})= \tOh(k)$ time.
Finally, due to $|\LZ(X^MY^M)|=\tOh(k^2)$ and $|\brkp_{k}(X,Y)|=\Oh(k^5)$ (\cref{lm:greedyalg}),
the algorithm of \cref{prp:greedify} uses $\tOh(k^2)$ space and costs $\tOh(k^5)$ time (dominating the overall decoding complexity).
If we only aim to retrieve $\min(\ed(X,Y),k+1)$, the algorithm of \cref{cor:ED} takes $\tOh(k^2)$ time and space (in which case the overall decoding uses $\tOh(k^2)$ space and costs $\tOh(k^3)$ time.

It remains to argue that the decoding algorithm is correct.
With $\delta_{\ref{prp:alg}}=\frac{\delta}{3}$ probability loss, we may assume that $\A_\Alg\in \ga_{c_{\ref{prp:alg}}(\ed(X,Y))^2}(X,Y)$ and $\width(\A_{\Alg})\le c_{\ref{prp:alg}}\ed(X,Y)$.
With $\delta_{\ref{thm:skH}}+\delta_{\ref{prp:mpi}}=\frac{2\delta}{3}$ further probability loss,
we may assume that the decoders of \cref{thm:skH,prp:mpi} are successful.
If $\ed(X,Y)>k$, then the correctness follows from \cref{prp:greedify,cor:ED} because 
$M\sub \mtch_{X,Y}(\A_\Alg)$ is a non-crossing matching.
Otherwise, \cref{prp:alg} guarantees $\cost(\A_\Alg)\le c_{\ref{prp:alg}}k^2$ and $\width(\A_{\Alg})\le c_{\ref{prp:alg}}k$ so, in particular, $k'\ge \width(\A_\Alg)+5k$.
By \cref{lem:bcgk}, we thus have $\hd(\bCGK_{k'}(W_X),\bCGK_{k'}(W_X))\le c_{\ref{lem:bcgk}}(1+c_{\ref{prp:alg}})k\log(3n)$. Hence, the decoders of \cref{thm:skH,prp:mpi} report the mismatch information.
\cref{lem:close} further implies that $\Delta_X(\A,\A_\Alg)\sub P_X$ and  $\Delta_Y(\A,\A_\Alg)\sub P_Y$
for all $\A\in \ga_k(X,Y)$. In particular, $\mtch_{X,Y}(\A)\sub M$, and therefore $\mtch_k(X,Y)\sub M$.
Consequently, the algorithms of \cref{prp:greedify,cor:ED} correctly compute $\gr_k(X,Y)$
and $\min(\ed(X,Y),k+1)$, respectively.
\end{proof}

Next, we boost the success probability and strengthen the sketches so that we can retrieve $\qgr_k(X,Y)$ instead of $\gr_k(X,Y)$. 

\begin{restatable}{corollary-restatable}{corskE}\label{cor:skq}
	There is a sketch $\skq_k$ (parametrized by $\delta\in(0,\frac12)$, integers $n\ge k \ge 1$, an alphabet $\Sigma = [0\dd n^{\Oh(1)})$, and a seed of $\Oh(\log^2 n \log(1/\delta))$ random bits) such that:
	\begin{enumerate}[label=\textrm{(\alph*)}]
		\item The sketch $\skq_k(S)$ of a string $\Sigma^{\le n}$ takes $\Oh(k^2 \log^3 n \log(1/\delta))$ bits. Given streaming access to $S$, it can be constructed in $\tOh(nk \log(1/\delta))$ time using $\tOh(k^2 \log(1/\delta))$ space.\label{it:skq_construction}
		\item There exists a decoding algorithm that, given $\skq_{k}(X), \skq_{k}(Y)$ for $X,Y \in \Sigma^{\le n}$,
		with probability at least $1-\delta$ computes $\qgr_k(X,Y)$.
		The algorithm takes $\tOh(k^5 \log(1/\delta))$ time and uses $\tOh(k^2 \log(1/\delta))$ space.\label{it:skq_qgr}
	\end{enumerate}
\end{restatable}
\begin{proof}
We shall use $\mu = \Oh(\log(1/\delta))$ sketches $\skE_{k_{\ref{thm:ske}}}$ of \cref{thm:ske}
with $\delta_{\ref{thm:ske}}=\frac13$, $n_{\ref{thm:ske}}=n+1$, $k_{\ref{thm:ske}}=k+1$, and independent seeds.
For each of the $\mu$ sketches  $\skE_{k_{\ref{thm:ske}}}$,
the sketch $\skq_k(S)$ contains $\skE_{k_{\ref{thm:ske}}}(\$_1S)$, $\skE_{k_{\ref{thm:ske}}}(\$_2S)$,
where $\$_1,\$_2\notin \Sigma$ are two distinct symbols.
This construction uses $\Oh(\mu \log^2 n_{\ref{thm:ske}})=\Oh(\log^2 n \log(1/\delta))$ random bits
and produces sketches of $\Oh(\mu k^2_{\ref{thm:ske}}\log^3 n_{\ref{thm:ske}})=\Oh(k^2 \log^3 n \log(1/\delta))$ bits. 

The encoding algorithm calls $2\mu$ instances of the encoding algorithm of \cref{thm:ske}.
Hence, it uses $\tOh(\mu k^2_{\ref{thm:ske}})=\Oh(k^2 \log(1/\delta))$ space and costs
$\tOh(\mu n_{\ref{thm:ske}} k_{\ref{thm:ske}})=\tOh(nk \log(1/\delta))$ time.

The decoding algorithm, given $\skq_{k}(X), \skq_{k}(Y)$ for $X,Y \in \Sigma^{\le n}$,
runs the decoder of \cref{thm:ske} for $\skE_{k_{\ref{thm:ske}}}(\$_1X),\skE_{k_{\ref{thm:ske}}}(\$_1Y)$
for each of the $\mu$ sketches  $\skE_{k_{\ref{thm:ske}}}$.
This yields $\mu$ candidates for $\gr_{k_{\ref{thm:ske}}}(\$_1X,\$_2Y)$,
which we convert to candidates for $\qgr_k(X,Y)$ using \cref{cor:qgr_to_gr}.
Finally, we determine the majority answer among the $\mu$ candidates.
The equality test uses \cref{prp:RLSLP}\ref{it:concat}\ref{it:lce}
to compare two candidates for $\CR(X^MY^M)$.

Recall that the decoding procedures of \cref{thm:ske} use $\tOh(\mu k^5_{\ref{thm:ske}})=\tOh(k^5 \log(1/\delta))$ time and $\tOh(\mu k^2_{\ref{thm:ske}})=\tOh(k^2 \log(1/\delta))$
space.
The applications of \cref{cor:qgr_to_gr} and the equality tests take  $\tOh(\mu k^2_{\ref{thm:ske}})=\tOh(k^2 \log(1/\delta))$ time and space. The entire decoding algorithm uses $\tOh(k^2 \log(1/\delta))$ space and $\tOh(\mu k^5_{\ref{thm:ske}})=\tOh(k^5 \log(1/\delta))$ time.

As for correctness, since the $\mu$ sketches  $\skE_{k_{\ref{thm:ske}}}$ are independent,
by the Chernoff bound, the majority answer is wrong with probability at most $\exp(-\Oh(\mu))$.
Setting $\mu = \Oh(\log (1/\delta))$ (with a sufficiently large constant factor) guarantees a success probability of $1-\delta$.
\end{proof}

%% file: rpst.tex
In this section, we present solutions for pattern matching with $k$ edits in the semi-streaming and streaming settings. 

\subsection{Periodicity under Edit Distance}
We start by recalling combinatorial properties of strings periodic under the edit distance.

\defkperiodic*

\begin{claim}\label{prop:period_prefixes}
Suppose that a string $X$ is a prefix of a string $Y$, where $|X| < |Y| \le 2|X|$. If $X$ is $k$-periodic with $k$-period $Q$, $|Q| \le |X| /  128k$, then either $Y$ is not $k$-periodic, or $Y$ is $k$-periodic with $k$-period $Q$.  
\end{claim}
\begin{proof}
\newcommand{\suff}{\mathsf{suff}}
\newcommand{\pref}{\mathsf{pref}}
Suppose by contradiction that $Y$ is $k$-periodic with $k$-period $Q' \neq Q$. Let $q = |Q|$ and $q' = |Q'|$. Assume first $q \le q'$. Fix an alignment of the smallest cost between $Y$ and a prefix of $(Q')^\infty$. It induces an alignment $\A'$ of cost at most $2k$ between $X$ and a prefix of $(Q')^\infty$ and hence generates a partition $X = X_1 X_2 \ldots X_z$, where each $X_i$, $1 \le i \le z-1$, is aligned with $Q'$ and $X_z$ is aligned with a prefix of~$Q'$. From $|Q| \le |X| /  128k$ we obtain $|X|\ge 128 k$ and therefore  
$$z \ge (|X|-2k)/q' \ge \frac{128-2}{128} |X| / (|Y|/128 k) \gg 20 k$$
Consider fragments $X_1 X_2 X_3 X_4$, $X_5 X_6 X_7 X_8$, and so on. The total number of such fragments is at least $(20k-3) / 4 > 4k$, and at least $2k+1$ of them are not edited under $\A'$. Fix an optimal alignment $\A$ between $X$ and a prefix $Q^\infty$. The cost of $\A$ is bounded from above by $2k$, and therefore there is at least one fragment $X_{4i+1}X_{4i+2}X_{4i+3}X_{4i+4}$ that is not edited. Consider one such fragment $F$. On the one hand, $F = Q' Q' Q' Q'$. On the other hand, $F = \suff(Q) Q^j \pref(Q)$, where $\suff(Q)$ and $\pref(Q)$ are some suffix and some prefix of $Q$, respectively. 

Suppose first that $q'$ is a multiple of $q$. In this case, $Q' = \left(Q[\ell+1\dd q] Q[1 \dd \ell]\right)^r$, where $\ell = q-|\suff(Q)|$ and an integer $r$, which contradicts the fact that $Q'$ is primitive. Otherwise, consider the copy of $Q$ that contains $F[q']$. Consider also a substring $Q Q Q$ of $F$ formed by the copy of $Q$ that contains $F[2q']$, and the preceding and succeeding copies of $Q$. We then obtain that there is an occurrence of $Q$ in $QQQ$ that is not aligned with any copy of $Q$ (otherwise, $q'$ is a multiple of $q$), and therefore $Q$ is not primitive, a contradiction.   

The case $q > q'$ can be treated analogously.
\end{proof}

Note that \cref{prop:period_prefixes} implies in particular that a string can have at most one $k$-period.

\subsection{Semi-streaming Algorithm}\label{sec:rpst}
We first present a deterministic algorithm for pattern matching with $k$ edits in the semi-streaming setting. 


\subsubsection{Preprocessing Stage} 
Consider a set $\Pi$ of $\Oh(\log m)$ prefixes $P_i$ of $P$ initialised to contain $P$ itself as well as the prefixes of length $2^\ell$ for all $0 \le \ell \le \lfloor{\log |P|} \rfloor$. Order the prefixes by lengths, and consider two consecutive prefixes $P', P''$. If $P'$ is $k$-periodic with $k$-period $Q'$ while~$P''$ is not $k$-periodic, we add two more prefixes to $\Pi$. Namely, if~$\ell$ be the maximum integer such that $P[1\dd\ell]$ is $k$-periodic with $k$-period $Q'$, add to $\Pi$ the prefixes $P[1\dd\ell]$ and $P[1\dd\ell+1]$. Note that $P[1\dd\ell+1]$ is not $k$-periodic by~\cref{prop:period_prefixes}. 

Let $\Pi = \{P_1, P_2, \ldots, P_z\}$ be the resulting set of prefixes. We assume that the patterns are ordered in the ascending order of their lengths. During the preprocessing step, for each $i$ such that $P_i$ is $k$-periodic, we compute its $k$-period~$Q_i$. We use notation $\ell_i = P_i$ and $q_i = |Q_i|$ (if defined). Importantly, we do not have to store $Q_i$ explicitly, we can simply memorize its endpoints in $P_i$ which takes $\Oh(\log m)$ extra space in total. We also store, for each of the $\Oh(k)$ rotations $D$ of $Q_i$ that can be a difference of a chain of $k$-edit occurrences of $P_i$ (\cref{cor:structure-period}), the encodings $\qgr_{32k} (D, D)$ and $\qgr_{30k} (P[\ell_{i-1}+1 \dd \ell_i], D^\infty[1\dd\Delta_i])$, where $\Delta_i = \ell_i - \ell_{i-1}+k$. 

\subsubsection{Main Stage}
The main stage of the algorithm starts after we have preprocessed the pattern. During the main stage, we receive the text as a stream, one character of a time. 
We exploit the following result:

\begin{fact}[cf.~\cite{DBLP:journals/iandc/Ukkonen85}]\label{fct:dynprog}
Given a read-only string $X$ of length $m$ and a streaming string $Y$. There is a dynamic programming algorithm that correctly identifies all prefixes $Y'$ of $Y$ within edit distance at most $k \le m$ from $X$, as well as $\ed(Y', X)$ itself. The algorithm takes $\Oh(km)$ time and $\Oh(k)$ space besides the space required to store $X$. 
\end{fact}

\paragraph{Chains of $k$-edit occurrences.}  
During the main stage of the algorithm, we store the following information. 
Let $r$ be the newly arrived position of the text $T$. For each~$2 \le i \le z$, consider all $k$-edit occurrences of $P_{i-1}$ in $T[r- \ell_{i} - k +1 \dd r]$. We call such occurrences \emph{active}. We denote the set of right endpoints of the active occurrences by $\aOcc_k^E(P_{i-1}, T)$. By~\cref{cor:structure-nonperiod,cor:structure-period}, $\aOcc_k^E(P_{i-1}, T)$ forms $\Oh(k^3)$ chains. For each chain~$\chain$, we store the following information:
\begin{enumerate}
\item The leftmost position $lp$ and the size $|\chain|$ of $\chain$;
\item An integer $\ed(\chain)$ equal to the smallest edit distance from $P_{i-1}$ to a suffix of $T[1\dd r]$ for every endpoint $r \in \chain$;
\item If $|\chain| \ge 2$, we store the shift $\Delta$ of the difference $D = Q_{i-1}[1+\Delta \dd q_{i-1}] Q_{i-1}[1 \dd q_{i-1}]$ of $\chain$. 
\end{enumerate}

If $p^\ast$ is the first position added to $\chain$ (it can be different from $lp$ as we will update chains to contain only active occurrences), then at the position $(p^\ast+1)$ we start running the dynamic programming algorithm for $T[p^\ast+1\dd]$ and $P[\ell_{i-1}+1\dd \ell_{i}]$ (\cref{fct:dynprog}). 

Furthermore, consider the moment when we detect the second position in $\chain$ (if it exists) and hence the difference $D$ of the chain. Starting from this moment, for every newly added position $p \in \chain$, at the position $(p+1)$ we start computing the greedy encoding $\qgr_{32k} (T[p+1 \dd p+\Delta_i], D^\infty[1\dd\Delta_i])$. We continue running the algorithm until either the computation is over or a new position in the chain is detected. In the end, we compute the encoding for the rightmost position in the chain.

\paragraph{Detecting new $k$-edit occurrences of $P_i$.} 
We now explain how to detect new $k$-edit occurrences of the prefixes $P_i$. Let $r$ be the latest arrived position of $T$. If $i = 1$, then since $k \ge 1$, $r \in \aOcc_k^E(P_1, T)$. Below we consider three possible cases for $i \ge 2$: $P_{i-1}$ is $k$-periodic, $P_i$ is not $k$-periodic; $P_{i-1}$ is not $k$-periodic; $P_{i-1}$ and $P_i$ are $k$-periodic. 

\subparagraph{Case 1: $P_{i-1}$ is $k$-periodic, $P_i$ is not $k$-periodic.} By construction, in this case $\ell_{i-1} + 1 = \ell_i$. The position $r \in \aOcc_k^E(P_{i}, T)$ iff one of the following conditions is satisfied:
\begin{enumerate}
\item The smallest edit distance from $P_{i-1}$ to a suffix of $T[1\dd r]$ is at most $k-1$ (this corresponds to the case when the last character of $P_i$ is deleted in an optimal alignment of a suffix of $T[1 \dd r]$ and $P_i$);
\item The smallest edit distance from $P_{i-1}$ to a suffix of $T[1\dd r-1]$ is at most $k-1$ and $P_i[\ell_i] \neq T[r]$ (this corresponds to the case when the last character of $P_i$ is substituted for $T[r]$ in an optimal alignment of a suffix of $T[1 \dd r]$ and $P_i$);
\item The smallest edit distance from $P_{i-1}$ to a suffix of $T[1\dd r-1]$ is at most $k$ and $P_i[\ell_i] = T[r]$ (this corresponds to the case when the last character of $P_i$ is matched with $T[r]$ in an optimal alignment of a suffix of $T[1 \dd r]$ and $P_i$);
\item The smallest edit distance from $P_{i}$ to a suffix of $T[1\dd r-1]$ is at most $k-1$ (this corresponds to the case when $T[r]$ is deleted in an optimal alignment of a suffix of $T[1 \dd r]$ and $P_i$).
\end{enumerate}
We can decide which of the conditions is satisfied, and therefore whether $r \in \aOcc_k^E(P_{i}, T)$ in $\Oh(k^3)$ time using $\aOcc_k^E(P_{i-1},T)$ and $\aOcc_k^E(P_{i},T)$. Moreover, we can compute the smallest edit distance from $P_i$ to a suffix of $T[1\dd r]$ if it is bounded by $k$. 

For the next two cases, we will use the following simple observation that follows from \cref{fact:edk2}:
\begin{observation}
Let $\ed_{i-1}(r')$ be the smallest edit distance from $P_{i-1}$ to a suffix of $T[1\dd r']$, and define
\begin{equation}\label{eq:ed}
d = \min_{\substack{r' \in \aOcc_k^E(P_{i-1}, T) \\ r' \in [r+1-\Delta_i, r+1-\Delta_i+2k]}} \{\ed_{i-1}(r') + \ed(P[\ell_{i-1}+1 \dd \ell_i], T[r'+1 \dd r])\}
\end{equation}
The smallest edit distance from $P_i$ to $T[1\dd r]$ is equal to $d$ if $d \le k$ and is larger than $k$ otherwise.
\end{observation}
It follows that to decide whether $r \in \aOcc_k^E(P_i, T)$, it suffices to compute the value $\min\{d, k+1\}$, where $d$ is as defined above. 

\subparagraph{Case 2: $P_{i-1}$ is not $k$-periodic.} In this case, $\aOcc_k^E(P_{i-1},T)$ is stored as $\Oh(k^3)$ chains of size one. Therefore, we can find the positions $r'$ from~\cref{eq:ed} in $\Oh(k^3)$ time. Moreover, for each position $r'$, we run the dynamic programming algorithm for $T[r'+1\dd]$ and $P[\ell_{i-1}+1\dd \ell_i]$, which outputs the edit distance between $T[r'+1\dd r]$ and $P[\ell_{i-1}+1\dd \ell_i]$ if it is at most $k$. As we also know the smallest edit distance between $P_{i-1}$ and a suffix of $T[1 \dd r']$, we can compute $d$ in $\Oh(k^3)$ time. 

\subparagraph{ Case 3: $P_{i-1}$ and $P_i$ are $k$-periodic.} 
We can identify all positions $r'$ from~\cref{eq:ed} in $\Oh(k^3)$ time. (It suffices to check each of the $\Oh(k^3)$ chains that we store for $P_{i-1}$). We must now test each of these positions. Consider a position $r'$ and let $\chain$ be the chain containing it. It suffices to compute the edit distance between $P[\ell_{i-1}+1\dd\ell_i]$ and $T[r'+1\dd r]$ as we already know the smallest edit distance from $P_{i-1}$ to a suffix of $T[1 \dd r']$. If $|\chain| = 1$, the distance has been computed by the dynamic programming algorithm.  Otherwise, we use quasi-greedy encodings. On a high level, our goal is to compute the edit distance between $\pi = P[\ell_{i-1} \dd \ell_i]$ and $\tau = T[r'+1 \dd r]$ via a string $\mu = D^\infty [1 \dd \Delta_i]$, where $D$ is the difference of $\chain$ and $\Delta_i = \ell_i-\ell_{i-1}+k$. 

\begin{lemma}\label{lm:dist_to_Y}
There is $\ed(\pi,\mu) \le 26k$.
\end{lemma}
\begin{proof}
As $P_{i-1}$ and $P_i$ are $k$-periodic, by~\cref{prop:period_prefixes} we obtain that $P_i = P[1\dd\ell_i]$ is $k$-periodic with $k$-period $Q_i = Q_{i-1}$, that is, there is a prefix of $Q_i^\infty$ such that the edit distance between it and $P_i$ is at most $2k$. By~\cref{fact:edk2}, there is a substring $Q_i^\infty[r \dd t]$ such that $|r-\ell_{i-1}| \le 2k$ and $|t-\ell_{i}| \le 2k$ and $\ed(Q_i^\infty[r \dd t], \pi) \le 2k$. By the triangle inequality, we obtain that $\ed(Q_i^\infty[\ell_{i-1}+1\dd\ell_i], \pi) \le 6k$. Let $a = \ell_{i-1}-7k \pmod{q_{i}}$ and $b = \ell_{i-1}+7k \pmod{q_{i}}$. By~\cref{cor:structure-period}, $D$ is a rotation of $Q_{i-1} = Q_i$ with shift $\Delta$, where $\Delta \in [a-3k,b+3k]$ if $a \le b$ and $\Delta \in [0, b+3k] \cup [a-3k, q_i)$ if $a > b$. It follows that $D^\infty = Q_i^\infty[s \dd]$, where $|s - \ell_{i-1}| \le 10k$. As $\mu = D^\infty[1\dd \Delta_i] = Q_i^\infty[s \dd s+\Delta_i-1]$, by the triangle inequality we obtain $\ed(\mu, Q_i^\infty[\ell_{i-1}+1\dd\ell_i]) \le 20k$. Applying the triangle inequality one more time, we obtain the claim.
\end{proof}

Let $\G_P = \qgr_{30k} (\pi, \mu)$ and $\G_T = \qgr_{30k} (\mu, \tau)$. By~\cref{cor:qgr_to_ed}, knowing $\G_P$ and $\G_T$ is sufficient to compute the edit distance between $\pi$ and $\tau$. Note that we do not know $\G_T$ yet, we must compute it using the available information. Let $p$ be the rightmost position in $\chain$. 

\begin{enumerate}
\item Recall that at the position $(p+1)$ we launched an algorithm that is computing $\qgr_{32k} (T[p+1 \dd p+\Delta_i], D^\infty[1\dd\Delta_i])$ with a delay of $k$ characters  (\cref{cor:greedy_long}). We have $p+\Delta_i \ge r' + \Delta_i \ge r$. Therefore, upon reaching $r$, we can use the memory of the algorithm to compute $\G_{T,1} = \qgr_{32k} (T[p+1 \dd r], D^\infty[1\dd r-p])$ in $\tOh(k^5)$ time and $\tOh(k^2)$ space (\cref{lm:greedy_concatenation}).

\item By the definition of chains, $T[r'+1\dd p] = D^j$ for some integer $j$. By~\cref{lm:greedy_concatenation}, we can use $\qgr_{32k}(D, D)$ computed during the preprocessing step to compute $\G_{T,2} = \qgr_{32k}(T[r'+1\dd p], D^j)$ in $\tOh(k^5)$ time and $\tOh(k^2)$ space. By applying~\cref{lm:greedy_concatenation} again, we can compute $\qgr_{32k}(T[r'+1\dd p], D^j D^\infty[1\dd r-p])$ in $\tOh(k^5)$ time and $\tOh(k^2)$ space.

\item Finally, we compute via~\cref{cor:greedy_long} the encoding $\qgr_{30k} (\eps, D^\infty[r-p+1\dd \Delta_i])$, where $\eps$ is the empty string and $\Delta_i-(r-p) \le 2k$. We then apply~\cref{obs:qgr_larger} to compute $\qgr_{30k+(\Delta_i-(r-p))}(T[r'+1\dd p], D^j D^\infty[1\dd r-p])$ and further~\cref{lm:greedy_concatenation} to compute $\qgr_{30k}(T[r'+1\dd p], D^\infty[1\dd \Delta_i]) = \G_T$.
\end{enumerate}

\paragraph{Updating the chains.}
When we detect a new $k$-edit occurrence of $P_i$, we must decide if it should be added to some existing chain or if we must create a new chain for this occurrence. 

To this end, for each $1 \le i \le z$, we consider $\Oh(k)$ of its rotations of $Q_i$ that can be the difference of a chain of $k$-edit occurrences of $P_i$ in $T$ (\cref{cor:structure-period}). For each rotation $R$, we run a constant-space and linear-time deterministic pattern matching algorithm~\cite{RYTTER2003763}. The algorithm processes the text $T$ as a stream and if there is an occurrence $T[\ell\dd r]$ of the rotation reports it while reading $T[r]$. The algorithm uses~$\Oh(1)$ space and $\Oh(1)$ amortised time per character of~$T$.

Suppose that we detect a new right endpoint $r$ of a $k$-edit occurrence $T[\ell \dd r]$ of $P_i$. We must decide whether $r$ belongs to an existing chain of $k$-edit occurrences of $P_i$ or starts a new one. In order to do this, we first find the chain $\chain$ that contains $r-q_i+1$ if it exists by checking each chain in turn. We then check that the smallest edit distance from a suffix of $T[1\dd r]$ to $P_i$ equals to  $\ed(\chain)$ and that $T[r-q_i+1\dd r]$ is equal to the difference of the chain. (Recall that we run the exact pattern matching algorithm for each rotation of $Q_i$ that can be the difference of a chain so that both checks can be performed in $\Oh(1)$ time). 

If these conditions are not satisfied, we create a new chain that contains $r$ only. Otherwise, we add $r$ to $\chain$ (i.e., increment the size of $\chain$). To finalize the update of the chains, we must delete all $k$-edit occurrences that become inactive: for each $i$ and for each chain of $k$-edit occurrences of $P_i$, we ``delete'' the first $k$-edit occurrence if it starts before $r-\ell_i-k+1$. To ``delete'' a $k$-edit occurrence, we simply update the endpoints of the first occurrence in the chain and the total number of occurrences in the chain.

\subsubsection{Analysis}
We summarize the results of this section:

\begin{theorem}\label{th:rpst}
Assume a read-only pattern $P$ of length $m$ and a streaming text $T$ of length $n$. There is a deterministic algorithm that finds the set $\Occ_k(P,T)$ using $\tOh(k^5)$ space and $\tOh(k^6)$ amortised time per character of the text $T$.
\end{theorem}
\begin{proof}
Let us first upper bound the space complexity of the algorithm. For each $i = 1, \ldots, z = \Oh(\log m)$, we store the set $\aOcc_k(P_i, T)$ as $\Oh(k^3)$ chains. For each chain, we launch the dynamic programming algorithm (\cref{fct:dynprog}), which takes $\tOh(k^2)$ space and~\cref{cor:greedy_long} that takes $\tOh(k^2)$ space. The pattern matching algorithms for the rotations of $Q_i$ take $\tOh(k)$ space in total. Finally, testing if a position of the text is the rightmost position of a $k$-edit occurrence of $P_i$ requires $\tOh(k^2)$ space. 

We now show the time bound. Updating the chains takes $\tOh(1)$ time. At any time, we run $\tOh(k^3)$ instances of~\cref{cor:greedy_long} that takes $\Oh(k^3)$ amortised time per character. To test each position $r$, we spend $\tOh(k \cdot k^3)$ time.
\end{proof}
 

%% file: spst.tex
\subsection{Streaming Algorithm}\label{sec:spst}
We now modify the algorithm for the semi-streaming model to develop a fully streaming algorithm. W.l.o.g., assume $k \le m$ and take $\delta = 1/n^c$ for $c$ large enough. 

\subsubsection{Preprocessing}
We define the prefixes $P_i$ and their periods $Q_i$ exactly in the same way as in Section~\ref{sec:rpst}. Recall that $\ell_i = |P_i|$ and $q_i = |Q_i|$. For every $i > 1$, we store the following information, where all sketches $\skq$ (\cref{cor:skq}) are parametrized by probability $\delta$, maximal length $\Delta_i$, the alphabet of $P$ and $T$, and a seed of $\Oh(\log^2 n \log(1/\delta))$ random bits:

\begin{enumerate}
\item $P[\ell_i]$ and the sketch $\skq_{k}(P[\ell_{i-1} \dd \ell_{i}])$;
\item For each of rotation $D$ of $Q_i$ that can be a difference of a chain of $k$-edit occurrences of $P_i$, the encoding $\qgr_{30k} (P[\ell_{i-1}+1 \dd \ell_i], D^\infty[1 \dd \Delta_i])$, where $\Delta_i = \ell_i - \ell_{i-1}+k$;
\item For each of rotation $D$ of $Q_i$ that can be a difference of a chain of $k$-edit occurrences of $P_i$, the sketches $\skq_{32k}(D)$ and $\skq_{32k}(D[1 \dd \Delta_i \pmod {q_i}])$.
\end{enumerate}

\subsubsection{Main Stage}
As in~\cref{sec:rpst}, for each $i$, we store $\aOcc_k^E(P_{i-1}, T)$ in $\Oh(k^3)$ chains. For each chain $\chain$, we store the leftmost position $lp$ in it, its size $|\chain|$, and the smallest edit distance, $\ed(\chain)$, from a suffix of $T[1\dd lp]$ to $P_{i-1}$. If $|\chain| \ge 2$, we also store its difference (defined by the shift of the rotation of $Q_{i-1}$). 

If $p^\ast$ is the first position added to $\chain$, then at the position $(p^\ast+1)$, we start running the streaming algorithm of~\cref{cor:skq}\ref{it:skq_construction} for computing the sketch $\skq_{k}(T[p^ast+1\dd p^ast+\Delta_{i}])$. 

Furthermore, consider the moment when we detect the second position in $\chain$ (if it exists) and hence the difference $D$ of the chain. Starting from this moment, for every newly added position $p \in \chain$, at the position $(p+1)$ we start computing the quasi-greedy encoding $\qgr_{32k} (T[p+1 \dd p+\Delta_i], D^\infty[1\dd\Delta_i])$ as follows: Assume that we have computed $\qgr_{32k} (T[p+1 \dd p+\ell\cdot q_i], D^\infty[1\dd\ell\cdot q_i])$. Suppose first that $(\ell+1)\cdot q_i \le \Delta_i$. While reading $T[p+\ell \cdot q_i+1\dd p+(\ell+1)\cdot q_i]$, we compute $\skq_{32k}(T[p+\ell \cdot q_i+1\dd p+(\ell+1)\cdot q_i])$ again via the algorithm of~\cref{cor:skq}\ref{it:skq_construction}. We then use this sketch and $\skq_{32k+1}(D)$ to compute $\qgr_{32k}(T[p+\ell \cdot q_i+1\dd p+(\ell+1)\cdot q_i], D)$ (\cref{cor:skq}\ref{it:skq_qgr}) and then $\qgr_{32k} (T[p+1 \dd p+(\ell+1)\cdot q_i], D^\infty[1\dd(\ell+1)\cdot q_i])$ (\cref{lm:greedy_concatenation}). If $(\ell+1)\cdot q_i > \Delta_i$, we use the sketches $\skq_{32k}(T[p+\ell\cdot q_i \dd p+\Delta_i])$ and $\skq_{32k}(D[1 \dd \Delta_i \pmod {q_i}])$, the rest is analogous. We continue running the algorithm until either the computation is completed or a new $k$-edit occurrence in the chain has been detected. In other words, in the end we compute the encoding for the rightmost position in the chain.

\paragraph{Detecting new $k$-edit occurrences of $P_i$.} 
We now explain how we modify the algorithm for detecting new $k$-edit occurrences of the prefixes $P_i$. The algorithm for Case~1 does not change.
Instead of the dynamic programming algorithm in Case~2, we use $\skq_{k}$ and use~\cref{cor:skq}\ref{it:skq_qgr} and then \cref{cor:qgr_to_ed} to compute the edit distance.  It remains to explain how we modify the algorithm for Case 3. 

We exploit~\cref{eq:ed} again. We can find all positions $r'$ in $\Oh(k^3)$ time. To test a position $r'$, it suffices to compute the edit distance between $P[\ell_{i-1}+1\dd \ell_i]$ and $T[r'+1\dd r]$. Let $\chain$ be the chain with difference $D$ that contains $r'$. If $|\chain| = 1$, we use the edit distance sketch, and otherwise the quasi-greedy encodings. By~\cref{lm:dist_to_Y}, $\ed(P[\ell_{i-1} \dd \ell_i], D^\infty [1 \dd \Delta_i]) \le 26k$ and hence by~\cref{cor:qgr_to_ed}, the edit distance between $P[\ell_{i-1}+1\dd \ell_i]$ and $T[r'+1\dd r]$ can be computed from the encodings $\G_P = \qgr_{30k} (\pi, \mu)$ and $\G_T = \qgr_{30k} (\mu, \tau)$ for $\pi = P[\ell_{i-1} \dd \ell_i]$, $\mu = D^\infty [1 \dd \Delta_i]$, and $\tau = T[r'+1 \dd r]$. $\G_P$ was computed during the preprocessing step and we store it explicitly. Hence, we only need to explain how to compute $\G_T$. Let $p$ be the rightmost position in the chain $\chain$. 

\begin{enumerate}
\item Recall that $T[r'+1\dd p] = D^j$ for $j = (p-r')/q_i$. We first compute $\qgr_{32k}(D, D)$ from $\skq_{32k}(D,D)$ via~\cref{cor:skq}\ref{it:skq_qgr}, and then $\G_{T,1} = \qgr_{32k}(T[r'+1\dd p], D^j) = \qgr_{32k}(D^j, D^j)$ in $\tOh(k^5)$ time and $\tOh(k^2)$ space as in Section~\ref{sec:rpst}. 

\item At the position $(p+1)$ we launched the streaming algorithm computing $\qgr_{32k} (T[p+1 \dd p+\Delta_i], D^\infty[1\dd\Delta_i])$ with a delay of $q_i$ characters. We have $p+\Delta_i \ge r' + \Delta_i \ge r$. Therefore, at a position $p+\ell \cdot q_i$, where $\ell = \lfloor (r-p+1) / q_i \rfloor$, the algorithm computes $\G_{T,2} = \qgr_{32k} (T[p+1 \dd p+\ell \cdot q_i], D^\infty[1\dd \ell \cdot q_i])$. 
The algorithm then continues to compute the sketch $\skq_{32k}(T[p+\ell \cdot q_i+1 \dd r])$. We use this sketch and the sketch $\skq_{32k}(D[1 \dd \Delta_i \pmod {q_i}])$ to compute $\G_{T,3} = \qgr_{30k} (T[p+\ell \cdot q_i+1 \dd r], D[1 \dd \Delta_i \pmod {q_i}])$ via~\cref{cor:skq}\ref{it:skq_qgr} and~\cref{obs:qgr_larger}. 

\item We finally concatenate $\G_{T,1}$ and $\G_{T,2}$ to obtain $\qgr_{32k} (T[r'+1 \dd r'+(\ell+j) \cdot q_i], D^\infty[1\dd (\ell+j) \cdot q_i])$, and then $\qgr_{32k} (T[r'+1 \dd r'+(\ell+j) \cdot q_i], D^\infty[1\dd (\ell+j) \cdot q_i])$ and $\G_{T,3}$ to obtain $\G_T$ via~\cref{lm:greedy_concatenation} (note that the difference of lengths of strings in $\G_{T,3}$ is bounded by~$2k$). 
\end{enumerate}

\paragraph{Updating the chains.}
When we detect a new $k$-edit occurrence of $P_i$, we must decide if it should be added to some existing chain or if we must create a new chain for this occurrence. We use the algorithm of Section~\ref{sec:rpst}, but replace the constant-space pattern matching algorithm with the streaming pattern matching algorithm~\cite{Porat:09} that for a rotation of $Q_i$ takes $\tOh(1)$ space and $\tOh(1)$ time per character and retrieves all its occurrences correctly with probability at least $1-\delta$. 

\subsubsection{Analysis}
We summarize the results of this section:

\begin{theorem}\label{th:spst}
Given a pattern $P$ of length $m$ and a text $T$ of length $n$. There is a streaming algorithm that finds the set $\Occ_k^E(P,T)$ using $\tOh(k^5)$ space and $\tOh(k^8)$ amortised time per character of the text $T$. The algorithm computes $\Occ_k^E(P,T)$ correctly with high probability. 
\end{theorem}
\begin{proof}
By~\cref{cor:structure-period}, for a fixed $i$ only $\Oh(k)$ rotations of $Q_i$ can be a difference of a chain of occurrences of $P_i$. The sketch $\skq_{30k+1}(\cdot)$ takes $\tOh(k^2)$ space (\cref{cor:skq}\ref{it:skq_construction}) and $\qgr_{30k}(\cdot,\cdot)$ $\tOh(k^2)$ space as well (\cref{cor:qgr_to_gr}). Therefore, the information computed during the preprocessing stage occupies $\tOh(k^3)$ space. During the main stage, we store $\tOh(k^3)$ chains. For each chain, we run the algorithm of~\cref{cor:skq}\ref{it:skq_construction} that takes $\tOh(k^2)$ space. The algorithm than computes the quasi-greedy encoding (\cref{cor:skq}\ref{it:skq_qgr}) takes $\tOh(k^2)$ space as well. In total, the information we store for the chains occupies $\tOh(k^5)$ space. When checking for new occurrences, we apply~\cref{lm:greedy_concatenation,cor:qgr_to_ed}, which require an overhead of $\tOh(k^2)$ space. Finally, the streaming pattern matching algorithms for the rotations of $Q_i$ that can be differences of occurrences of $P_i$ take $\tOh(k)$ space in total. The space bound follows. 

At any time, we run $\tOh(k^3)$ instances of the algorithm of ~\cref{cor:skq}\ref{it:skq_construction} that takes $\tOh(k)$ amortised time per character. In addition, for every character we run $\tOh(k^3)$ instances of the algorithms of~\cref{lm:greedy_concatenation,cor:qgr_to_ed} taking $\tOh(k^8)$ time in total. The pattern matching algorithms for the rotations of $Q_i$ take $\tOh(k)$ time per character.

Note that the only probabilistic procedures in the algorithm are streaming pattern matching~\cite{Porat:09} and that of~\cref{cor:skq}\ref{it:skq_qgr} that computes quasi-greedy encodings. These procedures are called $\poly(n,k) = \poly(n)$ times. By choosing the constant $c$ in $\delta = 1/n^c$ large enough, we can guarantee that the algorithm is correct with high probability by the union bound.
\end{proof}